\documentclass[a4paper,12pt, twoside,reqno]{amsart}
\usepackage{fullpage}
\usepackage[utf8]{inputenc}

\usepackage[foot]{amsaddr}

\usepackage{url}

\usepackage{dsfont}
\usepackage[all,2cell,cmtip]{xy}
\UseTwocells
\usepackage{amssymb}
\usepackage{amsmath,pdftexcmds,mathtools}
\usepackage[T1]{fontenc}
\usepackage{xfrac}

\usepackage{epigraph}

\DeclareFontFamily{U}{cbgreek}{}
\DeclareFontShape{U}{cbgreek}{m}{n}{
        <-6>    grmn0500
        <6-7>   grmn0600
        <7-8>   grmn0700
        <8-9>   grmn0800
        <9-10>  grmn0900
        <10-12> grmn1000
        <12-17> grmn1200
        <17->   grmn1728
      }{}
\DeclareFontShape{U}{cbgreek}{bx}{n}{
        <-6>    grxn0500
        <6-7>   grxn0600
        <7-8>   grxn0700
        <8-9>   grxn0800
        <9-10>  grxn0900
        <10-12> grxn1000
        <12-17> grxn1200
        <17->   grxn1728
      }{}

\makeatletter
\newcommand{\normalorbold}{%
  \ifnum\pdf@strcmp{\math@version}{bold}=\z@ bx\else m\fi
}
\makeatother

\usepackage{graphicx}
\usepackage{tikz}
\usepackage{array}
\usepackage{todonotes}

\allowdisplaybreaks[3]

\usepackage{epigraph}

\usetikzlibrary{decorations.pathmorphing}

\usepackage[outline]{contour}
\contourlength{2pt}

\usepackage{bigdelim}
\usepackage{enumerate}
\usepackage{mathrsfs}
\usepackage{color}
\usepackage{amsthm}
\usepackage{float}
\usepackage{accents}
\usepackage{stmaryrd}
\usepackage{physics}
\usepackage{tensor}

\usepackage{txfonts}

\usepackage{multirow}
\usepackage{multicol}
\theoremstyle{plain}
\newtheorem{theorem}{Theorem}[section]
\newtheorem{proposition}[theorem]{Proposition}

\newtheorem{lemma}[theorem]{Lemma}
\newtheorem*{theorem*}{Theorem}

\theoremstyle{definition}
\newtheorem{definition}[theorem]{Definition}
\newtheorem{example}[theorem]{Example}

\begin{document}

%%%%%%%%%%%%%%%%%%%%%%%%%%%%%%%%%%%%%%%%%%%%%%%%%%%%%%%%%%%%%%%%%%%%%%%%%%%

%%%%%% NOVOS COMANDOS

%%%%%%%%%%%%%%%%%%%%%%%%%%%%%%%%%%%%%%%%%%%%%%%%%%%%%%%%%%%%%%%%%%%%%%%%%%%

%%%%%% LETRAS

%%% A

\newcommand{\bbA}{\mathbb{A}}
\newcommand{\overbbA}{\overline{\bbA}}
\newcommand{\calA}{\mathcal{A}}
\newcommand{\frakA}{\mathfrak{A}}
\newcommand{\rma}{\mathrm{a}}
\newcommand{\tila}{\til{a}}
\newcommand{\tilrma}{\til{\rma}}
\newcommand{\rmA}{\mathrm{A}}
\newcommand{\tilrmA}{\til{\rmA}}
\newcommand{\tilalpha}{\til{\alpha}}
\newcommand{\rmaestrela}{\rma\upast}
\newcommand{\overa}{\overline{a}}
\newcommand{\overrmA}{\overline{\rmA}}
\newcommand{\dwnbbA}{\dwn{\bbA}}
\newcommand{\alphabfn}{\alpha\dwnbfn}
\newcommand{\dwnalpha}{\dwn{\alpha}}
\newcommand{\AHPnkk}{\rmA(k\dwnum,k\dwndois)}
\newcommand{\fraka}{\mathfrak{a}}
\newcommand{\dwna}{\dwn{a}}
\newcommand{\sfA}{\mathsf{A}}
\newcommand{\uprma}{\up{\rma}}
\newcommand{\upa}{\up{a}}
\newcommand{\dwnrmA}{\dwn{\rmA}}
\newcommand{\spaceavermA}{\spaceave{\rmA}}
\newcommand{\calAff}{\calA\dwn{\!f\!\!f}}
\newcommand{\dwnrmAum}{\dwn{\rmA\dwnum}}
\newcommand{\dwnrmAn}{\dwn{\rmA\dwnn}}
\newcommand{\dwnumrmA}{\dwn{1,\rmA}}
\newcommand{\upalpha}{^\alpha}
\newcommand{\uptilalpha}{\up{\tilalpha}}
\newcommand{\primealpha}{\alpha\prime}
\newcommand{\overprimealpha}{\overline{\primealpha}}
\newcommand{\upprimealpha}{\up{\primealpha}}
\newcommand{\upoverprimealpha}{\up{\overprimealpha}}
\newcommand{\overalpha}{\overline{\alpha}}
\newcommand{\upoveralpha}{\up{\overalpha}}
\newcommand{\bfA}{\mathbf{A}}

%%% B

\newcommand{\bbB}{\mathbb{B}}
\newcommand{\calB}{\mathcal{B}}
\newcommand{\rmB}{\mathrm{B}}
\newcommand{\tilb}{\til{b}}
\newcommand{\tilrmB}{\til{\rmB}}
\newcommand{\tilbeta}{\til{\beta}}
\newcommand{\bparm}{b^{(m)}}
\newcommand{\overb}{\overline{b}}
\newcommand{\sfB}{\mathsf{B}}
\newcommand{\dwnbeta}{\dwn{\beta}}
\newcommand{\betan}{\beta\dwn{n}}
\newcommand{\betabfn}{\beta\dwnbfn}
\newcommand{\BHPnkk}{\rmB(k\dwnum,k\dwndois)}
\newcommand{\frakB}{\mathfrak{B}}
\newcommand{\dwnb}{\dwn{b}}
\newcommand{\calBK}{\calB(\rmK)}
\newcommand{\spaceavermB}{\spaceave{\rmB}}
\newcommand{\dwnbbB}{\dwn{\bbB}}
\newcommand{\overrmB}{\overline{\rmB}}
\newcommand{\dwnrmB}{\dwn{\rmB}}
\newcommand{\hatrmB}{\hat{\rmB}}
\newcommand{\checkrmB}{\check{\rmB}}

%%% C

\newcommand{\C}{\mathbb{C}}
\newcommand{\CmM}{\C\upm\dwnsfM}
\newcommand{\rmC}{\mathrm{C}}
\newcommand{\cparm}{c^{(m)}}
\newcommand{\overc}{\overline{c}}
\newcommand{\covHPn}{\rmC\up{(n)}_{\rmH,\rmP}}
\newcommand{\covHPncom}[2]{\rmC\up{(n)}_{#1,#2}}
\newcommand{\bfcovHPn}{\calC\up{(n)}_{\rmH,\rmP}}
\newcommand{\bfcovHPncom}[2]{\calC\up{(n)}_{#1,#2}}
\newcommand{\bfcovHPinfty}{\calC\up{(\infty)}_{\rmH,\rmP}}
\newcommand{\bfcovHPinftycom}[2]{\calC\up{(\infty)}_{#1,#2}}
\newcommand{\bfcovHPnkk}{\calC\up{(n)}_{\rmH,\rmP}(k\dwnum,k\dwndois)}
\newcommand{\tilcovHPn}{\til{\rmC}\up{(n)}_{\rmH,\rmP}}
\newcommand{\covHPnM}{\rmC\up{(n)}_{\tilrmH,\tilrmP}}
\newcommand{\bfC}{\mathbf{C}}
\newcommand{\sfC}{\mathsf{C}}
\newcommand{\calC}{\mathcal{C}}
\newcommand{\dwnc}{\dwn{c}}
\newcommand{\rmc}{\mathrm{c}}
\newcommand{\dwnrmC}{\dwn{\rmC}}
\newcommand{\tilrmC}{\til{\rmC}}
\newcommand{\uprmc}{\up{\rmc}}

%%% D

\newcommand{\bbD}{\mathbb{D}}
\newcommand{\overd}{\overline{d}}
\newcommand{\bfpartial}{\mathbf{\partial}}
\newcommand{\rmd}{\mathrm{d}}
\newcommand{\uprmd}{\up{\rmd}}
\newcommand{\upmeiormd}{\up{\frac{\rmd}{2}}}
\newcommand{\rmD}{\mathrm{D}}
\newcommand{\tilrmD}{\til{\mathrm{D}}}
\newcommand{\rmDbfM}{\rmD\dwn{\bfM}}
\newcommand{\frakd}{\mathfrak{d}}
\newcommand{\tilbbD}{\til{\bbD}}
\newcommand{\calD}{\mathcal{D}}
\newcommand{\frakD}{\mathfrak{D}}
\newcommand{\overdelta}{\overline{\delta}}
\newcommand{\dwndelta}{\dwn{\delta}}
\newcommand{\dwnrmd}{\dwn{\rmd}}
\newcommand{\spaceaveDelta}{\spaceave{\Delta}}
\newcommand{\tilfrakD}{\til{\frakD}}
\newcommand{\dwndoisvirgn}{\dwn{2,n}}
\newcommand{\dwnpardois}{\dwn{(2)}}
\newcommand{\upvirgrmd}{\up{\,\rmd}}

%%% E

\newcommand{\frake}{\mathfrak{e}}
\newcommand{\ee}{\mathrm{e}}
\newcommand{\rmE}{\mathrm{E}}
\newcommand{\bfe}{\mathbf{e}}
\newcommand{\bfedwnrmxj}{\bfe\dwn{\rmx\dwnj}}
\newcommand{\bfedwnrmxi}{\bfe\dwn{\rmx\dwni}}
\newcommand{\sfE}{\mathsf{E}}
\newcommand{\overbfe}{\overline{\bfe}}
\newcommand{\frakE}{\mathfrak{E}}
\newcommand{\bbE}{\mathbb{E}}
\newcommand{\calE}{\mathcal{E}}
\newcommand{\dwnvarepsilon}{\dwn{\varepsilon}}
\newcommand{\dwnfrakEll}{\dwn{\frakE\dwnell}}
\newcommand{\dwncalEfrakEll}{\dwn{\calE(\frakE\dwnell)}}
\newcommand{\overcalE}{\overline{\calE}}
\newcommand{\bfE}{\mathbf{E}}
\newcommand{\rme}{\mathrm{e}} % NAO SUBSTITUI \ee -- CONTEXTOS DIFERENTES
\newcommand{\dwnrme}{\dwn{\rme}}
\newcommand{\tilvarepsilon}{\til{\varepsilon}}
\newcommand{\dwntilvarepsilon}{\dwn{\tilvarepsilon}}
\newcommand{\tilfrake}{\til{\frake}}
\newcommand{\brevefrake}{\breve{\frake}}
\newcommand{\hatfrake}{\hat{\frake}}
\newcommand{\tilepsilon}{\til{\epsilon}}
\newcommand{\checkbfe}{\check{\bfe}}

%%% F

\newcommand{\frakf}{\mathfrak{f}}
\newcommand{\rmf}{\mathrm{f}}
\newcommand{\rmF}{\mathrm{F}}
\newcommand{\tilf}{\til{f}}
\newcommand{\bbF}{\mathbb{F}}
\newcommand{\bff}{\mathbf{f}}
\newcommand{\overbff}{\overline{\bff}}
\newcommand{\dwnf}{\dwn{f}}
\newcommand{\dwnrmf}{\dwn{\rmf}}
\newcommand{\upPhi}{\up{\Phi}}
\newcommand{\tilPhi}{\til{\Phi}}
\newcommand{\PhidwnbfvLambda}{\Phi\dwn{\bfv,\,\Lambda}}
\newcommand{\tilvarphi}{\til{\varphi}}
\newcommand{\frakF}{\mathfrak{F}}
\newcommand{\ringrmf}{\mathring{\rmf}}
\newcommand{\absrmf}{\absoluto{\rmf}}
\newcommand{\dwnPhi}{\dwn{\Phi}}
\newcommand{\dwnPhibfl}{\dwn{\Phi,\bfl}}
\newcommand{\dwnPhibflbfn}{\dwn{\Phi,\bfl,\bfn}}
\newcommand{\dwnPhiumbfl}{\dwn{\Phi,\um,\bfl}}
\newcommand{\dwnPhibeta}{\dwn{\Phi\!,\beta}}
\newcommand{\upPhibeta}{\up{\Phi\!,\beta}}
\newcommand{\overrmf}{\overline{\rmf}}
\newcommand{\tilrmf}{\til{\rmf}}
\newcommand{\rmfdwnVert}{\rmf\dwnVert}
\newcommand{\hatrmf}{\hat{\rmf}}
\newcommand{\calF}{\mathcal{F}}
\newcommand{\checkrmf}{\check{\rmf}}
\newcommand{\sff}{\mathsf{f}}

%%% G

\newcommand{\gupmais}{g^+}
\newcommand{\gupmenos}{g^-}
\newcommand{\circGamma}{\Gamma^\circ}
\newcommand{\oGamma}{\overline{\Gamma}}
\newcommand{\rmG}{\mathrm{G}}
\newcommand{\rmGM}{\mathrm{G}\dwnsfM}
\newcommand{\tilg}{\til{g}}
\newcommand{\tilGamma}{\til{\Gamma}}
\newcommand{\Gammaoplus}{\Gamma_\oplus}
\newcommand{\GammaM}{\Gamma\dwnsfM}
\newcommand{\GammaS}{\Gamma\dwnsfS}
\newcommand{\GammaL}{\Gamma\dwnrmL}
\newcommand{\Gammainfty}{\Gamma\dwn{\infty}}
\newcommand{\frakg}{\mathfrak{g}}
\newcommand{\frakG}{\mathfrak{G}}
\newcommand{\rmg}{\mathrm{g}}
\newcommand{\dwnrmg}{\dwn{\rmg}}
\newcommand{\dwnrmG}{\dwn{\rmG}}
\newcommand{\tilrmg}{\til{\rmg}}
\newcommand{\overrmg}{\overline{\rmg}}
\newcommand{\frakGdote}{\frakG\!\!\cdot\!\!\frake}
\newcommand{\frakGdotte}{\frakG\!\cdot\!\frake}
\newcommand{\frakGdotbrevefrake}{\frakG\!\!\cdot\!\!\brevefrake}
\newcommand{\frakGdottbrevefrake}{\frakG\!\cdot\!\brevefrake}
\newcommand{\frakGdotf}{\frakG\!\!\cdot\!\!\frakf}
\newcommand{\frakGdottf}{\frakG\!\cdot\!\frakf}
\newcommand{\frakGdotg}{\frakG\!\!\cdot\!\!\frakg}
\newcommand{\frakGdottg}{\frakG\!\cdot\!\frakg}
\newcommand{\frakGdoth}{\frakG\!\!\cdot\!\!\frakh}
\newcommand{\frakGdotth}{\frakG\!\cdot\!\frakh}
\newcommand{\dwngammarmg}{\dwn{\gamma\rmg}}
\newcommand{\dwnrmgrmG}{\dwn{\rmg(\rmG)}}
\newcommand{\dwnhatrmgrmG}{\dwn{\hatrmg(\rmG)}}
\newcommand{\dwnparrmG}{\dwn{(\rmG)}}
\newcommand{\hatrmg}{\hat{\rmg}}
\newcommand{\dwnhatrmg}{\dwn{\hatrmg}}
\newcommand{\dwnfraklhatrmgbeta}{\dwn{\frakl,\hatrmg,\beta}}
\newcommand{\barrmg}{\bar{\rmg}}
\newcommand{\dwnbarrmg}{\dwn{\barrmg}}
\newcommand{\dwnfraklbarrmgbeta}{\dwn{\frakl,\barrmg,\beta}}
\newcommand{\rmgprime}{\rmg\prime}
\newcommand{\dwnrmgprime}{\dwn{\rmgprime}}
\newcommand{\dwnfraklrmgprimebeta}{\dwn{\frakl,\rmgprime,\beta}}
\newcommand{\hatGamma}{\hat{\Gamma}}
\newcommand{\bfG}{\mathbf{G}}

%%% H

\newcommand{\bfH}{\mathbf{H}}
\newcommand{\calH}{\mathcal{H}}
\newcommand{\ocalH}{\overline{\mathcal{H}}}
\newcommand{\frakh}{\mathfrak{h}}
\newcommand{\frakH}{\mathfrak{H}}
\newcommand{\rmH}{\mathrm{H}}
\newcommand{\overrmH}{\overline{\rmH}}
\newcommand{\tilh}{\til{h}}
\newcommand{\rmHoplus}{\rmH\dwn{\oplus}}
\newcommand{\tilrmH}{\til{\rmH}}
\newcommand{\rmh}{\mathrm{h}}
\newcommand{\rmhbfn}{\rmh\dwn{\bfn}}
\newcommand{\rmhparn}{\rmh\paren{n}}
\newcommand{\rmHP}{\rmH\dwnrmP}
\newcommand{\rmHL}{\rmH\dwnrmL}
\newcommand{\rmHbfM}{\rmH\dwn{\bfM}}
\newcommand{\rmHo}{\rmH\dwn{\rmo}}
\newcommand{\dwnrmh}{\dwn{\rmh}}
\newcommand{\tilrmh}{\til{\rmh}}
\newcommand{\overrmh}{\overline{\rmh}}
\newcommand{\hatrmh}{\hat{\rmh}}
\newcommand{\checkrmh}{\check{\rmh}}
\newcommand{\tilbfH}{\til{\bfH}}
\newcommand{\checktilrmh}{\check{\tilrmh}}
\newcommand{\dwnrrmHbeta}{\dwn{\rmH,\beta}}

%%% I

\newcommand{\ii}{\mathrm{i}}
\newcommand{\rmI}{\mathrm{I}}
\newcommand{\tili}{\til{i}}
\newcommand{\dwni}{\dwn{i}}
\newcommand{\dwnrmi}{\dwn{\rmi}}
\newcommand{\fraki}{\mathfrak{i}}
\newcommand{\rmi}{\mathrm{i}}
\newcommand{\dwninvtra}{\dwn{\um}}
\newcommand{\frakI}{\mathfrak{I}}
\newcommand{\upi}{\up{i}}

%%% J
\newcommand{\rmj}{\mathrm{j}}
\newcommand{\rmJ}{\mathrm{J}}
\newcommand{\dwnj}{\dwn{j}}
\newcommand{\frakj}{\mathfrak{j}}
\newcommand{\dwniota}{\dwn{\iota}}
\newcommand{\upj}{\up{j}}
\newcommand{\dwnji}{\dwn{j,i}}
\newcommand{\tilj}{\til{j}}

%%% K

\newcommand{\rmk}{\mathrm{k}}
\newcommand{\bfk}{\mathbf{k}}
\newcommand{\calK}{\mathcal{K}}
\newcommand{\dwnk}{\dwn{k}}
\newcommand{\tilrmK}{\til{\mathrm{K}}}
\newcommand{\tilk}{\til{k}}
\newcommand{\tilkappa}{\til{\kappa}}
\newcommand{\rmK}{\mathrm{K}}
\newcommand{\rmKparbaixon}{\rmK_{(n)}}
\newcommand{\rmKparbaixom}{\rmK_{(m)}}
\newcommand{\rmKparbaixocom}[1]{\rmK_{(#1)}}
\newcommand{\rmKparbaixozero}{\rmK_{(0)}}
\newcommand{\rmKestrela}{\rmK\upast}
\newcommand{\grassK}{\wedge \rmK}
\newcommand{\grassKn}{\wedge_n \! \rmK}
\newcommand{\grassKum}{\wedge_1 \rmK}
\newcommand{\grassKm}{\wedge_m \! \rmK}
\newcommand{\grassKestrela}{\wedge^{\! \! *} \rmK}
\newcommand{\grassKestrelacom}[1]{\wedge^{\! \! *} #1}
\newcommand{\grassKestrelaPK}{\wedge^{\! \! *} \rmP\rmK}
\newcommand{\grassKestrelaoplus}{\wedge^{\! \! *} \rmKoplus}
\newcommand{\grassKestrelaparbaixon}{\wedge^{\! \! *} \rmKparbaixon}
\newcommand{\grassKestrelaparbaixonm}{\wedge^{\! \! *} \rmK_{(nm)}}
\newcommand{\grassKestrelan}{\wedge^{\! \! *}_n \rmK}
\newcommand{\grassKestrelafuncao}[1]{\wedge^{\! \! *}_{#1} \rmK}
\newcommand{\grassKestrelaoplusfuncao}[1]{\wedge^{\! \! *}_{#1} \rmKoplus}
\newcommand{\grassKestrelaum}{\wedge^{\! \! *}_1 \rmK}
\newcommand{\grassKestrelam}{\wedge^{\! \! *}_m \rmK}
\newcommand{\grassKestrelamparbaixon}{\wedge^{\! \! *}_m \rmKparbaixon} 
\newcommand{\grassKestrelazero}{\wedge^{\! \! *}_0 \rmK}
\newcommand{\rmKoplus}{\rmK_\oplus}
\newcommand{\rmKopluscom}[1]{\rmK_\oplus^{(#1)}}
\newcommand{\fockK}{\rmF(\rmK)}
\newcommand{\upk}{^k}
\newcommand{\uppark}{\up{(k)}}
\newcommand{\dwnpark}{\dwn{\paren{k}}}
\newcommand{\rmKparbaixoPcom}[2]{\rmK^\rmP_{(#1 #2)}}
\newcommand{\kupi}{k^{(i)}}
\newcommand{\kupicom}[1]{k^{(#1)}}
\newcommand{\rmKM}{\rmK\dwnsfM}
\newcommand{\rmKS}{\rmK\dwnsfS}
\newcommand{\rmKL}{\rmK\dwnrmL}
\newcommand{\rmKinfty}{\rmK\dwn{\infty}}
\newcommand{\frakK}{\mathfrak{K}}
\newcommand{\frakk}{\mathfrak{k}}
\newcommand{\updoisk}{\up{2k}}
\newcommand{\bbK}{\mathbb{K}}
\newcommand{\tilxi}{\til{\xi}}
\newcommand{\dwnkmaisum}{\dwn{k+1}}
\newcommand{\dwnktooinfty}{\dwn{k\too\infty}}
\newcommand{\dwnksfs}{\dwn{k\sfs}}
\newcommand{\dwnksft}{\dwn{k\sft}}
\newcommand{\hatrmK}{\hat{\rmK}}
\newcommand{\dwnmenosk}{\dwn{-k}}
\newcommand{\checkrmK}{\check{\rmK}}
\newcommand{\checkcalK}{\check{\calK}}
\newcommand{\dwnrmk}{\dwn{\rmk}}
\newcommand{\dwnparrmk}{\dwnpar{\rmk}}
\newcommand{\dwnksfsalpha}{\dwn{(k,\sfs,\alpha)}}
\newcommand{\dwnqsftsfo}{\dwn{(q,\sft,\sfo)}} 

%%% L

\newcommand{\rmL}{\mathrm{L}}
\newcommand{\dwnrmL}{\dwn{\rmL}}
\newcommand{\dwnrmLrmf}{\dwn{\rmL\dwnrmf}}
\newcommand{\dwnrmLrmi}{\dwn{\rmL\dwnrmi}}
\newcommand{\dwnl}{\dwn{l}}
\newcommand{\dwnparl}{\dwn{\paren{l}}}
\newcommand{\LambdaL}{\Lambda\dwnrmL}
\newcommand{\dwnLambdaL}{\dwn{\LambdaL}}
\newcommand{\Lambdal}{\Lambda\dwnl}
\newcommand{\Lambdainfty}{\Lambda\dwninfty}
\newcommand{\elldois}{\ell\updois}
\newcommand{\overlambda}{\overline{\lambda}}
\newcommand{\dwnlambda}{\dwn{\lambda}}
\newcommand{\bfl}{\mathbf{l}}
\newcommand{\dwnLambda}{\dwn{\Lambda}}
\newcommand{\dwnbfl}{\dwn{\bfl}}
\newcommand{\dwnell}{\dwn{\ell}}
\newcommand{\dwnrmLbfl}{\dwn{\rmL,\bfl}}
\newcommand{\dwnbflrmA}{\dwn{\bfl,\rmA}}
\newcommand{\dwnbfln}{\dwn{\bfl(n)}}
\newcommand{\tilLambda}{\til{\Lambda}}
\newcommand{\dwnbflnx}{\dwn{\bfl,n,x}}
\newcommand{\bflcdotsfR}{\bfl\cdot\sfR}
\newcommand{\dwnbflbfn}{\dwn{\bfl,\bfn}}
\newcommand{\dwnrmLbflbfn}{\dwn{\rmL,\bfl\dwnbfn}}
\newcommand{\bflbfncdotsfR}{\bfl\dwnbfn\!\!\cdot\!\sfR}
\newcommand{\dwnellbfn}{\dwn{\ell\dwnbfn}}
\newcommand{\dwnrmLtooinfty}{\dwn{\rmL\too\infty}}
\newcommand{\dwnrmLinN}{\dwn{\rmL\in\N}}
\newcommand{\dwnLambdaPhi}{\dwn{\Lambda,\Phi}}
\newcommand{\dwnLambdaPhibeta}{\dwn{\Lambda,\Phi,\beta}}
\newcommand{\dwnLambdarmLPhi}{\dwn{\Lambda\dwnrmL,\Phi}}
\newcommand{\dwnLambdarmLPhibeta}{\dwn{\Lambda\dwnrmL,\Phi,\beta}}
\newcommand{\dwnumbfl}{\dwn{\um,\bfl}}
\newcommand{\dwnLambdarmL}{\dwn{\Lambda\dwnrmL}} % TORNA \dwnLambdaL VESTIGIAL
\newcommand{\dwntilLambda}{\dwn{\tilLambda}}
\newcommand{\dwnLambdaplusx}{\dwn{\Lambda+x}}
\newcommand{\upmaisLambda}{\up{+\Lambda}}
\newcommand{\upmenosLambda}{\up{-\Lambda}}
\newcommand{\dwnumbfldwnbfn}{\dwn{\um,\bfl\dwnbfn}}
\newcommand{\bfLambda}{\mathbf{\Lambda}}
\newcommand{\dwnumbfldwnbfndwnbfm}{\dwn{\um,\bfl\dwnbfndwnbfm}}
\newcommand{\tilbfLambda}{\til{\bfLambda}}
\newcommand{\ummenoslambda}{(1-\lambda)}
\newcommand{\Lambdalinha}{\Lambda\upprime}
\newcommand{\dwnLambdalinha}{\dwn{\Lambdalinha}}
\newcommand{\bbL}{\mathbb{L}}
\newcommand{\bbLdiamond}{\bbL\dwn{\,\,\cbdiamond}}
\newcommand{\frakl}{\mathfrak{l}}
\newcommand{\dwnfrakl}{\dwn{\frakl}}
\newcommand{\dwnLambdafrakl}{\dwn{\Lambda,\frakl}}
\newcommand{\dwnLambdarmLfraklx}{\dwn{\Lambda\dwnrmL,\frakl(x)}} 
\newcommand{\dwnLambdafraklbeta}{\dwn{\Lambda,\frakl,\beta}}
\newcommand{\dwnLambdarmLfraklbeta}{\dwn{\Lambda\dwnrmL,\frakl,\beta}}
\newcommand{\dwnLambdafraklmenos}{\dwn{\Lambda,\frakl\dwnmenos}}
\newcommand{\dwnLambdafraklmais}{\dwn{\Lambda,\frakl\dwnmais}}
\newcommand{\dwnfraklrmo}{\dwn{\frakl,\rmo}}
\newcommand{\dwnfraklbeta}{\dwn{\frakl,\beta}}
\newcommand{\dwnfraklbflbfn}{\dwn{\frakl,\bfl,\bfn}}
\newcommand{\dwnfraklxbflbfn}{\dwn{\frakl(x),\bfl,\bfn}}
\newcommand{\dwnfraklbfl}{\dwn{\frakl,\bfl}}
\newcommand{\dwnfraklxbfl}{\dwn{\frakl(x),\bfl}}
\newcommand{\dwnfraklx}{\dwn{\frakl(x)}}
\newcommand{\dwnfraklmenos}{\dwn{\frakl\dwnmenos}}
\newcommand{\dwnfraklmais}{\dwn{\frakl\dwnmais}}
\newcommand{\dwnfraklpm}{\dwn{\frakl\dwnpm}}
\newcommand{\dwnfraklmp}{\dwn{\frakl\dwnmp}}
\newcommand{\dwnfraklmenosx}{\dwn{\frakl\dwnmenos(x)}}
\newcommand{\dwnfraklmaisx}{\dwn{\frakl\dwnmais(x)}}
\newcommand{\dwnfraklmenosbfl}{\dwn{\frakl\dwnmenos,\bfl}}
\newcommand{\dwnfraklmaisbfl}{\dwn{\frakl\dwnmais,\bfl}}
\newcommand{\dwnfraklmenosxbfl}{\dwn{\frakl\dwnmenos(x),\bfl}}
\newcommand{\dwnfraklmaisxbfl}{\dwn{\frakl\dwnmais(x),\bfl}}
\newcommand{\dwnfraklmenosbeta}{\dwn{\frakl\dwnmenos,\beta}}
\newcommand{\dwnfraklmaisbeta}{\dwn{\frakl\dwnmais,\beta}}
\newcommand{\dwnfraklrmg}{\dwn{\frakl,\rmg}}
\newcommand{\dwnfraklrmgbeta}{\dwn{\frakl,\rmg,\beta}}
\newcommand{\dwnfraklrmhbeta}{\dwn{\frakl,\rmh,\beta}}
\newcommand{\dwnLambdafraklrmg}{\dwn{\Lambda,\frakl,\rmg}}
\newcommand{\stackrmLtooinfty}{\stackrel{\rmL\too\infty}{\tooo}}
\newcommand{\dwnLambdafraklrmgbeta}{\dwn{\Lambda,\frakl,\rmg,\beta}}
\newcommand{\dwnfraklrmgrmGbeta}{\dwn{\frakl,\rmg(\rmG),\beta}}
\newcommand{\dwnLambdarmhrmv}{\dwn{\Lambda,\rmh,\rmv}}
\newcommand{\hatLambda}{\hat{\Lambda}}
\newcommand{\upparbfl}{\uppar{\bfl}}
\newcommand{\upparrmL}{\uppar{\rmL}}
\newcommand{\upparrmLupprime}{\uppar{\rmL\upprime}}
\newcommand{\rmLupprime}{\rmL\upprime}
\newcommand{\dwnrmLrmLupprime}{\dwn{\rmL,\rmLupprime}}
\newcommand{\tilrmL}{\til{\rmL}}
\newcommand{\uppartilrmL}{\uppar{\tilrmL}}
\newcommand{\dwnrmLtilrmL}{\dwn{\rmL,\tilrmL}}
\newcommand{\dwnrmLrmo}{\dwn{\rmL\dwnrmo}}
\newcommand{\dwndoisrmLrmo}{\dwn{2\rmL\dwnrmo}}
\newcommand{\meiormL}{\frac{\rmL}{2}}
\newcommand{\dwndoisrmL}{\dwn{2\rmL}}
\newcommand{\dwnparrmL}{\dwnpar{\rmL}}
\newcommand{\dwnpardoisrmLmaisum}{\dwnpar{\doisrmLmaisum}}
\newcommand{\uppardoisrmLmaisum}{\uppar{2\rmL+1}}
\newcommand{\dwndoisrmLmaisum}{\dwn{2\rmL+1}}
\newcommand{\dwnmeiormL}{\dwn{\rmL/2}}
\newcommand{\dwnmeiodoisrmLmaisum}{\dwn{(2\rmL+1)/2}}
\newcommand{\dwnrmLk}{\dwn{\rmL,k}}
\newcommand{\dwnparrmLk}{\dwnpar{\rmL,k}}
\newcommand{\doisrmLmaisum}{2\rmL+1}
\newcommand{\dwndoisrmLmaisumrmk}{\dwn{\doisrmLmaisum,\rmk}}
\newcommand{\dwnpardoisrmLmaisumrmk}{\dwnpar{\doisrmLmaisum,\rmk}}
\newcommand{\rmLtooinfty}{\rmL\too\infty}
\newcommand{\dwnrmLrmk}{\dwn{\rmL,\rmk}}
\newcommand{\dwnparrmLrmk}{\dwnpar{\rmL,\rmk}}
\newcommand{\dwnLambdaPsi}{\dwn{\Lambda,\Psi}}
\newcommand{\umsobrermLupmeiormd}{\umsobre{\rmL\upmeiormd}}
\newcommand{\umsobrermLuprmd}{\umsobre{\rmL\uprmd}}
\newcommand{\umsobreabsolutoLambdarmLupmeio}{\frac{1}{\absoluto{\Lambda\dwnrmL}\upmeio}}
\newcommand{\umsobreabsolutoLambdarmL}{\frac{1}{\absoluto{\Lambda\dwnrmL}}}

%%% M

\newcommand{\frakm}{\mathfrak{m}}
\newcommand{\frakM}{\mathfrak{M}}
\newcommand{\rmM}{\mathrm{M}}
\newcommand{\upparm}{^{(m)}}
\newcommand{\uprmM}{\up{\rmM}}
\newcommand{\bbM}{\mathbb{M}}
\newcommand{\bfM}{\mathbf{M}}
\newcommand{\bfMraise}{\raisebox{0.01cm}{$\bfM$}}
\newcommand{\bfMupmenosraise}{\raisebox{0.01cm}{$\bfM\upmenos$}}
\newcommand{\sfM}{\mathsf{M}}
\newcommand{\upm}{\up{m}}
\newcommand{\dwnsfM}{\dwn{\sfM}}
\newcommand{\dwnm}{\dwn{m}}
\newcommand{\dwnmu}{\dwn{\mu}}
\newcommand{\dwnfrakM}{\dwn{\frakM}}
\newcommand{\tilrmM}{\til{\rmM}}
\newcommand{\MMC}{\mathrm{MMC}}
\newcommand{\bfm}{\mathbf{m}}
\newcommand{\dwnbfm}{\dwn{\bfm}}
\newcommand{\dwnmmaisum}{\dwn{m+1}}
\newcommand{\dwnmj}{\dwn{m,j}}
\newcommand{\tilm}{\til{m}}
\newcommand{\dwnfrakm}{\dwn{\frakm}}
\newcommand{\dwnfrakmbeta}{\dwn{\frakm,\beta}}
\newcommand{\dwnfrakmc}{\dwn{\frakm,c}}

%%% N

\newcommand{\N}{\mathbb{N}}
\newcommand{\calN}{\mathcal{N}}
\newcommand{\rmN}{\mathrm{N}}
\newcommand{\rmn}{\mathrm{n}}
\newcommand{\upnn}{^n}
\newcommand{\updoisn}{\up{2n}}
\newcommand{\uprmN}{^\rmN}
\newcommand{\upparn}{^{(n)}}
\newcommand{\upparnbeta}{^{(n\dwnbeta)}}
\newcommand{\downn}{_n}
\newcommand{\dwnn}{\dwn{n}}
\newcommand{\nbeta}{n\dwnbeta}
\newcommand{\bfn}{\mathbf{n}}
\newcommand{\dwnbfn}{\dwn{\bfn}}
\newcommand{\dwnnu}{\dwn{\nu}}
\newcommand{\upnatural}{\up{\natural}}
\newcommand{\dwndoisn}{\dwn{2n}}
\newcommand{\ninN}{n\in\N}
\newcommand{\dwnninN}{\dwn{\ninN}}
\newcommand{\upparnast}{^{(n)*}}
\newcommand{\rmNbfl}{\rmN\bfl}
\newcommand{\dwnntooinfty}{\dwn{n\too\infty}}
\newcommand{\bfN}{\mathbf{N}}
\newcommand{\updoisuprmN}{\up{2\uprmN}}
\newcommand{\dwnbfndwnbfm}{\dwn{\bfn\dwnbfm}}
\newcommand{\dwnnmaisum}{\dwn{n+1}}
\newcommand{\dwnnVert}{\dwn{n\Vert}}
\newcommand{\dwnndwni}{\dwn{n\dwni}}
\newcommand{\dwnndwnimaisum}{\dwn{n\dwn{i+1}}}
\newcommand{\dwnrmN}{\dwn{\rmN}}
\newcommand{\dwnrmNtooinfty}{\dwn{\rmN\too\infty}}
\newcommand{\stackntooinfty}{\stackrel{n\too\infty}{\tooo}}
\newcommand{\doisrmN}{2\rmN}
\newcommand{\dwndoisrmN}{\dwn{\doisrmN}}
\newcommand{\updoisrmN}{\up{\doisrmN}}
\newcommand{\overnu}{\overline{\nu}}
\newcommand{\dsN}{\mathds{N}}
\newcommand{\dwnN}{\dwn{N}}

%%% O

\newcommand{\rmo}{\mathrm{o}}
\newcommand{\uprmo}{\up{\rmo}}
\newcommand{\dwnrmo}{\dwn{\rmo}}
\newcommand{\dwnrmormo}{\dwn{\rmo\rmo}}
\newcommand{\rmO}{\mathrm{O}}
\newcommand{\calO}{\mathcal{O}}
\newcommand{\dwnrmormL}{\dwn{\rmo,\rmL}}
\newcommand{\dwnrmobflrmL}{\dwn{\rmo,\bfl,\rmL}}
\newcommand{\dwnrmobfl}{\dwn{\rmo,\bfl}}
\newcommand{\sfo}{\mathsf{o}}
\newcommand{\tilsfo}{\til{\sfo}}
\newcommand{\primesfo}{\sfo\prime}
\newcommand{\upsfo}{\up{\sfo}}
\newcommand{\uptilsfo}{\up{\tilsfo}}
\newcommand{\upprimesfo}{\up{\primesfo}}
\newcommand{\oversfo}{\overline{\sfo}}
\newcommand{\upoversfo}{\up{\oversfo}}
\newcommand{\overprimesfo}{\overline{\primesfo}}
\newcommand{\upoverprimesfo}{\up{\overprimesfo}}
\newcommand{\dwnop}{\dwn{\mathsf{op}}}

%%% P

\newcommand{\bbP}{\mathbb{P}}
\newcommand{\bfp}{\mathbf{p}}
\newcommand{\bfP}{\mathbf{P}}
\newcommand{\dwnbfP}{\dwn{\bfP}}
\newcommand{\rmP}{\mathrm{P}}
\newcommand{\tilp}{\til{p}}
\newcommand{\tilrmP}{\til{\rmP}}
\newcommand{\rmPoplus}{\mathrm{P}_\oplus}
\newcommand{\rmPopluscom}[1]{\mathrm{P}_\oplus^{(#1)}}
\newcommand{\dwnrmP}{\dwn{\rmP}}
\newcommand{\dwnrmPoplus}{\dwn{\rmPoplus}}
\newcommand{\rmPL}{\rmP\dwnrmL}
\newcommand{\rmPKL}{\rmP\dwn{\rmKL}}
\newcommand{\rmPinfty}{\rmP\dwninfty}
\newcommand{\dwnp}{\dwn{p}}
\newcommand{\upparp}{\uppar{p}}
\newcommand{\dwnpi}{\dwn{\pi}}
\newcommand{\upp}{\up{p}}
\newcommand{\tilpi}{\til{\pi}}
\newcommand{\bbPrmf}{\bbP\dwn{\rmf}}
\newcommand{\sfP}{\mathsf{P}}
\newcommand{\sfp}{\mathsf{p}}
\newcommand{\frakP}{\mathfrak{P}}
\newcommand{\bbPart}{\bbP\up{\div}}
\newcommand{\dwnrmPmenos}{\dwn{\rmP\dwnmenos}}
\newcommand{\dwnrmPmais}{\dwn{\rmP\dwnmais}}
\newcommand{\tilsfp}{\til{\sfp}}
\newcommand{\rmp}{\mathrm{p}}
\newcommand{\uprmp}{\up{\rmp}}
\newcommand{\dwnparp}{\dwnpar{p}}
\newcommand{\upvirgp}{\up{\,p}}
\newcommand{\upumsobrep}{\up{\frac{1}{p}}}
\newcommand{\tilvarpi}{\til{\varpi}}
\newcommand{\brevevarpi}{\breve{\varpi}}
\newcommand{\frakp}{\mathfrak{p}}
\newcommand{\brevefrakp}{\breve{\frakp}}
\newcommand{\dotvarpi}{\dot{\varpi}}
\newcommand{\ddotvarpi}{\ddot{\varpi}}
\newcommand{\hatvarpi}{\hat{\varpi}}
\newcommand{\barvarpi}{\bar{\varpi}}
\newcommand{\dwnpsfs}{\dwn{p\sfs}}
\newcommand{\dwnpsft}{\dwn{p\sft}}
\newcommand{\dwnpk}{\dwn{pk}}
\newcommand{\dwnpq}{\dwn{pq}}
\newcommand{\dwnrmp}{\dwn{\rmp}}
\newcommand{\overrmp}{\overline{\rmp}}
\newcommand{\dwnoverrmp}{\dwn{\overrmp}}
\newcommand{\primermp}{\rmp\prime}
\newcommand{\doispi}{2\pi}
\newcommand{\hatPi}{\hat{\Pi}}
\newcommand{\dwntilPhi}{\dwn{\tilPhi}}
\newcommand{\dwnPsin}{\dwn{\Psi,n}}

%%% Q

\newcommand{\Q}{\mathbb{Q}}
\newcommand{\rmQ}{\mathrm{Q}}
\newcommand{\tilq}{\til{q}}
\newcommand{\bfQ}{\mathbf{Q}}
\newcommand{\dwnq}{\dwn{q}}
\newcommand{\calQ}{\mathcal{Q}}
\newcommand{\bfq}{\mathbf{q}}
\newcommand{\upq}{\up{q}}
\newcommand{\upquatro}{\up{4}}
\newcommand{\dwnquatro}{\dwn{4}}
\newcommand{\rmq}{\mathrm{q}}
\newcommand{\frakq}{\mathfrak{q}}
\newcommand{\dwnqsfs}{\dwn{q\sfs}}
\newcommand{\dwnqsft}{\dwn{q\sft}}
\newcommand{\dwnrmq}{\dwn{\rmq}}
\newcommand{\primermq}{\rmq\prime}

%%% R

\newcommand{\R}{\mathbb{R}}
\newcommand{\Rmais}{\mathbb{R}^+}
\newcommand{\Rzeromais}{\mathbb{R}^+_0}
\newcommand{\upr}{\up{r}}
\newcommand{\dwnr}{\dwn{r}}
\newcommand{\rmR}{\mathrm{R}}
\newcommand{\rmr}{\mathrm{r}}
\newcommand{\sfR}{\mathsf{R}}
\newcommand{\dwnR}{\dwn{\R}}
\newcommand{\calR}{\mathcal{R}}
\newcommand{\dwnrho}{\dwn{\rho}}
\newcommand{\dwnhatrho}{\dwn{\hatrho}}
\newcommand{\frakR}{\mathfrak{R}}
\newcommand{\dwnRCA}{\dwn{\mathrm{RCA}}}
\newcommand{\dwnrhormF}{\dwn{\rho,\rmF}}
\newcommand{\dwnrhormPmenos}{\dwn{\rho,\rmP\dwnmenos}}
\newcommand{\dwnrhormPmais}{\dwn{\rho,\rmP\dwnmais}}
\newcommand{\dwnrhopsi}{\dwn{\rho,\psi}}
\newcommand{\dwntilrho}{\dwn{\tilrho}}
\newcommand{\frakr}{\mathfrak{r}}
\newcommand{\uprho}{\up{\rho}}
\newcommand{\overrho}{\overline{\rho}}

%%% S

\newcommand{\calS}{\mathcal{S}}
\newcommand{\fraks}{\mathfrak{s}}
\newcommand{\rms}{\mathrm{s}}
\newcommand{\rmS}{\mathrm{S}}
\newcommand{\rmSbfn}{\mathrm{S}\dwnbfn}
\newcommand{\rmSbfncom}[1]{\mathrm{S}\dwn{\bfn #1}}
\newcommand{\rmSbfnmais}{\mathrm{S}\dwn{\bfn +}}
\newcommand{\rmSbfnmenos}{\mathrm{S}\dwn{\bfn-}}
\newcommand{\ups}{\up{s}}
\newcommand{\dwns}{\dwn{s}}
\newcommand{\sfS}{\mathsf{S}}
\newcommand{\dwnsfS}{\dwn{\sfS}}
\newcommand{\sfs}{\mathsf{s}}
\newcommand{\dwnsfs}{\dwn{\sfs}}
\newcommand{\bfS}{\mathbf{S}}
\newcommand{\dwnst}{\dwn{st}}
\newcommand{\dwnrms}{\dwn{\rms}}
\newcommand{\dwnsigma}{\dwn{\sigma}}
\newcommand{\dwnrmS}{\dwn{\rmS}}
\newcommand{\subseteqvet}{\stackrel{\bbV}{\subseteq}}
\newcommand{\frakS}{\mathfrak{S}}
\newcommand{\rmSi}{\mathrm{Si}}
\newcommand{\dwnrmSdwni}{\dwn{\rmS\dwni}}
\newcommand{\tilsfs}{\til{\sfs}}
\newcommand{\primesfs}{\sfs\prime}
\newcommand{\dwnsfssft}{\dwn{\sfs\sft}}
\newcommand{\upsfS}{\up{\sfS}}

%%% T

\newcommand{\calT}{\mathcal{T}}
\newcommand{\frakt}{\mathfrak{t}}
\newcommand{\rmt}{\mathrm{t}}
\newcommand{\rmT}{\mathrm{T}}
\newcommand{\SFT}{\textsf{T}}
\newcommand{\sft}{\mathsf{t}}
\newcommand{\dwntheta}{\dwn{\theta}}
\newcommand{\dwntres}{\dwn{3}}
\newcommand{\uptres}{\up{3}}
\newcommand{\dwnrmT}{\dwn{\rmT}}
\newcommand{\overrmT}{\overline{\rmT}}
\newcommand{\dwnrmt}{\dwn{\rmt}}
\newcommand{\bfT}{\mathbf{T}}
\newcommand{\tilTheta}{\til{\Theta}}
\newcommand{\dwnTheta}{\dwn{\Theta}}
\newcommand{\dwntilTheta}{\dwn{\tilTheta}}
\newcommand{\dwnThetabflbfn}{\dwn{\Theta,\bfl,\bfn}}
\newcommand{\dwntilThetabflbfn}{\dwn{\tilTheta\bfl,\bfn}}
\newcommand{\tilsft}{\til{\sft}}
\newcommand{\primesft}{\sft\prime}
\newcommand{\dwntimes}{\dwn{\times}}
\newcommand{\hatrmT}{\hat{\rmT}}
\newcommand{\bbT}{\mathbb{T}}
\newcommand{\hatrmt}{\hat{\rmt}}
\newcommand{\dwnsft}{\dwn{\sft}}
\newcommand{\checkrmT}{\check{\rmT}}
\newcommand{\checkbfT}{\check{\bfT}}
\newcommand{\dwnbbT}{\dwn{\bbT}}
\newcommand{\dwnbbTuprmd}{\dwn{\bbT\uprmd}}
\newcommand{\dwnbbTuprmddwnmeio}{\dwn{\bbT\uprmd\dwnmeio}}
\newcommand{\tilrmT}{\til{\rmT}}
\newcommand{\checkcalT}{\check{\calT}}
\newcommand{\hatcalT}{\hat{\calT}}

%%% U

\newcommand{\calU}{\mathcal{U}}
\newcommand{\overrmU}{\overline{\rmU}}
\newcommand{\rmU}{\mathrm{U}}
\newcommand{\rmUoplus}{\rmU_\oplus}
\newcommand{\rmUopluscom}[1]{\rmU_\oplus^{(n)}}
\newcommand{\dwnu}{\dwn{u}}
\newcommand{\sfU}{\mathsf{U}}
\newcommand{\bfu}{\mathbf{u}}
\newcommand{\dwnrmU}{\dwn{\rmU}}
\newcommand{\dwnsfUx}{\dwn{\sfU\dwnx}}
\newcommand{\dwnsfU}{\dwn{\sfU}}
\newcommand{\upum}{\up{1}}
\newcommand{\frakU}{\mathfrak{U}}
\newcommand{\bfU}{\mathbf{U}}
\newcommand{\fraku}{\mathfrak{u}}
\newcommand{\dwnumn}{\dwn{1,n}}
\newcommand{\dwnUpsilonfrakl}{\dwn{\Upsilon\dwnfrakl}}

%%% V

\newcommand{\rmV}{\mathrm{V}}
\newcommand{\rmVestrelan}{(\rmV\upast)\upn}
\newcommand{\rmVn}{\rmV\upn}
\newcommand{\rmVestrela}{\rmV\upast}
\newcommand{\grassV}{\wedge \rmV}
\newcommand{\grassVn}{\wedge_n \! \rmV}
\newcommand{\grassVum}{\wedge_1 \! \rmV}
\newcommand{\grassVm}{\wedge_m \! \rmV}
\newcommand{\grassVestrela}{\wedge^{\! \! *} \rmV}
\newcommand{\grassVestrelan}{\wedge^{\! \! *}_n \! \rmV}
\newcommand{\grassVestrelaum}{\wedge^{\! \! *}_1 \! \rmV}
\newcommand{\grassVestrelam}{\wedge^{\! \! *}_m \! \rmV}
\newcommand{\grassVestrelazero}{\wedge^{\! \! *}_0 \rmV}
\newcommand{\fockV}{\rmF(\rmV)}
\newcommand{\rmv}{\mathrm{v}}
\newcommand{\dwnrmv}{\dwn{\rmv}}
\newcommand{\bbV}{\mathbb{V}}
\newcommand{\bbVn}{\bbV\dwnn}
\newcommand{\bfv}{\mathbf{v}}
\newcommand{\bfvn}{\bfv\dwnn}
\newcommand{\bbVR}{\bbV\dwn{\R}}
\newcommand{\bbVdiamond}{\bbV\dwncbdiamond}
\newcommand{\bbVdiamondn}{\bbV\dwn{\cbdiamond,n}}
\newcommand{\dwnbfv}{\dwn{\bfv}}
\newcommand{\dwnbfvLambda}{\dwn{\bfv,\,\Lambda}}
\newcommand{\calV}{\mathcal{V}}
\newcommand{\dwnrmV}{\dwn{\rmV}}
\newcommand{\bfV}{\mathbf{V}}
\newcommand{\dwnVert}{\dwn{\Vert}}
\newcommand{\tilbfv}{\til{\bfv}}
\newcommand{\frakv}{\mathfrak{v}}
\newcommand{\overrmv}{\overline{\rmv}}
\newcommand{\hatrmv}{\hat{\rmv}}
\newcommand{\scrV}{\mathscr{V}}

%%% W

\newcommand{\w}{\omega}
\newcommand{\rmw}{\mathrm{w}}
\newcommand{\calW}{\mathcal{W}}
\newcommand{\dwnw}{\dwn{w}}
\newcommand{\dwnrmw}{\dwn{\rmw}}
\newcommand{\bfw}{\mathbf{w}}
\newcommand{\bbW}{\mathbb{W}}
\newcommand{\bbWdiamond}{\bbW\dwn{\,\cbdiamond}}
\newcommand{\dwnbfw}{\dwn{\bfw}}
\newcommand{\overrmW}{\overline{\rmW}}
\newcommand{\dwnrmW}{\dwn{\rmW}}
\newcommand{\sfw}{\mathsf{w}}
\newcommand{\dwnbbWdiamond}{\dwn{\bbWdiamond}}
\newcommand{\frakW}{\mathfrak{W}}
\newcommand{\frakw}{\mathfrak{w}}
\newcommand{\tilfrakW}{\til{\frakW}}
\newcommand{\brevefrakW}{\breve{\frakW}}
\newcommand{\ddotfrakW}{\ddot{\frakW}}
\newcommand{\dotfrakW}{\dot{\frakW}}
\newcommand{\hatfrakW}{\hat{\frakW}}
\newcommand{\barfrakW}{\bar{\frakW}}
\newcommand{\dwnwsfs}{\dwn{w\sfs}}
\newcommand{\dwnwsft}{\dwn{w\sft}}
\newcommand{\dwnwtilsfs}{\dwn{w\tilsfs}}
\newcommand{\dwnwtilsft}{\dwn{w\tilsft}}
\newcommand{\dwnwprimesfs}{\dwn{w\primesfs}}
\newcommand{\dwnwprimesft}{\dwn{w\primesft}}
\newcommand{\overrmw}{\overline{\rmw}}

%%% X

\newcommand{\bbX}{\mathbb{X}}
\newcommand{\bbXL}{\bbX\dwnrmL}
\newcommand{\bbXinfty}{\bbX\dwninfty}
\newcommand{\bfx}{\mathbf{x}}
\newcommand{\calX}{\mathcal{X}}
\newcommand{\rmX}{\mathrm{X}}
\newcommand{\rmx}{\mathrm{x}}
\newcommand{\overx}{\overline{x}}
\newcommand{\dwnx}{\dwn{x}}
\newcommand{\dwnrmx}{\dwn{\rmx}}
\newcommand{\dwnrmxrmy}{\dwn{\rmx\rmy}}
\newcommand{\dwnxy}{\dwn{xy}}
\newcommand{\dwnxs}{\dwn{xs}}
\newcommand{\dwnxt}{\dwn{xt}}
\newcommand{\dwnxupplus}{\dwn{(x,\uparrow,+)}}
\newcommand{\dwnxdwnplus}{\dwn{(x,\downarrow,+)}}
\newcommand{\dwnxupminus}{\dwn{(x,\uparrow,-)}}
\newcommand{\dwnxdwnminus}{\dwn{(x,\downarrow,-)}}
\newcommand{\dwnxsnu}{\dwn{(x,s,\nu)}}
\newcommand{\hatx}{\hat{x}}
\newcommand{\tilx}{\til{x}}
\newcommand{\upx}{\up{x}}
\newcommand{\dwnxdwnrmo}{\dwn{x\dwnrmo}}
\newcommand{\upparxLambda}{\up{(x,\Lambda)}}
\newcommand{\parxLambda}{(x,\Lambda)}
\newcommand{\dwnxsfs}{\dwn{x\sfs}} 
\newcommand{\dwnxsft}{\dwn{x\sft}}
\newcommand{\dwnxtilsfs}{\dwn{x\tilsfs}}
\newcommand{\dwnxtilsft}{\dwn{x\tilsft}}
\newcommand{\dwnxprimesfs}{\dwn{x\primesfs}}
\newcommand{\dwnxprimesft}{\dwn{x\primesft}}
\newcommand{\overrmx}{\overline{\rmx}}
\newcommand{\hatbbX}{\hat{\bbX}}
\newcommand{\dwnxsfsalpha}{\dwn{(x,\sfs,\alpha)}} 
\newcommand{\dwnxcdot}{\dwn{x\,\cdot}}

%%% Y

\newcommand{\bfy}{\mathbf{y}}
\newcommand{\rmY}{\mathrm{Y}}
\newcommand{\overy}{\overline{y}}
\newcommand{\rmy}{\mathrm{y}}
\newcommand{\dwnrmy}{\dwn{\rmy}}
\newcommand{\dwny}{\dwn{y}}
\newcommand{\dwnyx}{\dwn{yx}}
\newcommand{\dwnys}{\dwn{ys}}
\newcommand{\dwnyt}{\dwn{yt}}
\newcommand{\dwnyupplus}{\dwn{(y,\uparrow,+)}}
\newcommand{\dwnydwnplus}{\dwn{(y,\downarrow,+)}}
\newcommand{\dwnyupminus}{\dwn{(y,\uparrow,-)}}
\newcommand{\dwnydwnminus}{\dwn{(y,\downarrow,-)}}
\newcommand{\dwnysnu}{\dwn{(y,s,\nu)}}
\newcommand{\dwnyzero}{\dwn{y\zero}}
\newcommand{\dwnysfs}{\dwn{y\sfs}}
\newcommand{\dwnytilsfs}{\dwn{y\tilsfs}}
\newcommand{\dwnysft}{\dwn{y\sft}}
\newcommand{\dwnytilsft}{\dwn{y\tilsft}}
\newcommand{\dwnyprimesfs}{\dwn{y\primesfs}}
\newcommand{\dwnyprimesft}{\dwn{y\primesft}}
\newcommand{\overrmy}{\overline{\rmy}}
\newcommand{\dwnysftsfo}{\dwn{(y,\sft,\sfo)}}

%%% Z

\newcommand{\Z}{\mathbb{Z}}
\newcommand{\zero}{\mathbf{0}}
\newcommand{\rmz}{\mathrm{z}}
\newcommand{\dwnz}{\dwn{z}}
\newcommand{\dwnrmz}{\dwn{\rmz}}
\newcommand{\dwnzerozero}{\dwn{\zero\zero}}
\newcommand{\dwnzero}{\dwn{\zero}}
\newcommand{\dwnzzero}{\dwn{0}}
\newcommand{\zzero}{0}
\newcommand{\frakZ}{\mathfrak{Z}}
\newcommand{\tilz}{\til{z}}
\newcommand{\upzero}{\up{\zero}}
\newcommand{\dwnzsfs}{\dwn{z\sfs}}
\newcommand{\dwnzsft}{\dwn{z\sft}}
\newcommand{\dwnztilsfs}{\dwn{z\tilsfs}}
\newcommand{\dwnztilsft}{\dwn{z\tilsft}}
\newcommand{\dwnzprimesfs}{\dwn{z\primesfs}}
\newcommand{\dwnzprimesft}{\dwn{z\primesft}}
\newcommand{\overrmz}{\overline{\rmz}}
\newcommand{\dwnmenosz}{\dwn{-z}}

%%% ACENTOS
\newcommand{\til}[1]{\tilde{#1}}
\newcommand{\upprime}{\up{\prime}}

%%%%%%%%%%%%%%%%%%%%%%%%%%%%%%%%%%%%%%%%%%%%%%%%%%%%%%%%%%%%%%%%%%%%%%%%%%%

%%%%%% MATEMATICOS

% CDOTVIRG
\newcommand{\cdotvirg}{\,\cdot\,}

% EXPOENTES
\newcommand{\upmeio}{^{\frac{1}{2}}}
\newcommand{\upsubseteq}{\up{\subseteq}}
\newcommand{\upsubsetneq}{\up{\subsetneq}}
\newcommand{\upcap}{\up{\cap}}
\newcommand{\upcup}{\up{\cup}}

% SUBSCRITOS
\newcommand{\dwnpm}{\dwn{\pm}}
\newcommand{\dwnmp}{\dwn{\mp}}
\newcommand{\dwnplus}{\dwn{+}}
\newcommand{\dwnmais}{\dwn{+}}
\newcommand{\dwnminus}{\dwn{-}}
\newcommand{\dwnmenos}{\dwn{-}}
\newcommand{\dwnoplus}{\dwn{\oplus}}
\newcommand{\dwnupplus}{\dwn{\uparrow,+}}
\newcommand{\dwndwnplus}{\dwn{\downarrow,+}}
\newcommand{\dwnupminus}{\dwn{\uparrow,-}}
\newcommand{\dwndwnminus}{\dwn{\downarrow,-}}
\newcommand{\dwnboldmeio}{\dwn{\um/2}}
\newcommand{\dwnmaior}{\dwn{>}}
\newcommand{\dwnmenor}{\dwn{<}}
\newcommand{\dwnmeio}{\dwn{\meio}}

%%%%%% COLCHETEAMENTO
\newcommand{\dbrack}[1]{\ldbrack #1 \rdbrack}
\newcommand{\igual}[1]{\stackrel{\prov{#1}}{=}}
\newcommand{\desiesq}[1]{\stackrel{\prov{#1}}{\leq}}
\newcommand{\desidir}[1]{\stackrel{\prov{#1}}{\geq}}
\newcommand{\paren}[1]{(#1)}
\newcommand{\classe}[1]{[#1]}
\newcommand{\commentLucas}[1]{[++L \textsc{#1}  L++]}

% COLCHETES GRASSMANN
\newcommand{\colchPP}[2]{\left[#1\right. \left[#2\right.}
\newcommand{\colchPG}[2]{\left[#1\right. \left.#2\right]}
\newcommand{\colchGP}[2]{\left.#1\right] \left[#2\right.}
\newcommand{\colchGG}[2]{\left.#1\right] \left.#2\right]}
\newcommand{\colchP}[1]{\left[#1\right.}
\newcommand{\colchG}[1]{\left.#1\right]}  

%%%%%% CONJUNTOS

% BASE ORTO
\newcommand{\baseortonormal}{\{ \psi_i \}_{i \in \rmI}}

% C VAZIO
\newcommand{\vazio}{\{\,\,\, \}}

% FUNCOES
\newcommand{\funcoes}{\mathcal{F}}

% OPERADORES LIMITADOS
\newcommand{\oplim}{\mathcal{B}}
\newcommand{\oplimH}{\mathcal{B}(\calH)}

%%%%%% FUNCOES

% ANTICOMUTADOR
\newcommand{\anticomuta}[2]{\left[ #1, #2 \right]_+}

% COLCH
%\newcommand{\colchrmh}[2]{\left[#1 \mid#2\right]\dwnk}
\newcommand{\colchrmh}[2]{\rmh\dwn{[#1 \mid #2]}}

% COMUTADOR
\newcommand{\comuta}[2]{\left[#1, #2 \right]_{-}}

% CONJUNTO
\newcommand{\conj}[1]{\left\{#1\right\}}
\newcommand{\conjpm}{\conj{+,-}}

% UPDIV
\newcommand{\updiv}{\up{\div}}

% ELEMENTO BILINEA
\newcommand{\bilinBSB}{\prodintsimpl{\rmB}{\rmS \rmB}}

% ESPACO GERADO
\newcommand{\lin}{\mathrm{lin}}

% EXPONENCIAL GRASSMANN
\newcommand{\eegrasspm}[3]{\ee_{#1 \left(#2,#3\right)}}
\newcommand{\eegrassp}[2]{\ee_{\left(#1,#2\right)}}
\newcommand{\eegrassm}[2]{\ee_{-\left(#1,#2\right)}}
\newcommand{\EEchern}[1]{\rmE\upparn_{#1}}
\newcommand{\EEcherncom}[2]{\rmE\up{(#2)}_{#1}}
\newcommand{\eez}{\ee\up{\frac{\betabfn\BhnB}{2}}}

% NORMA
\newcommand{\norma}[1]{\Vert #1 \Vert}
\newcommand{\normax}[1]{\Vert #1 \Vert\dwn{\infty}}
\newcommand{\normanome}{\Vert \cdot \Vert}
\newcommand{\absoluto}[1]{\left|#1\right|}

% NUCLEO
\newcommand{\nucl}{\text{\textrm{N\'ucleo}}}

% NUMERICAS
\newcommand{\heaviside}{\Theta}

% OPLUS, OTIMES, ETC
\newcommand{\osum}{\oplus}
\newcommand{\otimescwedge}{\stackrel{\cwedge}{\otimes}}
\newcommand{\otimescvee}{\stackrel{\cvee}{\otimes}}
\newcommand{\otimesrmo}{\stackrel{\rmo}{\otimes}}

% PERP
\newcommand{\upperp}{^\perp}

% PRODUTO INTERNO
\newcommand{\prodint}[2]{\left< #1, #2 \right>_{\calH}}
\newcommand{\prodintsimpl}[2]{\left< #1, #2 \right>}
\newcommand{\prodintarg}[3]{\left< #1, #2 \right>_{#3}}
\newcommand{\prodintfrak}[2]{\left< #1, #2 \right>_{\frakH}}
\newcommand{\pint}[2]{\left< #1, #2 \right>}
\newcommand{\pintcom}[3]{\left< #1, #2 \right>\dwn{#3}}
\newcommand{\pintpar}[2]{\left( #1, #2 \right)_\wedge}
\newcommand{\pintparn}[2]{\left( #1, #2 \right)_n}
\newcommand{\pintparzero}[2]{\left( #1, #2 \right)_0}
\newcommand{\pintparum}[2]{\left( #1, #2 \right)_1}

% WEDGES
\newcommand{\wedgeestrela}{\wedge^{\! \! *}}
\newcommand{\cwedge}{\curlywedge}
\newcommand{\cvee}{\curlyvee}
\newcommand{\bdiamond}{\blackdiamond}
\newcommand{\cbdiamond}{\!\!\blackdiamond}
\newcommand{\dwncbdiamond}{\dwn{\cbdiamond}}

%%%%%% LOGICOS

\newcommand{\baixo}{\downarrow}
\newcommand{\cima}{\uparrow}
\newcommand{\sse}{\leftrightarrow}
\newcommand{\Sse}{\Leftrightarrow}
\newcommand{\To}{\Rightarrow}
\newcommand{\tee}{\mathbb{T}}
\newcommand{\Oord}{\overrightarrow{\rmO}}

%%%%%% NOMES

% ALGEBRA AUTODUAL
\newcommand{\Aadual}{\overline{\bbA}_{\smallRAD}}

% BCS
\newcommand{\BCS}{\textrm{BCS}}
\newcommand{\mBCS}{\mathrm{BCS}}

% CAR
\newcommand{\car}{\mathrm{CAR}}
\newcommand{\CAR}{\textrm{CAR }}
\newcommand{\scar}{\mathrm{CAR\smallAD}}
\newcommand{\sCAR}{\textrm{CAR\smallAD }}
\newcommand{\mCAR}{\mathrm{CAR}}

\newcommand{\rca}{\mathrm{RCA}}
\newcommand{\RCA}{\textrm{RCA }}
\newcommand{\srca}{\mathrm{RCA\smallAD}}
\newcommand{\sRCA}{\textrm{RCA\smallAD }}

\newcommand{\smallAD}{\text{{\tiny AD }}}
\newcommand{\smallRAD}{\text{{\tiny RAD }}}

% CPAMM
\newcommand{\cpamm}{\fontsize{14.9pt}{10pt}\selectfont \textsc{c.p.a.m.}\hskip3pt\large}

% ENVELOPE CONVEXO E SUPORTE COMPACTO
\newcommand{\rmco}{\mathrm{co}}
\newcommand{\overrmco}{\overline{\rmco}}
\newcommand{\dwnrmco}{\dwn{\rmco}}

% ESPECTRO
%\newcommand{\spec}{\mathrm{spec}}

% ESTRELA
\newcommand{\estrela}{* $\,$}
\newcommand{\emphestrela}{\emph{*}}

% IMAGEM
\newcommand{\ran}{\mathrm{ran}}

% KMS
\newcommand{\KMS}{\textrm{KMS}}

% PFAFIANO
\newcommand{\pfa}{\mathrm{Pf}}

% PROVAVEL
\newcommand{\rmPr}{\mathrm{Pr}}

% REAL E IMAGINARIO
\newcommand{\realp}{\mathsf{Re}}
\newcommand{\imagp}{\mathsf{Im}}

% SUPORTE
\newcommand{\supp}{\mathrm{supp}}

% TRUE
\newcommand{\rmTr}{\mathrm{Tr}}

% TRACO
%\newcommand{\tr}{\mathrm{tr}}
%\newcommand{\Tr}{\mathrm{Tr}}
\newcommand{\central}{\rho_{\um/2}}
\newcommand{\rhon}{\rho\dwn{n}}
\newcommand{\rhobfnmenos}{\rho\dwn{\bfn -}}
\newcommand{\rhobfnmais}{\rho\dwn{\bfn +}}
\newcommand{\rhoparn}{\rho\dwn{\paren{n}}}

%%%%%% SIMBOLOS
\newcommand{\defdir}{\stackrel{\, \, \, \, \text{{\large$.$}}}{=}}
\newcommand{\defesq}{\stackrel{\!  \text{{\large$.$ }}}{=}}
\newcommand{\eqdot}{\stackrel{\text{{\large$.$}}}{=}}
\newcommand{\menos}{-}
\newcommand{\mais}{+}
\newcommand{\menosinfty}{{-\infty}}
\newcommand{\dwninfty}{\dwn{\infty}}
\newcommand{\upinfty}{\up{\infty}}
\newcommand{\menosum}{^{-1}}
\newcommand{\too}{\rightarrow}
\newcommand{\tooo}{\longrightarrow}
\newcommand{\miff}{\mathrm{iff}}
\newcommand{\um}{\mathbf{1}}
\newcommand{\bbum}{\mathbb{1}}
\newcommand{\circast}{\circledast}
\newcommand{\upast}{^\ast}
\newcommand{\updagger}{\up{\dagger}}
\newcommand{\upbeta}{^\beta}
\newcommand{\upcircast}{^\circast}
\newcommand{\upcircastcom}[1]{(#1)^\circast}
\newcommand{\upcircastfunc}{(\, \cdot \,)\upcircast}
\newcommand{\upastn}{^{\ast n}}
\newcommand{\updois}{^2}
\newcommand{\upmenos}{^{-1}}
\newcommand{\dwnum}{\dwn{1}}
\newcommand{\dwnparum}{\dwnpar{1}}
\newcommand{\dwnumbold}{\dwn{\um}}
\newcommand{\dwnlequm}{\dwn{\leq 1}}
\newcommand{\dwndois}{\dwn{2}}
\newcommand{\uppi}{^\pi}
\newcommand{\upsigma}{^\sigma}
\newcommand{\upvarsigma}{^\varsigma}
\newcommand{\up}[1]{^{#1}}
\newcommand{\dwn}[1]{_{#1}}
\newcommand{\uppar}[1]{^{(#1)}}
\newcommand{\sgn}{\mathrm{sgn}}
\newcommand{\upmaissinal}{\up{+}}
\newcommand{\upmenossinal}{\up{-}}
\newcommand{\upsharp}{\up{\sharp}}
\newcommand{\upflat}{\up{\flat}}
\newcommand{\dwnpar}[1]{\dwn{(#1)}}
\newcommand{\upcolch}[1]{\up{[#1]}}
\newcommand{\dwncolch}[1]{\dwn{[#1]}}
\newcommand{\upmais}{\up{+}}
\newcommand{\upmenonos}{\up{-}}
\newcommand{\spaceave}[1]{\widehat{#1}}
\newcommand{\upoplus}{\up{\oplus}}
\newcommand{\meio}{1/2}
\newcommand{\colch}[1]{[#1]}
\newcommand{\umsobre}[1]{\frac{1}{#1}}
\newcommand{\upwedgestarn}{\up{\wedge\ast n}}
\newcommand{\upwedgen}{\up{\wedge n}}

%%%%%%%%%%%%%%%%%%%%%%%%%%%%%%%%%%%%%%%%%%%%%%%%%%%%%%%%%%%%%%%%%%%%%%%%%%%

%%%%%% PARAGRAFACAO

% ALINHAMENTO
\newcommand{\linha}{$\,$\\}
\newcommand{\noin}{\noindent}

% CORES
\newcommand{\azui}[1]{{\color[rgb]{0,0,1}{#1}}}
\newcommand{\vermei}[1]{{\color[rgb]{1,0,0}{#1}}}

% DEMARCACAO DE MUDANCA DA TESE
\newcommand{\FFF}{\ESP [FFF] \ESP}
\newcommand{\MMM}{\ESP \vermei{[MMM]} \ESP}
\newcommand{\NNN}{\ESP \azui{[NNN]} \ESP}

% DESENHO
\newcommand{\circdiam}{0.025}
\newcommand{\sitedist}{0.045}

% ESPACAMENTO
\newcommand{\ESP}{$\,$}
\newcommand{\esp}{\,}
\newcommand{\tiraesp}{\!}
\newcommand{\tab}{\quad \quad \quad}
\newcommand{\TAB}{$\quad \quad \quad$}
\newcommand{\nqquad}{\! \! \!}

% FIM DE DEFINICAO 
\newcommand{\fimdef}{\hfill \textbf{\textsc{||||}} \par\vspace{5pt}}
\newcommand{\fimprova}{\hfill \textbf{\textsc{////}} \par\vspace{5pt}}
\newcommand{\fimprovaq}{\hfill \hbox{\vrule width 7pt depth 0pt height 7pt}%
  \par\vspace{10pt}}

% ITEM
\newcommand{\itembox}{\item[\BOX]}
\newcommand{\itemdef}[1]{\textsf{#1}}
\newcommand{\itemdefi}{\itemdef{i}}
\newcommand{\itemdefii}{\itemdef{ii}}
\newcommand{\itemdefiii}{\itemdef{iii}}
\newcommand{\itemdefiv}{\itemdef{iv}}
\newcommand{\itemdefv}{\itemdef{v}}
\newcommand{\itemdefvi}{\itemdef{vi}}

% PONTUACAO
\newcommand{\BOX}{$\Box$}
\newcommand{\doispontos}{\, \, \mathbf{:} \, \,}
\newcommand{\etc}{\textsf{etc}}
\newcommand{\punkt}{\, \, \, .}
\newcommand{\setpunkt}{\, \, \mathbf{:} \, \,}
\newcommand{\punktvirg}{\, \, \, ;}
\newcommand{\virg}{\, \, \, ,}
\newcommand{\avirg}{,\,}
\newcommand{\parncom}[1]{^{(#1)}}
\newcommand{\parbaixoncom}[1]{_{(#1)}}
\newcommand{\parparbaixon}{)_{(n)}}
\newcommand{\parparbaixocom}[2]{(#2)_{(#1)}}

% TIPOS
\newcommand{\prov}[1]{\text{\textsf{#1}}}
\newcommand{\parinin}[1]{{\LARGE \textbf{#1}}}

%%%%%%%%%%%%%%%%%%%%%%%%%%%%%%%%%%%%%%%%%%%%%%%%%%%%%%%%%%%%%%%%%%%%%%%%%%%

%%%%%% PALAVRAS LATINAS 

% A POSTERIORI

\newcommand{\APOSTERIORI}{\emph{a posteriori}}

% CONFERE

\newcommand{\CF}{\textrm{cf}.}

% IBID

\newcommand{\IBID}{\textrm{ibid}.}

%%%%%%%%%%%%%%%%%%%%%%%%%%%%%%%%%%%%%%%%%%%%%%%%%%%%%%%%%%%%%%%%%%%%%%%%%%%

%%%%%% FILOSOFICOS

% METAFISICA ARISTOTELES
\newcommand{\metafis}{\emph{Metaf\'isica}}
\newcommand{\cmetafis}{\cite{metafis}}
\newcommand{\ccmetafis}[1]{\cite[#1]{metafis}}

%%%%%%%%%%%%%%%%%%%%%%%%%%%%%%%%%%%%%%%%%%%%%%%%%%%%%%%%%%%%%%%%%%%%%%%%%%%

%%%%%% ALGEBRA AUTODUAL

% ELEMENTO BILINEAR
\newcommand{\BSB}{\prodintsimpl{\rmB}{\rmS \rmB}}
\newcommand{\BHB}{\prodintsimpl{\rmB}{\rmH \rmB}}
\newcommand{\BHBbf}{\pint{\rmB}{\bfH\rmB}}
\newcommand{\BhnB}{\prodintsimpl{\rmB}{\rmhbfn \rmB}}

% ELEMENTOS
\newcommand{\rmBpsi}{\rmB(\psi)}
\newcommand{\rmBpsiestrela}{\rmB\upast(\psi)}
\newcommand{\rmBphi}{\rmB(\phi)}
\newcommand{\rmBphiestrela}{\rmB\upast(\phi)}

% VETORES
\newcommand{\tilpsi}{\til{\psi}}
\newcommand{\tilphi}{\til{\phi}}
\newcommand{\tilrho}{\til{\rho}}
\newcommand{\rhotil}{\tilrho}
\newcommand{\hatrho}{\hat{\rho}}

%%%%%%%%%%%%%%%%%%%%%%%%%%%%%%%%%%%%%%%%%%%%%%%%%%%%%%%%%%%%%%%%%%%%%%%%%%%

%%%%%% ALGEBRA DE GRASSMANN

% PSIS
\newcommand{\psiestrela}{\psi\upast}
\newcommand{\overpsiestrela}{\overline{\psiestrela}}
\newcommand{\overpsiestrelaparn}{\overline{\psiestrelaparn}}
\newcommand{\overpsikum}{\overline{\psi_k}_{(1)}}
\newcommand{\overpsikzero}{\overline{\psi_k}_{(0)}}
\newcommand{\psiparn}{\psi^{(n)}}
\newcommand{\psiparbaixon}{\psi_{(n)}}
\newcommand{\psiestrelaparn}{\psi^{*(n)}}
\newcommand{\psiestrelaparm}{\psi^{*(m)}}
\newcommand{\wedgepsiestrela}{\wedge \psi \upast}
\newcommand{\wedgepsi}{\wedge \psi}
\newcommand{\wedgeoverpsi}{\wedge \overpsi}
\newcommand{\wedgeoverpsiestrela}{\wedge \overpsiestrela}
\newcommand{\xiestrela}{\xi\upast}
\newcommand{\xitilestrela}{\tilde{\xi}\upast}
\newcommand{\xiestrelaupk}{(\xiestrela)\upk}
\newcommand{\overxiestrela}{\overline{\xiestrela}}
\newcommand{\overxiestrelaparn}{\overline{\xiestrelaparn}}
\newcommand{\overxi}{\overline{\xi}}
\newcommand{\zetaestrela}{\zeta\upast}
\newcommand{\etaestrela}{\eta\upast}
\newcommand{\overzetaestrela}{\overline{\zetaestrela}}
\newcommand{\overetaestrela}{\overline{\etaestrela}}
\newcommand{\overzetaestrelaparn}{\overline{\zetaestrela}^{(n)}}
\newcommand{\xiparn}{\xi^{(n)}}
\newcommand{\xiparm}{\xi^{(m)}}
\newcommand{\zetaparn}{\zeta^{(n)}}
\newcommand{\zetaparm}{\zeta^{(m)}}
\newcommand{\phiestrela}{\phi\upast}
\newcommand{\wedgephi}{\wedge \phi}
\newcommand{\phiparn}{\phi^{(n)}}
\newcommand{\phiparm}{\phi^{(m)}}
\newcommand{\phiparbaixocom}[2]{\phi_{#1(#2)}}
\newcommand{\overphiestrelaparn}{\overline{\phiestrelaparn}}
\newcommand{\thetaestrela}{\theta\upast}
\newcommand{\thetaestrelaparm}{\theta^{*(m)}}
\newcommand{\xiestrelaparn}{\xi^{*(n)}}
\newcommand{\xiestrelaparzero}{\xi^{*(0)}}
\newcommand{\xiestrelaparbaixon}{\xi\upast_{(n)}}
\newcommand{\xiestrelaparbaixonm}{\xi\upast_{(nm)}}
\newcommand{\xiestrelaparm}{\xi^{*(m)}}
\newcommand{\xiestrelaparmparbaixon}{\xi^{*(m)}_{(n)}}
\newcommand{\zetaestrelaparn}{\zeta^{*(n)}}
\newcommand{\zetaestrelaparm}{\zeta^{*(m)}}
\newcommand{\phiestrelaparn}{\phi^{*(n)}}
\newcommand{\phiestrelaparm}{\phi^{*(m)}}
\newcommand{\overphi}{\overline{\phi}}
\newcommand{\overPhi}{\overline{\Phi}}
\newcommand{\overGammaphi}{\overline{\Gamma \phi}}
\newcommand{\overGammaPhi}{\overline{\Gamma \Phi}}
\newcommand{\overphiparbaixok}{\overphi_{(k)}}
\newcommand{\overpsi}{\overline{\psi}}
\newcommand{\overPsi}{\overline{\Psi}}
\newcommand{\overpsiparbaixok}{\overpsi_{(k)}}
\newcommand{\overGammapsi}{\overline{\Gamma \psi}}
\newcommand{\overvarphi}{\overline{\varphi}}
\newcommand{\overvarphiuprmo}{\overline{\varphi^\rmo}}
\newcommand{\overvarphiuprmol}{\overline{\varphi_l^\rmo}}
\newcommand{\Gammavarphi}{\Gamma \varphi}
\newcommand{\overGammavarphi}{\overline{\Gamma \varphi}}
\newcommand{\overGammavarphiuprmo}{\overline{\Gamma \varphi^\rmo}}
\newcommand{\overGammavarphiuprmol}{\overline{\Gamma \varphi_l^\rmo}}
\newcommand{\overvarphiparbaixocom}[2]{\overvarphi_{#1(#2)}}
\newcommand{\overGammavarphiparbaixocom}[2]{\overGammavarphi_{#1(#2)}}
\newcommand{\overtilvarphi}{\overline{\tilvarphi}}
\newcommand{\varphiuprmo}{\varphi^{\rmo}}
\newcommand{\psiuprmo}{\psi^{\rmo}}
\newcommand{\psiparbaixocom}[2]{\psi_{#1(#2)}}
\newcommand{\overpsiuprmo}{\overline{\psi^{\rmo}}}
\newcommand{\varphiestrela}{\varphi \upast}
\newcommand{\varphilparbaixon}{\varphi_{l(n)}}
\newcommand{\varphiparbaixocom}[2]{\varphi_{#1(#2)}}
\newcommand{\Gammavarphiparbaixocom}[2]{\Gamma \varphi_{#1(#2)}}
\newcommand{\overphiestrela}{\overline{\phi\upast}}
\newcommand{\overthetaestrela}{\overline{\thetaestrela}}
\newcommand{\thetaestrelaparn}{\theta^{*(n)}}

% CONJUNTO
\newcommand{\rmPK}{\rmP \rmK}
\newcommand{\tilrmPK}{\til{\rmPK}}

% DERIVADA DE GRASSMANN
\newcommand{\grassderiv}[2]{\frac{\delta #2}{\delta #1}}
\newcommand{\grassderivpartial}[1]{\partial_{#1}}

% ELEMENTO BILINEAR
\newcommand{\KHK}{\pint{\rmK}{\rmH \rmK}}
\newcommand{\KHKcom}[2]{\pint{\rmK_{(#1)}}{\rmH \rmK_{(#2)}}}
\newcommand{\KHKmenosumcom}[2]{\pint{\rmK_{(#1)}}{\rmH\menosum \rmK_{(#2)}}}
\newcommand{\KK}[2]{\pint{\rmK_{(#1)}}{\rmK_{(#2)}}}
\newcommand{\KHKmenosum}{{\pint{\rmK}{\rmH\menosum \rmK}}}
\newcommand{\PKPK}[2]{\pint{(\rmP\rmK)_{(#1)}}{(\rmP\rmK)_{(#2)}}}
\newcommand{\PKPKsans}{\pint{\rmP\rmK}{\rmP\rmK}}

% FUNCAO UPSILO
\newcommand{\overUpsilon}{\overline{\Upsilon}}
\newcommand{\Upsilonestrela}{\Upsilon \upast}

% INTEGRAL DE GRASSMANN
\newcommand{\grassintKzero}{\int_\rmP \dd \rmK}
\newcommand{\grassintK}[1]{\int_\rmP \dd \rmK_{(#1)}}
\newcommand{\grassintKP}[2]{\int_{#2} \dd \rmK_{#1}}
\newcommand{\grassintcom}[2]{\int_{#2} \dd #1}
\newcommand{\grassintKPoplusn}{\int_{\rmPoplus} \dd \rmKopluscom{n}}
\newcommand{\grassintgaussKzero}[1]{\int_\rmP \dd \rmK(\mu_{#1})}
\newcommand{\grassintgaussK}[2]{\int_\rmP \dd \rmK_{(#1)}(\mu_{#2})}
\newcommand{\grassintgaussKP}[4]{\int_{#4} \dd \rmK_{#1}^{#2}(\mu_{#3})}
\newcommand{\grassintgausscom}[3]{\int_{#3} \dd #1(\mu_{#2})}

% ISOMORFISMO
\newcommand{\rmkrmP}{\rmk^{\rmP}_\bbA}
\newcommand{\rmkrmPupmenos}{\left(\rmkrmP\right)\upmenos}
\newcommand{\rmkcalB}{\rmk^{\calB}_\bbA}
\newcommand{\rmkwedgeP}{\rmk^{\rmP}_{\calB}}
\newcommand{\rmkrmPcom}[1]{\rmk^\rmP_{#1}}
\newcommand{\rmkrmPijtokl}{\rmk^\rmP_{ij \to kl}}
\newcommand{\rmkrmPijtoklcom}[4]{\rmk^\rmP_{#1 #2 \to #3 #4}}

% PRODUTO CIRCULAR
\newcommand{\prodcircP}{\stackrel{\rmP}{\circ}}
\newcommand{\produtoriacircP}[3]{\stackrel{\nqquad \nqquad \nqquad \rmP}{\bigcirc_{#1=#2}^{#3}}}

%%%%%%%%%%%%%%%%%%%%%%%%%%%%%%%%%%%%%%%%%%%%%%%%%%%%%%%%%%%%%%%%%%%%%%%%%%%%%%%%%%%%%%%%%%%%%%%%%%%%%%%%%%

%%%%%% ANTIGO: NOVOS COMANDOS

\newcommand{\sss}{$\,$}
\newcommand{\nss}{$\! \! \!$}
\newcommand{\nh}{n_{\mathcal{H}}}
\newcommand{\lTo}{\Leftarrow}
\newcommand{\upto}{\uparrow}
\newcommand{\db}[1]{\ldbrack #1 \rdbrack}
\newcommand{\mtext}[1]{\text{\small{\textsl{#1}}}}
\newcommand{\existsin}[2]{\exists #1  \! \! \: \in \! \! \! \; #2}
\newcommand{\thereis}{\exists}
\newcommand{\thereisin}[2]{\exists #1 \! \! \; \in \! \! \! \; #2}
\newcommand{\forallin}[2]{\forall #1  \! \! \: \in \! \!  \; #2}
\newcommand{\set}[1]{\left\{ #1 \right\}}
\newcommand{\mitemdef}[2]{\item[#1.] $#2$}
\newcommand{\mitemlisib}[1]{\item[] $#1$}
\newcommand{\mitemteo}[3]{\item[(#1. #2)] $#3$}
\newcommand{\ifff}[1]{$(#1) \iff$}
\newcommand{\mstring}[2]{\mtext{#1}\db{#2}}
\newcommand{\mstringp}[3]{\mtext{#1}(#2)\db{#3}}
\newcommand{\prodintl}[2]{\left< \left. #1 \,  \right|  #2 \right>}
\newcommand{\prodintlpar}[2]{\left( \left. #1 \,  \right|  #2 \right)}
\newcommand{\prodintlk}[3]{\left< \left. #1 \,  \right|  #2 \right>_{\! #3}}
\newcommand{\prodintlpark}[3]{\left( \left. #1 \,  \right|  #2 \right)_{\! #3}}
\newcommand{\prodintr}[2]{\left<  #1 \left| \,  #2 \right. \right>}
\newcommand{\prodintrpar}[2]{\left(  #1 \left| \,  #2 \right. \right)}
\newcommand{\prodintrk}[3]{\left<  #1 \left| \,  #2 \right. \right>_{\! #3}}
\newcommand{\prodintrpark}[3]{\left(  #1 \left| \,  #2 \right. \right)_{\! #3}}
\newcommand{\hilbaga}{\mathcal{H}}
\newcommand{\hilbagal}{\overline{\mathcal{H}}}
\newcommand{\hilbagatimes}[1]{\mathcal{H}^{[#1]}}
\newcommand{\uhilbaga}[1]{\mathcal{H}_{(#1)}}
\newcommand{\ouhilbaga}[1]{\overline{\mathcal{H}}_{(#1)}}
\newcommand{\ohilbaga}[1]{\overline{\mathcal{H}^{(#1)}}}
\newcommand{\nohilbaga}[1]{\mathcal{H}^{(#1)}}
\newcommand{\ohilbagaij}[2]{
  \overline{\mathcal{H}_{#1} \oplus \overline{\mathcal{H}}_{#2}}
}
\newcommand{\osumaga}{\mathcal{H} \oplus \overline{\mathcal{H}}}
\newcommand{\oosumaga}{\overline{\mathcal{H} \oplus \overline{\mathcal{H}}}}
\newcommand{\osumoaga}{\overline{\mathcal{H}} \oplus \mathcal{H}}
\newcommand{\bigwedgelin}[1]{\bigwedge \,^{\! \! \! \! #1} \,}
\newcommand{\bigwedgea}{\bigwedge \! \!}
\newcommand{\bigwedgeo}{\bigwedge \!}
\newcommand{\reais}{\mathbb{R}}
\newcommand{\complexos}{\mathbb{C}}
\newcommand{\naturais}{\mathbb{N}}
\newcommand{\inteiros}{\mathbb{Z}}
\newcommand{\toroT}{\mathbb{T}}
\newcommand{\toroX}{\mathbb{X}}
\newcommand{\toroXSL}{\mathbb{X}_{\mathrm{S} \times \mathfrak{L}}}
\newcommand{\toroTbxXI}{\mathbb{T}_\beta \times \mathbb{X}_{\mathrm{I}}}

\newcommand{\mbf}[1]{\mathbf{#1}}
\newcommand{\mbfF}{\mathbf{F}}
\newcommand{\mbfr}{\mathbf{r}}
\newcommand{\tilmbfr}{\tilde{\mathbf{r}}}
\newcommand{\mbfq}{\mathbf{q}}
\newcommand{\tilmbfq}{\tilde{\mathbf{q}}}
\newcommand{\mbfX}{\mathbf{X}}
\newcommand{\tilmbfX}{\tilde{\mathbf{X}}}
\newcommand{\ombf}[1]{\overline{\mathbf{#1}}}
\newcommand{\ombfX}{\ombf{X}}
\newcommand{\ob}{\overline{b}}
\newcommand{\oC}{\overline{C}}
\newcommand{\of}{\overline{f}}
\newcommand{\og}{\overline{g}}
\newcommand{\oG}{\overline{G}}
\newcommand{\oj}{\overline{j}}
\newcommand{\oJ}{\overline{J}}
\newcommand{\oL}{\overline{L}}
\newcommand{\oM}{\overline{M}}
\newcommand{\oN}{\overline{N}}
\newcommand{\onu}{\overline{\nu}}
\newcommand{\ou}{\overline{u}}
\newcommand{\oU}{\overline{U}}
\newcommand{\ov}{\overline{v}}
\newcommand{\oV}{\overline{V}}
\newcommand{\ow}{\overline{w}}
\newcommand{\oW}{\overline{W}}
\newcommand{\ox}{\overline{x}}
\newcommand{\oX}{\overline{X}}
\newcommand{\oY}{\overline{Y}}
\newcommand{\oalpha}{\overline{\alpha}}
\newcommand{\obeta}{\overline{\beta}}
\newcommand{\oomega}{\overline{\omega}}
\newcommand{\tilomega}{\tilde{\omega}}
\newcommand{\Ups}{\Upsilon}
\newcommand{\oUps}{\overline{\Upsilon}}
\newcommand{\oot}{\longleftarrow}
\newcommand{\bigwedgeuv}[2]{\sideset{_{(#1)}}{_{(#2)}} \bigwedge}
\newcommand{\bigwedgeUV}[2]{\sideset{#1}{#2} \bigwedge}
\newcommand{\psum}[1]{\sideset{}{^{(p)}}\sum_{#1}}
\newcommand{\timessum}[1]{\sideset{}{^{\bigotimes}}\sum_{#1}}
\newcommand{\bigotimespsi}{\bigotimes \!}
\newcommand{\fislim}[2]{\sideset{}{^{#1}}\lim_{#2}}
\newcommand{\leftrightpar}[1]{\left( #1 \right)}
\newcommand{\veczero}{\mathbf{0}}
\newcommand{\oveczero}{\overline{\mathbf{0}}}
\newcommand{\veczerof}[1]{\mathbf{0}_{#1}}
\newcommand{\oveczerof}[1]{\overline{\mathbf{0}}_{#1}}
\newcommand{\fespaco}[1]{\mathfrak{F}_{\bigwedge}^{#1}}
\newcommand{\ffuncao}{\mathfrak{F}}
\newcommand{\fa}{\mathfrak{a}}
\newcommand{\fA}{\mathfrak{A}}
\newcommand{\fB}{\mathfrak{B}}
\newcommand{\fd}{\mathfrak{d}}
\newcommand{\fD}{\mathfrak{D}}
\newcommand{\fe}{\mathfrak{e}}
\newcommand{\ff}{\mathfrak{f}}
\newcommand{\fg}{\mathfrak{g}}
\newcommand{\fH}{\mathfrak{H}}
\newcommand{\frki}{\mathfrak{i}}
\newcommand{\fl}{\mathfrak{l}}
\newcommand{\fL}{\mathfrak{L}}
\newcommand{\fm}{\mathfrak{m}}
\newcommand{\fM}{\mathfrak{M}}
\newcommand{\fS}{\mathfrak{S}}
\newcommand{\ft}{\mathfrak{t}}
\newcommand{\fT}{\mathfrak{T}}
\newcommand{\fu}{\mathfrak{u}}
\newcommand{\fw}{\mathfrak{w}}
\newcommand{\fibA}{\mathfib{A}}
\newcommand{\fibf}{\mathfib{f}}
\newcommand{\fibg}{\mathfib{g}}
\newcommand{\fibm}{\mathfib{m}}
\newcommand{\fibS}{\mathfib{S}}
\newcommand{\fibu}{\mathfib{u}}
\newcommand{\fecho}[1]{\dot{\overline{#1}}}
\newcommand{\wedgecria}[2]{\wedge \, ^{\! \! #1}_{\! \! #2} \,}
\newcommand{\opneg}[2]{\neg^{#1}_{#2}}
\newcommand{\tiln}{\tilde{n}}
%newcommand{\tilm}{\tilde{m}}
\newcommand{\tilalgW}{\tilde{\mathcal{W}}}
\newcommand{\vecum}{\mathbf{1}}
\newcommand{\vecumf}[1]{\mathbf{1}_{#1}}
\newcommand{\ovecum}{\overline{\mathbf{1}}}
\newcommand{\ovecumf}[1]{\overline{\mathbf{1}}_{#1}}
\newcommand{\algA}{\mathcal{A}}
\newcommand{\algAell}{\mathcal{A}_{\ell^2(\mathfrak{L})}}
\newcommand{\algASell}{\mathcal{A}_{\ell^2(\mathrm{S} \times \mathfrak{L})}}
\newcommand{\oalgA}{\overline{\mathcal{A}}}
\newcommand{\ealgA}{\overline{\mathcal{A}}_{+,1}}
\newcommand{\ealgAell}{\overline{\mathcal{A}}_{+,1}(\ell^2(\mathfrak{L}))}
\newcommand{\ealgASell}{\overline{\mathcal{A}}_{+,1}(\ell^2(\mathrm{S} \times \mathfrak{L}))}
\newcommand{\algB}{\mathcal{B}}
\newcommand{\algC}{\mathcal{C}}
\newcommand{\algE}{\mathcal{E}}
\newcommand{\algF}{\mathcal{F}}
\newcommand{\Fock}{\mathcal{F}(\mathcal{H})}
\newcommand{\Fockbos}{\mathcal{F}_{+}(\mathcal{H})}
\newcommand{\Fockferm}{\mathcal{F}_{-}(\mathcal{H})}
\newcommand{\algG}{\mathcal{G}}
\newcommand{\algI}{\mathcal{I}}
\newcommand{\algK}{\mathcal{K}}
\newcommand{\algKeF}{\mathcal{K}^{\epsilon, \mathbf{F}}}
\newcommand{\algM}{\mathcal{M}}
\newcommand{\algO}{\mathcal{O}}
\newcommand{\algP}{\mathcal{P}}
\newcommand{\algR}{\mathcal{R}}
\newcommand{\algS}{\mathcal{S}}
\newcommand{\algT}{\mathcal{T}}
\newcommand{\Nu}{\mathcal{V}}
\newcommand{\Nud}{\mathcal{V}_{(2)}}
\newcommand{\algX}{\mathcal{X}}
\newcommand{\algW}{\mathcal{W}}
\newcommand{\algWF}{\mathcal{W}^{\mathbf{F}}}
\newcommand{\algZ}{\mathcal{Z}}
\newcommand{\modulo}[1]{\left| #1 \right|}
\newcommand{\normprod}[1]{\, \, : \! \! #1 \! \! : \, \,}
\newcommand{\floor}[1]{\left\lfloor #1 \right\rfloor}

\newcommand{\mRED}{\mathrm{RED}}
\newcommand{\mBOG}{\mathrm{BOG}}
\newcommand{\mca}{\mtext{ca}}
\newcommand{\mGm}{\mtext{Gm}}
\newcommand{\mInt}{\mathrm{Int}}
\newcommand{\mKMS}{\mathrm{KMS}}
\newcommand{\mls}{\mtext{ls}}
\newcommand{\mNint}{\mathrm{Nint}}
\newcommand{\mNintxsl}{\mathrm{Nint}\leftrightpar{\toroXSL}}
\newcommand{\mot}{\mtext{ot}}
\newcommand{\mRCA}{\mtext{RCA}}
\newcommand{\mRCC}{\mtext{RCC}}
\newcommand{\mtr}{\mtext{tr}}
\newcommand{\conjn}[1]{\mathbb{N}(#1)}
\newcommand{\conjnz}[1]{\mathbb{N}_0(#1)}
\newcommand{\bmalpha}{\bm{\alpha}}
\newcommand{\bma}{\bm{a}}
\newcommand{\bmb}{\bm{b}}
\newcommand{\bmc}{\bm{c}}
\newcommand{\bmd}{\bm{d}}
\newcommand{\bml}{\bm{l}}
\newcommand{\bms}{\bm{s}}
\newcommand{\bmum}{\bm{1}}

\newcommand{\rmb}{\mathrm{b}}
\newcommand{\ormb}{\overline{\mathrm{b}}}
\newcommand{\rmBar}{\mathrm{Bar}}
\newcommand{\rmCo}{\mathrm{Co}}

\newcommand{\rmeq}{\, \mathrm{eq}}

\newcommand{\rmfis}{\mathrm{f\text{\'i}s}}
\newcommand{\rmefis}{\text{\emph{f\'is}}}

\newcommand{\rmFerm}{\mathrm{Ferm}}
\newcommand{\rmGr}{\mathrm{Gr}}
\newcommand{\rmIm}{\mathrm{Im}}
\newcommand{\Ibase}{\mathrm{I}}
\newcommand{\Ibasek}[1]{\mathrm{I}_{#1}}
\newcommand{\rmm}{\mathrm{m}}
\newcommand{\rmmod}{\, \mathrm{mod} \,}

\newcommand{\rmPC}{\mathrm{P}_{\mathbb{C}}}
\newcommand{\rmrem}{\mathrm{res}}
\newcommand{\rmRe}{\mathrm{Re}}

\newcommand{\rmSt}{\mathrm{St}}
\newcommand{\rmtr}{\mathrm{tr}}

\newcommand{\rmu}{\mathrm{u}}

\newcommand{\ormu}{\overline{\mathrm{u}}}
\newcommand{\rmW}{\mathrm{W}}
\newcommand{\doisaeleinf}{2^{\fL}_{<\infty}}
\newcommand{\doisquad}{\quad \quad}
\newcommand{\tresquad}{\quad \quad \quad}
\newcommand{\grassderiva}[2]{\frac{\delta #1}{\delta #2}}
\newcommand{\intd}{\! \mathrm{d}}
\newcommand{\intgrassmann}[2]{\int_{#1} \! \mathrm{d}\mathcal{H}_{#2}}
\newcommand{\intgaussmann}[3]{\int_{#1} \! \mathrm{d}\mu_{#3}(\mathcal{H}_{#2})}
\newcommand{\inttriangle}{\triangle }
\newcommand{\kappaop}[2]{\kappa^{#1}_{#2}}
\newcommand{\expo}[1]{\exp \! \left( #1 \right)}
\newcommand{\BK}{\text{\emph{B-K}}}

\newcommand{\kboltz}{\mathrm{k}_{\mathrm{B}}}
\newcommand{\tsfA}{\text{\textsf{A}}}
\newcommand{\spinpracima}{\uparrow}
\newcommand{\spinprabaixo}{\downarrow}
\newcommand{\tilX}{\til{X}}
\newcommand{\nquad}{\! \! \! \! \! \!}
\newcommand{\doisnquad}{\! \! \! \! \! \! \! \! \! \!}
\newcommand{\fechab}[2]{\left[#1, #2 \right)}

%%%%%%%%%%%%%%%%%%%%%%%%%%%%%%%%%%%%%%%%%%%%%%%%%%%%%%%%%%%%%%%%%%%%%%%%%

%%%%%% TITLE, ABSTRACT

\title{The weakly interacting tenfold way}

\author{Lucas C.P.A.M. M\"{u}ssnich$^1$, Renato V. Vieira$^2$}

\address{$^1$ Instituto de Ciências Matemáticas e de Computação, Universidade de São Paulo (ICMC-USP), São Carlos, SP, Brasil}
\address{
$^2$ Instituto de Matemática, Estatística e Computação Científica, Universidade Estadual de Campinas (IMECC-UNICAMP), Campinas, SP, Brasil. }

\begin{abstract}

The tenfold way is a classification scheme for the building 
blocks of free fermion systems. More precisely, 
it classifies the isomorphism classes of spaces of equivariant free Hamiltonians 
in irreducible fermion systems with symmetries. 
This classification scheme naturally leads to the K-theoretical classification
of topological phases of matter, known as the periodic table of topological 
insulators and superconductors. Topological K-theory is represented by spectra $KU$ and $KO$, and in this article we present realizations of these spectra in terms of time evolution operators of irreducible free fermion systems with symmetries, with explicit 
formulas for the structural suspension maps. 
We introduce a geometric definition of the space of weakly interacting time evolution operators, as
the complement of the cut locus of the subspace of free operators. Our main result is that spectra $KU^{wi}$ and $KO^{wi}$ of weakly interacting time evolution operators deformation retract to $KU$ and $KO$. We thus have a stable homotopy theoretical proof that the tenfold way is stable to weak interactions.
\end{abstract}

\maketitle

\tableofcontents

%%%%%%%%%%%%%%%%%%%%%%%%%%%%%%%%%%%%%%%%%%%%%%%%%%%%%%%%%%%%%%%%%%%%%%%%%

%%%%%% INCLUDE FILES

%%%%%%%%%%%%%%%%%
%%% S % S % S %%%
%%%%%%%%%%%%%%%%%

\section{Introduction}

The tenfold way is a classification scheme for the building 
blocks of non-interacting fermion systems. More precisely, 
it classifies the isomorphism classes of spaces of equivariant free Hamiltonians 
in irreducible fermion systems with symmetries. 
This classification scheme naturally leads to the K-theoretical classification
of topological phases of matter, known as the periodic table of topological 
insulators and superconductors. 

In this article we show how the topological K-theory spectra $KU$ and $KO$ have realizations in terms of time evolution operators of irreducible free fermion systems with symmetries, with explicit 
formulas for the structural suspension maps. 
We further give a geometric definition of weakly interacting time evolution operators. We show how associated spectra $KU^{wi}$ and $KO^{wi}$ of weakly interacting time evolution operators deformation retract to $KU$ and $KO$, which means they represent the same cohomology theories. Thus, we provide a stable homotopy theoretical proof that the tenfold way is stable to weak interactions.

One perspective on the tenfold way is to consider that fermion state spaces 
can be modeled by representations of Clifford algebras, of which there are 
exactly 10 Morita equivalence classes. Another perspective is that, due to 
Schur's lemma, the automorphisms of irreducible supergroup representations 
form associative super division algebras, of which there are exactly ten. 
These points of view are connected by the fact that every Clifford algebra is Morita 
equivalent to a super division algebra \cite{baez2020tenfold,freed2013twisted}.

Another perspective on the tenfold way is in terms of the 10 infinite families 
of compact symmetric spaces classified by Cartan \cite{cartan1926classe,cartan1927classe}. 
In \cite{heinzner2005symmetry} Heinzner, Huckleberry and Zirnbauer extended Dyson's 
threefold way \cite{dyson1962threefold} by classifying fermion systems with equivariant Hamiltonians 
that are  quadratic on creation and annihilation operators. There they showed that the 
systems' spaces of time evolution operators are compact symmetric spaces, and 
that all 10 classes classified by Cartan can be obtained this way. 
In \cite{agarwala2017tenfold} Agarwala, Haldar and Shenoy showed that all 
classes can be obtained by the subspace of free Hamiltonians. 
They further explicitly lay out the structure of equivariant Hamiltonians that preserve particle number, both for 
free systems and interacting ones.

The role of topological K-theory in the tenfold way was made explicit by Kitaev 
in \cite{kitaev2009periodic}. There he shows that by deforming gapped Hamiltonians 
by a process of spectral flattening you can obtain classifying spaces for free fermion phases, 
which are all compact symmetric spaces. He also points out that if you 
have a family of irreducible systems parametrized by a space $\Lambda$, then its topological 
phases are classified by homotopy classes of maps from $\Lambda$ to one of the classifying spaces, 
which are in bijection with difference classes of equivariant vector bundles over $\Lambda$ 
that model free fermion ground states. These homotopy classes form groups isomorphic to 
some topological K-group of $\Lambda$, the latter depending on the internal symmetries 
of the system and on the dimension of $\Lambda$.

From the point of view of stable homotopy theory, the connection between the classification of free fermion systems and topological K-theory arises from the fact that the symmetric spaces of free time evolution operators form the underlying spaces 
of spectra $KU$ and $KO$. These spectra represent, in the sense of Brown's representability 
theorem, complex and real topological K-theory \cite{adams1974stable}. This means the K-groups of $\Lambda$ are isomorphic to stable homotopy groups of mapping spectra $KU^\Lambda$ and $KO^\Lambda$.

As Freed and Moore 
explain in \cite{freed2013twisted}, topological phases of crystalline systems are classified by twisted equivariant K-theory \cite{may2006parametrized}. A $d$-dimensional crystal is a solid material whose constituents (atoms, molecules, or ions) are distributed over the translation lattice of a crystallographic group. A crystallographic group is a discrete subgroup $\mathcal C\subset \text{Iso}(\mathds R^d)$ of isometries of Euclidean space containing a normal subgroup $\mathcal N$ of translations that is a lattice, and such that the stabilizer of the origin $\mathcal P\coloneqq \text{Stab}_{\mathcal C}(\vec 0)\cong \frac{\mathcal C}{\mathcal N}$, called the point group, is contained in $O(d)$. The Brillouin zone of the crystal is the momentum space modeled by the Pontryagin dual $\widehat{\mathcal N}\coloneqq\texttt{Grp}(\mathcal N,\mathds S^1)$, which is naturally equipped with a $\mathcal P$-space structure. It is common, eg \cite{cornfeld2021tenfold}, to consider crystallographic groups over the lattice $\mathds Z^d$, whose Brillouin zone is the torus $\mathds T^d$. In a crystalline fermion system we further assume each constituent has an internal symmetry group $G$ generated by chiral, charge-conjugation or time-reversal  symmetries. If we assume the internal and crystallographic symmetries commute then the symmetry group of the system is $\mathcal C\times G$, and if they don't commute it is a semi-direct product $\mathcal C\ltimes^\alpha G$. For a compact Lie group $P$, the $P$-equivariant cohomology theories are represented by $P$-spectra, and there are $P$-spectra $KU_P$ and $KO_P$ that represent complex and real $P$-equivariant K-theory \cite{LMS,may1996equivariant}. In this context topological phases are classified by twisted equivariant K-groups, which are composed of  sections of $KU_{\mathcal P}$-bundles or $KO_{\mathcal P}$-bundles over $\widehat{\mathcal N}$, with the bundle structure induced by the twisting $\alpha$. Stable homotopy theory provides an 
efficient computational tool for topological classification problems, for instance 
in \cite{cornfeld2021tenfold} the formalism of equivariant ring spectra is used to 
obtain the complete classification of topologically distinct quantum states of 3D 
crystalline topological insulators and superconductors for key space-groups, by 
considering fermion systems parametrized by the Brillouin zone torus.

In \cite{kitaev2009periodic} Kitaev points out that the tenfold way should be stable to weak 
interactions, and further gives an example of how two topological phases distinguished 
by K-theoretical invariants can be connected by a continuous path 
of interacting phases. Thus, though sufficiently strong interactions may break the 
classification scheme of the tenfold way, weakly interacting fermion systems should 
still be classified by topological K-theory. In this article we give a stable homotopy theoretical proof of this stability to weak interactions, within the geometric definition of weak interaction here given. In the full interacting regime topological 
phases are classified by cobordism, represented by Thom's bordism spectra \cite{freed2021reflection,Kapustin_2015}.

In order to prove the stability of the tenfold way to weak interactions, we start by reviewing the structure of the spaces of free time evolution operators underlying $KU$ and $KO$, as shown in \cite{agarwala2017tenfold,heinzner2005symmetry}. Using this construction we provide explicit formulas for the suspension maps, which are determined by Cartan embeddings, and by the 
homotopy equivalences Bott used in his periodicity theorem \cite{bott1959stable}. 
These suspension maps are closely related to the diagonal map Kennedy and 
Zirnbauer used in their  homotopy-theoretic proof of the periodic table \cite{kennedy2015bott,kennedy2016bott}. We further give a geometric definition of weakly interacting time evolution operators in terms of the complement of the cut locus of the submanifold of free operators within 
the full interacting space. This allows us to define spectra $KU^{wi}$ and $KO^{wi}$ composed of weakly interacting operators, with our weak interaction  condition guaranteeing that these spectra deformation retract to $KU$ and $KO$. Since the cut locus is closed our definition of weak interaction is stable to small perturbations. Thus the weakly interacting regime reproduces the classifying table of topological insulators and superconductors (see table \ref{tab:ClassifKT}).

\begin{table}[ht]
\caption{Each row represents a universal symmetry class, identified
by a Cartan-Altland-Zirnbauer (CAZ) label. The symmetry class may present chiral $\hat S$, charge-conjugation $\hat C$ or time-reversal $\hat T$ symmetries. The signature $(\epsilon_S\epsilon_C\epsilon_T)$ identifies the absence of the associated symmetry when the entry is $0$, or its square when the entry is $\pm 1$. The third column lists the spaces $\mathscr M^{\epsilon_S\epsilon_C\epsilon_T}_{wi}$ of equivariant weakly interacting time evolution operators. The fourth column lists the associated Cartan compact symmetric spaces. The last columns indicate the stable homotopy groups $\pi^S_n\mathscr M^{\epsilon_S\epsilon_C\epsilon_T}_{wi}$.}
\label{tab:ClassifKT}
\centering
    \begin{tabular}{|c|c|c|c|c|c|c|c|c|c|c|c|}\hline
     $\bullet$ &CAZ&$KU^{wi}_\bullet$&$KU_\bullet$&$\pi_0^S$&$\pi^S_1$&$\pi^S_2$&$\pi^S_3$&$\pi_4^S$&$\pi^S_5$&$\pi^S_6$&$\pi^S_7$\\\hline\hline
        0& AIII&$\mathscr M^{100}_{wi}$& $BU(\infty)$&$\mathds Z$&$0$&$\mathds Z$&$0$&$\mathds Z$&$0$&$\mathds Z$&$0$\\\hline
        1& A&$\mathscr M^{000}_{wi}$& $U(\infty)$&$0$&$\mathds Z$&$0$&$\mathds Z$&$0$&$\mathds Z$&$0$&$\mathds Z$\\\hline
    \end{tabular}
    
    \vspace{0.5cm}
    
    \begin{tabular}{|c|c|c|c|c|c|c|c|c|c|c|c|}\hline
        $\bullet$ &CAZ&$KO^{wi}_\bullet$&$KO_\bullet$&$\pi_0^S$&$\pi^S_1$&$\pi^S_2$&$\pi^S_3$&$\pi_4^S$&$\pi^S_5$&$\pi^S_6$&$\pi^S_7$\\\hline\hline
         0&BDI&$\mathscr M^{111}_{wi}$&$ BO(\infty)$&$\mathds Z$&$0$&$0$&$0$&$\mathds Z$&$0$&$\mathds Z_2$&$\mathds Z_2$\\\hline
         1&AI&$\mathscr M^{001}_{wi}$&$\frac{U(\infty)}{O(\infty)}$&$0$&$0$&$0$&$\mathds Z$&$0$&$\mathds Z_2$&$\mathds Z_2$&$\mathds Z$\\\hline
         2&CI&$\mathscr M^{1-11}_{wi}$&$\frac{Sp(\sfrac{\infty}{2})}{U(\sfrac{\infty}{2})}$&$0$&$0$&$\mathds Z$&$0$&$\mathds Z_2$&$\mathds Z_2$&$\mathds Z$&$0$\\\hline
         3&C&$\mathscr M^{0-10}_{wi}$&$ Sp(\sfrac{\infty}{2})$&$0$&$\mathds Z$&$0$&$\mathds Z_2$&$\mathds Z_2$&$\mathds Z$&$0$&$0$\\\hline
         4&CII&$\mathscr M^{1-1-1}_{wi}$&$ BSp(\infty)$&$\mathds Z$&$0$&$\mathds Z_2$&$\mathds Z_2$&$\mathds Z$&$0$&$0$&$0$\\\hline
         5&AII&$\mathscr M^{00-1}_{wi}$&$\frac{U(\infty)}{Sp(\sfrac{\infty}{2})}$&$0$&$\mathds Z_2$&$\mathds Z_2$&$\mathds Z$&$0$&$0$&$0$&$\mathds Z$\\\hline
         6&DIII&$\mathscr M^{11-1}_{wi}$&$\frac{O(\infty)}{U(\sfrac{\infty}{2})}$&$\mathds Z_2$&$\mathds Z_2$&$\mathds Z$&$0$&$0$&$0$&$\mathds Z$&$0$\\\hline
         7&D&$\mathscr M^{010}_{wi}$&$ O(\infty)$&$\mathds Z_2$&$\mathds Z$&$0$&$0$&$0$&$\mathds Z$&$0$&$\mathds Z_2$\\\hline 
    \end{tabular}
\end{table}

\textbf{Structure of the article.}
In section \ref{sec2} we review the classification of irreducible free fermion systems with symmetries by compact symmetric spaces. Following \cite{heinzner2005symmetry} we consider the Nambu space model for fermions, and as in \cite{agarwala2017tenfold} we present convenient choice of basis to describe the structure of the symmetries and free equivariant Hamiltonians. We further show how the symmetries determine the symmetric structure of the spaces of time evolution operators, and give their structure in the chosen convenient basis for each symmetry class.

In section \ref{sec3} we review the construction of the topological K-theory functors, and the definition of the spectra that represent cohomology theories. We then construct spectra $KU$ and $KO$ that represent complex and real topological K-theories in terms of free time evolution operators, including explicit formulas for the structural suspension maps.

In section \ref{sec4} we give our geometric definition of weakly interacting time evolution operators, and show how these form spectra $KU^{wi}$ and $KO^{wi}$ that deformation retract onto the subspectra of free operators.

In section \ref{sec5} we make some remarks concerning possible extensions of our work, in particular regarding the classification of crystalline fermion systems by twisted equivariant K-theory, and the classification of interacting systems by cobordism.\\

\textbf{Notation and terminology.}
For $A=\mqty[a_{ij}]\in M_{N}\mathds C$ we denote its conjugate by $A^*=\mqty[a_{ij}^*]$, its transpose by $A^t=\mqty[a_{ji}]$, and its conjugate transpose by $A^\dagger=A^{*t}=A^{t*}=\mqty[a_{ji}^*]$.
We will denote by $K$ the anti-linear complex conjugation operator
$
K:\mathds C^N\to \mathds C^N$, $Kv=\mqty[ v_i^*]$. For $A\in M_{N}\mathds C$ we have $KAK=A^*$.
For $p,q\in\mathds N$ we set the matrix  $1_{p,q}=\begin{bsmallmatrix}1_p&0\\0&-1_q\end{bsmallmatrix}$. For $N$ even we set $J_N=\begin{bsmallmatrix}0&1_{\sfrac{N}{2}}\\-1_{\sfrac{N}{2}}&0\end{bsmallmatrix}$ and $F_N=\begin{bsmallmatrix}0&1_{\sfrac{N}{2}}\\1_{\sfrac{N}{2}}&0\end{bsmallmatrix}$. We note that the compact symplectic group $Sp(\sfrac{N}{2})$ is composed of the unitary matrices $U\in U(N)$ such that $J_NUJ_N^\dagger=U^*$.

The smash product of pointed spaces $X,Y\in \texttt{Top}_*$ is $X\wedge Y\coloneqq \frac{X\times Y}{X\times \{y_0\}\cup \{x_0\}\times Y}$. We denote the standard interval by $I=[0,1]$, the pointed circle  by $\mathds S^1\coloneqq I_{/0\sim 1}\in \texttt{Top}_*$, with the identified extremities as base point, and the $n$-th sphere by $\mathds S^n\coloneqq (\mathds S^1)^{\wedge n}$.  We also set $I_+\coloneqq I\sqcup\{x_0\}\in \texttt{Top}_*$ the space obtained by adding a disjoint base point. A pointed map $h:X\wedge I_+\to Y$ is a base point preserving homotopy, and a pointed map $\sigma:X\wedge \mathds S^1\to Y$ is a base point preserving homotopy from the constant map to itself. The $n$-th loop space of a pointed space $X$ is $\Omega^n X\coloneqq X^{\mathds S^n}$, and its set of connected components form the $n$-th homotopy group $\pi_n X\coloneqq \pi_0\Omega^nX$.

Let $\mathcal G$ be a Lie group, and $\mathfrak g$ be its associated Lie algebra. We will denote the group level commutator by $[g,h]=ghg^{-1}h^{-1}$ for $g,h\in \mathcal G$, and the algebra level commutator by $\{A,B\}=AB-BA$ for $A,B\in \mathfrak g$.\\

\textbf{Acknowledgements.}
Renato V. Vieira was financed by Conselho Nacional de Desenvolvimento Científico e Tecnológico – Brasil (CNPq), grant no. 150669/2024-0.

Lucas C.P.A.M. M\"ussnich was financed by Coordenação de Aperfeiçoamento de Pessoal de Nível Superior - Brasil (CAPES).

The authors would like to thank N. Javier Buitrago Aza for fruitful discussions during the initial stages of this work.
%%%%%%%%%%%%%%%%%
%%% S % S % S %%%
%%%%%%%%%%%%%%%%%

\section{Classification of free fermion systems with symmetries}\label{sec2}

\subsection{Nambu space model for fermions}

    The following construction is standard in the mathematical physics literature. See, for instance, \cite{araki-AMS, book-bru-pedra,bratteli-dois}. We adopt the Clifford-algebraic framework to maintain consistency with the main references here used \cite{heinzner2005symmetry,kitaev2009periodic}. For a complete discussion on the mathematical aspects of quantization, see, for example,
    \cite{book-derezinski}.
    
    The state space of a single fermion is modeled by a Hilbert space.
    We shall only consider the finite-dimensional case, which already presents complexities
    regarding the classification of symmetries.
    Given $N\in\mathds N$, let $\mathscr V_N$ be a Hilbert space with dimension $N$.
    The Fock space of many-fermion states is the exterior algebra 
    \begin{align*}\textstyle
        \bigwedge \mathscr V_N \coloneqq
        \bigoplus_{n=0}^N \mathscr V_N^{\wedge n},
    \end{align*}
    where each 
    $\mathscr V_N^{\wedge n}$ is the space of $n$-particle anti-symmetrized states
    with hermitian structure 
    $$
        \braket{\wedge_i u_i}{\wedge_jv_j}_{\mathscr V_N^{\wedge n}}\coloneqq 
        \det\mqty[\braket{u_i}{v_j}_{\mathscr V_N}]
        .
    $$ 
    Recall that $\mathscr V_N^{\wedge 1} \cong \mathscr V_N$, and that 
    the direct sum goes up to $N$ due to anti-symmetry and finite dimension. The subspace $\mathscr V_N^{\wedge 0}=\mathds C$ is associated with
    the \emph{vacuum state}. We set $\ket{0}\coloneqq(1,0,0,\dots)\in \mathscr V_N^{\wedge 0}$,
    and interpret it as the zero-particle state. For each 
    $v\in \mathscr V_N$ we associate \emph{creation} and \emph{annihilation} operators,
    that allow the construction of $n$-particle states from the vacuum. These are, 
    respectively, $\varepsilon(v)\in\mathfrak B(\bigwedge\mathscr V_N)$ (degree $1$)
    and $\iota(v^\dagger)\in \mathfrak B(\bigwedge\mathscr V_N)$ (degree $-1$), given by 
    \begin{align*}
        \varepsilon(v)(u_1\wedge\cdots \wedge u_n) & \coloneqq 
        v\wedge u_1\wedge\cdots \wedge u_n, \\
        \iota(v^\dagger)(u_1\wedge\cdots \wedge u_n) &\textstyle \coloneqq
        \sum_{i=1}^n (-1)^{i+1}
        v^\dagger(u_j)\,u_1\wedge\cdots \wedge u_{j-1}\wedge u_{j+1}\wedge\cdots u_n,
    \end{align*}
    where $v^\dagger\in \mathscr V_N^*$ is the dual of $v$ via the Fr\'echet-Riesz representation theorem. Denote by 
    $\mathcal B=\{\ket{i}\}_{1\leq i\leq N}$ an orthonormal basis for $\mathscr V_N$, and 
    and by $\mathcal B^*=\{\bra{i}\}_{1\leq i\leq N}$ its dual base. We set
    $$
        a_i^\dagger\coloneqq \varepsilon(\ket{i}),\qquad 
        a_i\coloneqq \iota(\bra{i}),
    $$
    so that 
    $$
        a_i^\dagger\ket{0}=\ket{i}
        ,\qquad
        a_i\ket{j}=\delta_{ij}\ket{0}
        .
    $$
    These operators satisfy the  
    \emph{canonical anti-commutation relations} (CAR),
    characteristic of fermion systems:
    $$
        a_i^\dagger a_j^\dagger+a_j^\dagger a_i^\dagger=0
        ,\qquad
        a_i a_j+a_j a_i=0
        ,\qquad
        a_i^\dagger a_j+a_j a_i^\dagger=\delta_{ij}
        .
    $$

\begin{definition}
    The \textit{Nambu space associated with $\mathscr V_N$} is the Hilbert space $\mathscr W_N\coloneqq \mathscr V_N\oplus \mathscr V_N^*$ equipped with the symmetric bilinear form
$$
    b: \mathscr W_N\times \mathscr W_N\to \mathds C
    ,\qquad b(v_1+f_1,v_2+f_2)\coloneqq f_1v_2+f_2v_1.
$$
\end{definition}

The Nambu space embeds into $\mathfrak B(\bigwedge\mathscr V_N)$ via the creation and annihilation operators. Elements of $\mathscr W_N$ are called fermion field operators. The bilinear structure encodes the CAR, in the sense that for all $\Psi_1,\Psi_2\in \mathscr W$ we have
$$
    \Psi_1\Psi_2+\Psi_2\Psi_1=b(\Psi_1,\Psi_2)\hat 1
$$
as operators on $\bigwedge \mathscr V_N$. The associative algebra generated by $\mathscr W_N$ in $\mathfrak B(\bigwedge \mathscr V_N)$ is isomorphic to the Clifford algebra
$$\textstyle
    Cl(\mathscr W_N,b)
            \coloneqq\bigoplus_{n=0}^{2N}
            \mathscr W_N^{\otimes n}/_{\Psi_1\otimes\Psi_2+\Psi_2\otimes\Psi_1-b(\Psi_1,\Psi_2)}.
$$
We have a natural filtered algebra structure
$Cl(\mathscr W_N,b)=\bigoplus_{n=0}^{2N}Cl^n(\mathscr W_N,b)$, with 
$$Cl^n(\mathscr W_N,b)\cong \mathscr W_N^{\wedge n}.$$

\begin{definition}
    Let $(\mathscr W_N,b)$ be a Nambu space. The \textit{space of free Hamiltonians} is
$$\textstyle
\mathscr H_N\coloneqq \{\hat H=\sum_{ij}H_{ij}a_i^\dagger a_j\mid H\in M_N\mathds C,\ H^\dagger=H\}\subset Cl^2(\mathscr W_N,b).
$$

The adjoint action of $\mathscr H_N$ on $Cl(\mathscr W_N,b)$ is 
$$
ad_{\hat H}\Psi\coloneqq \{\hat H,\Psi\}.
$$

\end{definition}

The Hilbert space $\mathscr V_N\subset Cl(\mathscr W_N,b)$ is closed under this action, with 
$$\textstyle
    ad_{\hat H} v=\sum_{ij}H_{ij}\braket{j}{v}\ket{i}
    ,\qquad [ad_{\hat H}]_{\mathcal B}=H.
$$
Thus  the adjoint action of free Hamiltonians is faithfully represented in $\mathscr V_N$.

In the next section we will describe the symmetries of fermion systems at the Nambu space level, so we will have to consider the Hamiltonian adjoint actions on $\mathscr W_N$, represented in the basis $\widetilde{\mathcal B}=\mathcal B\sqcup \mathcal B^*$ by
$$
    [ad_{\hat H}]_{\widetilde{\mathcal B}}=
    \mqty[H&0\\0&-H^*].
$$

The exponential map associates
to each free Hamiltonian $\hat H\in \mathscr H_N$ a unitary operator $\exp(-i\hat H)$ in $Cl(\mathscr W_N,b)$, with associated adjoint action 
$$
    Ad_{\exp(-i\hat H)}\Psi\coloneqq\exp(-i\hat H)\Psi \exp(i\hat H).
$$

\begin{definition}
    The \textit{space of free time evolution operators}
    on the Nambu space $(\mathscr W_N,b)$ is
    $$
        \mathscr M_N
        \coloneqq \{Ad_{\exp(-i\hat H)}\in Aut(Cl(\mathscr W_N,b))\mid \hat H\in \mathscr H_N\}.
    $$
\end{definition}

The Hilbert space $\mathscr V_N$ is closed under the action of free time evolution operators, with
$$
    [Ad_{\exp(-i\hat H)}]_{\mathcal B}=\exp(-iH).
$$

At the Nambu space level we have 
$$
    [Ad_{\exp(-i\hat H)}]_{\widetilde{\mathcal B}}=\mqty[\exp(-iH)&0\\0&\exp(iH^*)].
$$ 

Since $i\mathscr H_N\cong \mathfrak u(N)$ we have $\mathscr M_N\cong U(N)$.

\subsection{Nambu space symmetries}

Let $\hat U:\mathscr W_N\to\mathscr W_N$ be a  bijection. We say $\hat U$ is a \textit{linear} Nambu space symmetry if it is a unitary operator such that $b(\hat U\Psi_1,\hat U\Psi_2)=b(\Psi_1,\Psi_2)$, and that it is an \textit{anti-linear} Nambu space symmetry if it is an anti-unitary operator such that $b(\hat U\Psi_1,\hat U\Psi_2)=\overline{b(\Psi_1,\Psi_2)}$. A Nambu space symmetry $\hat U$, unitary or anti-unitary, is  \textit{usual} if it preserves the Nambu space decomposition, ie if $\hat U(\mathscr V_N)=\mathscr V_N$ and $\hat U(\mathscr V_N^*)=\mathscr V_N^*$. A symmetry is transposing if $\hat U(\mathscr V_N)=\mathscr V_N^*$ and $\hat U(\mathscr V_N^*)=\mathscr V_N$. 

% Compatibility with the bilinear form $b$ implies that for any $\hat U\in Aut(\mathscr W,b)$, $f\in \mathscr V^*$ and $v\in \mathscr V$ we have $\hat Uf(v)=b(\hat Uf,v)=b(f,\hat U^\dagger v)=f(\hat U^\dagger v)$, which means $\hat U\mid_{\mathscr V^*}=(\hat U^\dagger)^t$.

% Every Nambu space comes naturally equipped with a transposing anti-unitary involution operator $\widetilde C$ defined as 
% $$
%     \widetilde C(v+f)=f^\dagger+v^\dagger.
% $$
% Under the basis $\widetilde{\mathcal B}$ we have the expression $\widetilde C(\sum_i v_i\ket{i}+f_i\bra{i})=\sum_i f_i^* \ket{i}+v_i^*\bra{i}$. This symmetry relates the unitary and bilinear structure of $\mathscr W$ by $\braket{\Psi_1}{\Psi_2}=b(\widetilde C\Psi_1,\Psi_2)$. The compatibility of the symmetries with the bilinear form $b$ implies they to commute with $\widetilde C$. In the unitary cases this follows from 
% $$
%     \braket{\widetilde C\hat U\Psi_1}{\Psi_2}=b(\hat Uv,w)=b(v,\hat U^\dagger w)=\braket{\widetilde C\Psi_1}{\hat U^\dagger\Psi_2}=\braket{\hat U\widetilde C\Psi_1}{\Psi_2},
% $$
% and in the anti-unitary cases this follows by an analogous argument. Similarly, comutation with $\widetilde C$ implies compatibility with $b$.

The set of all Nambu space symmetries forms a group $\mathcal U$, which decomposes as a disjoint union 
$$
    \mathcal U=\mathcal U_{UL}\sqcup \mathcal U_{TA}\sqcup \mathcal U_{TL}\sqcup \mathcal U_{UA}.
$$

The set $\mathcal U_{UL}$ of \textit{regular} symmetries is  composed of the usual linear symmetries, which forms a normal subgroup of $\mathcal U$. For all $\hat U\in \mathcal U_{UL}$ compatibility with the bilinear form $b$ means that 
$$
    \hat Uf(v)=b(\hat Uf,v)=b(f,\hat U^\dagger v)=f(\hat U^\dagger v),
$$
so 
$$
    \hat U\mid_{\mathscr V_N^*}=(\hat U\mid_{\mathscr V_N}^\dagger)^t.
$$
In the basis $\widetilde{\mathcal B}$ a regular symmetry is then represented by a matrix
$$
    [\hat U]_{\widetilde{\mathcal B}}=\mqty[U&0\\0&U^*],
$$
for some $U\in U(N)$.\\

The set $\mathcal U_{UA}$ of \textit{time-reversal} symmetries is composed of the usual anti-linear symmetries. Since anti-unitary operators are unitary operators composed with the complex conjugation, the time-reversal symmetries are represented by 
$$
    [\hat T]_{\widetilde B}=\mqty[T&0\\0&T^*]K,
$$
for some $T\in U(N)$.\\

The set $\mathcal U_{TL}$ of \textit{charge-conjugation} symmetries is composed of the transposing linear symmetries. For all $\hat C\in \mathcal U_{UA}$ 
compatibility with the bilinear form $b$ now means that 
$$
    g(\hat Cf)=b(\hat Cf,g)=b(f,\hat C^\dagger g)=f(\hat C^\dagger g),
$$ 
so again we have $\hat C\mid_{\mathscr V_N^*}=(\hat C\mid_{\mathscr V_N}^\dagger)^t$.  Thus a charge-conjugation symmetry is represented by 
$$
    [\hat C]_{\widetilde B}=\mqty[0&C\\C^*&0],
$$ 
for some $C\in U(N)$.\\

The set $\mathcal U_{TA}$ of \textit{sublattice}, or \textit{chiral}, symmetries is composed of the transposing anti-linear symmetries. A sublattice symmetry $\hat S\in \mathcal U_{TA}$ is represented by 
$$
    [\hat S]_{\widetilde B}
    =\mqty[0&S\\S^*&0]K,
$$
for some $S\in U(N)$.

Any non-regular symmetry squares to a regular one, ie if $\hat U\in \mathcal U_{TA}\sqcup \mathcal U_{TL}\sqcup \mathcal U_{UA}$ then $\hat U^2\in \mathcal U_{UL}$. Also, the product of two distinct type of non-regular symmetries is of the third type. This means the quotient $\mathcal U/_{\mathcal U_{UL}}$ is isomorphic to the Klein 4-group $\mathds Z_2^2$.

\begin{definition}
    A \textit{Nambu space with symmetries} is a representation $\rho:G\to \mathcal U$ of a compact Lie group $G$.
\end{definition}

For any Nambu space with symmetries we have an induced decomposition 
$$
    G=G_{UL}\sqcup G_{TA}\sqcup G_{TL}\sqcup G_{UA}.
$$

The $G$-action on $\mathscr W_N$ extends to an action by automorphism of $Cl(\mathscr W_N,b)$ by 
$$
    \hat U(\wedge_i\Psi_i)=\wedge_i\hat U\Psi_i.
$$

\subsection{Reduction to grotesque fermion systems}

In a Nambu space with symmetries, the subgroup $G_{UL}$ acts unitarily on the Hilbert space $\mathscr V_N$. 
Let $\widehat G_{UL}$ be the set of isomorphism classes of irreducible unitary $G_{UL}$-representations. For each class $\lambda\in \widehat G_{UL}$ fix a representative $\rho_\lambda:G_{UL}\to U(\mathscr R_\lambda)$, and set $d_\lambda\coloneqq \dim \mathscr R_\lambda$.  Setting  $\mathscr E_\lambda\coloneqq \text{Hom}_{G_{UL}}(\mathscr R_\lambda,\mathscr V_N)$ and $\mathscr E_\lambda^*\coloneqq \text{Hom}_{G_{UL}}(\mathscr R_\lambda^*,\mathscr V^*_N)$,
on each $(\mathscr E_\lambda\otimes \mathscr R_\lambda)\oplus(\mathscr E_\lambda^*\otimes \mathscr R_\lambda^*)$ we have 
the natural $G_{UL}$-action
$$
g(h\otimes r+f\otimes t)\coloneqq h\otimes \rho_\lambda(g) r+f\otimes \rho_\lambda(g)t.
$$

The map
$$\textstyle
\Phi:\bigoplus_{\lambda\in\widehat{G}_{UL}}(\mathscr E_\lambda\otimes \mathscr R_\lambda)\oplus(\mathscr E_\lambda^*\otimes \mathscr R_\lambda^*)\to\mathscr W_N,
\quad 
\Phi(\bigoplus_\lambda h_\lambda\otimes r_\lambda+f_\lambda\otimes t_\lambda)
\coloneqq
\sum_\lambda h_\lambda r_\lambda+f_\lambda t_\lambda
$$
is a $G_{UL}$-equivariant isomorphism. We can define an inner product on $\mathscr E_\lambda$ by
$$
    \braket{h_1}{h_2}_{\mathscr E_\lambda}=\frac{\braket{h_1r_1}{h_2r_2}_{\mathscr V_N}}{\braket{r_1}{r_2}_{\mathscr R_\lambda}}
$$
for arbitrary $r_1,r_2\in \mathscr R_\lambda$. The map $\Phi$ is an isometry, which induces the symmetric bilinear form
$$
    b(h_1\otimes r_1+f_1\otimes t_1,
    h_2\otimes r_2+f_2\otimes t_2)
    \coloneqq 
    \sfrac{1}{d_\lambda}((f_2h_1)(t_2r_1)
    +
    (f_1h_2)(t_1r_2)).
$$

We now want to consider the action of the full symmetry group $G$. If $g\in G_{UL}\sqcup G_{UA}$ then for each $\lambda\in \widehat G_{UL}$ there is some $\lambda^g\in \widehat G_{UL}$ such that  $\rho_\lambda(g)(\mathscr E_\lambda\otimes \mathscr R_\lambda)=\mathscr E_{\lambda^g}\otimes \mathscr R_{\lambda^g}$. If $g\in G_{TL}\sqcup G_{TA}$ then for each $\lambda\in \widehat G_{UL}$ there is some $\lambda^g\in \widehat G_{UL}$ such that  $\rho_\lambda(g)(\mathscr E_\lambda\otimes \mathscr R_\lambda)=\mathscr E_{\lambda^g}^*\otimes \mathscr R_{\lambda^g}^*$. Defining the equivalence relation $\sim$ on $\widehat G_{UL}$ by $\lambda_1\sim \lambda_2$ iff $\exists g\in G:\lambda_2=\lambda_{1}^g$, we can define for each $\Lambda\in \widehat G_{UL}/_\sim$ the blocks
$$\textstyle
    \mathscr A_{\Lambda}
    \coloneqq \bigoplus_{\lambda\in \Lambda}\mathscr E_\lambda\otimes \mathscr R_\lambda,
    \qquad 
    \mathscr B_{\Lambda}\coloneqq \mathscr A_{\Lambda}\oplus \mathscr A_{\Lambda}^*.
$$
Then $\bigoplus_{\Lambda\in \widehat G_{UL}/\sim}\mathscr B_\Lambda$ is a sum of $G$-invariant blocks, which is isomorphic to $\mathscr W_N$ by $\Phi$.

For $\hat H\in \mathscr H_N$ the adjoint action on $\mathscr V_N$ induces an action on each $\mathscr E_{\Lambda}$ by 
$(ad_{\hat H}h)r
\coloneqq 
ad_{\hat H}(hr)=\{\hat H, hr\}$. The action on each $\mathscr E_\lambda\otimes \mathscr R_\lambda$ is then $ad_{\hat H}(h\otimes r)
\coloneqq 
(ad_{\hat H}h)\otimes r$. Thus, on each $\mathscr A_\Lambda$, $G_{UL}$ acts trivially on the $\mathscr E_\lambda$ terms, and the Hamiltonians act trivially on the $\mathscr R_\lambda$ terms. 
Since $\text{Hom}_{G_{UL}}(\mathscr E_{\lambda_1}\otimes \mathscr R_{\lambda_1},\mathscr E_{\lambda_2}\otimes \mathscr R_{\lambda_2})\cong \text{Hom}(\mathscr E_{\lambda_1},\mathscr E_{\lambda_2})\otimes \text{Hom}_{G_{UL}}(\mathscr R_{\lambda_1},\mathscr R_{\lambda_2})$, 
we can reduce the classification problem to the cases where $G_{UL}$ acts trivially on the underlying one-particle Hilbert space.

We can then consider from now on fermion systems with trivial $G_{UL}$-actions. Free Hamiltonians always commute with the particle number operator $\hat N\coloneqq \sum_ja_j^\dagger a_j$. We can then consider \textit{grotesque} fermion system, which are Nambu spaces with symmetries such that
$$G_{UL}=\left\{\exp(-i2\pi t\hat N)\mid t\in \mathds S^1\right\}\cong U(1).$$

If $g\in G_{UA}\sqcup G_{TA}$ then $g^2\in G_{UL}$, so $g^2=z\hat1$ for some $z\in U(1)$. By anti-linearity $zg=g^2g=gg^2=gz=z^*g$, which implies $z=\pm 1$. This means that if $\hat T\in G_{UA}$ then $TT^*=\pm 1_N$, and if $\hat S\in G_{TA}$ then $S^2=\pm 1_N$.

If $\hat C\in G_{TL}$, then we again have $\hat C^2=z\hat1$, and since $[\hat C]_{\widetilde{\mathcal B}}=\begin{bsmallmatrix}
    0&C\\C^*&0
\end{bsmallmatrix}$ we get $CC^*=z1_N$. 
Since $\hat C^2\hat C=\hat C\hat C^2$ implies $zC=z^*C$, we again conclude $z=\pm 1$.

Let $\hat S\in G_{TA}$. If $\hat S^2=-\hat 1$ then $(\text{exp}(\sfrac{i\pi}{2}\hat N)\hat S)^{2}=\hat 1$.  Thus, in the presence of an anti-idempotent chiral symmetry, we can always assume an idempotent representative.

Thus, in a grotesque fermion system, the symmetry group is assumed to be generated by $G_{UL}\cong U(1)$ and at most two non-ordinary symmetries that square to $\pm\hat 1$. If one of the generators is a chiral symmetry, we may assume it to square to the identity. If all 3 types of non-ordinary symmetries are present, we choose the generators such that $\hat S=\hat T\hat C$. In the presence of time-reversal symmetry, we set $\epsilon_T\in \{1,-1\}$ such that $\hat T^2=\epsilon_T\hat 1$. Similarly, in the presence of charge-conjugation we set $\epsilon_C\in \{1,-1\}$ such that $\hat C^2=\epsilon_C\hat 1$.  We also set $\epsilon_S=1$ whenever the system has chiral symmetry.

We then get 10 classes of grotesque Nambu spaces with symmetries. We can index them by a signature $(\epsilon_S\epsilon_C\epsilon_T)\in \{0,1\}\times \{-1,0,1\}^2$, where $0$ indicates the absence of the associated symmetry type.

\subsection{Convenient choice of basis}

We now show how to choose convenient basis for the 10 symmetry classes so the unitary matrices representing the symmetries have a convenient form, which are listed in Table \ref{tab:Classif1}. \\

\begin{table}[ht]
    \caption{The first column lists the Cartan-Atland-Zirnbauer (CAZ) symmetry class label, followed by their corresponding symmetry signature. Columns three through five specify the matrices representing the chiral, charge-conjugation and time-inversion symmetries. The final column describes the structure of the adjoint representation of the  equivariant free Hamiltonians.}
    \label{tab:Classif1}
    \centering
    \begin{tabular}{|c|c|c|c|c|c|}\hline
     \text{CAZ}&$(\epsilon_S\epsilon_C\epsilon_T)$& $S$& $C$& $T$&$[ad_{\hat H}]_{\mathcal B}$ \\\hline\hline
     AIII&$(100)$&$1_{m,N-m}$&-&-&$H=\mqty[0&b\\b^\dagger&0]$\\\hline 
     
     A&$(000)$&-&-&-&-\\ \hline\hline

     BDI&$(111)$&$1_{m,N-m}$&$1_{m,N-m}$&$1_{N}$&$H=\mqty[0&b\\b^t&0],b=b^*$\\\hline

     AI&$(001)$&-&-&$1_N$&$H=H^t$\\\hline

     CI&$(1\text{-}11)$&$1_{\sfrac{N}{2},\sfrac{N}{2}}$&$-J_N$&$F_N$&$H=\mqty[0&b\\b^*&0],b^t=b$\\\hline

     C&$(0\text{-}10)$&-&$J_N$&-&$H=\mqty[a&b\\b^*&-a^*],
         \mqty{a^\dagger=a
         ,\\
         b^t=b}$\\\hline

     CII&$(1\text{-}1\text{-}1)$&$1_{m,N-m}$&$\mqty[-J_{m}&0\\0&J_{N-m}]$&$\mqty[J_{m}&0\\0&J_{N-m}]$&$H=     \mqty[0&B\\B^\dagger&0],B=\mqty[b_0&b_1\\-b_1^*&b_0^*]$\\\hline
     
     AII&$(00\text{-}1)$&-&-&$J_N$&$H=\mqty[a&b\\-b^*&a^*],
         \mqty{a^\dagger=a,\\b^t=-b^t}$\\\hline

     DIII&$(11\text{-}1)$&$1_{\sfrac{N}{2},\sfrac{N}{2}}$&$F_N$&$J_N$&$H=\mqty[0&b\\-b^*&0],b^t=-b$\\\hline

     D&$(010)$&-&$1_N$&-&$H=-H^t$\\\hline
    \end{tabular}
\end{table}

\textbf{Class A $(000)$:} This class of systems have no non-ordinary symmetries, so any orthonormal basis may be chosen.\\

Let's consider now the classes whose non-ordinary symmetries are generated by a single operator.\\

\textbf{Class AIII $(100)$:} Let $\hat S\in G_{TA}$ be a chiral symmetry  such that $\hat S^2=\hat 1$, so its associated matrix satisfies $S^2=1_N$. Since $S$ is unitary that means $S^\dagger=S$, and so the eigenvalues of $S$ are all $\pm1$. 
This means that we have a decomposition $\mathscr V_N=\mathscr V_N^+\oplus \mathscr V_N^-$ into the positive and negative eigenspaces of $\hat S$. Let $m=\dim \mathscr V_N^+$. If $\mathcal B^{\pm}$ are basis for $\mathscr V_N^\pm$ and $\mathcal B=\mathcal B^+\sqcup \mathcal B^-$, then 
$$
    S=1_{m,N-m}.
$$

Consider now systems with only a charge-conjugation symmetry $\hat C$ satisfying $\hat C^2=\epsilon_C\hat 1$, or equivalently $CC^*=\epsilon_C1_N$. This means $C^\dagger=\epsilon_CC^*$, or equivalently $C=\epsilon_C C^t$. \\

\textbf{Class D $(010)$:} Let $\hat C\in G_{TL}$ be a charge-conjugation symmetry such that $\hat C^2=\hat 1$, so that it is represent by a symmetric matrix $C$. By Takagi decomposition, there is some unitary matrix $A$ such that $C=AA^t$.  If $R$ is some unitary matrix representing a change of basis in $\mathscr V_N$, it induces a change of basis in $\mathscr W_N$ represented by $\begin{bsmallmatrix}
        R&0\\0&R^*
    \end{bsmallmatrix}$. This change of basis transforms the representation of $\hat C$ into $\begin{bsmallmatrix}
        0&RCR^t\\R^*C^*R^\dagger&0
    \end{bsmallmatrix}$. We may then conclude, using $R=A$, that in some basis we have 
    $$
        C=1_N.
    $$

\textbf{Class C $(0\text{-}10)$:} Suppose now that $\hat C\in G_{TL}$ satisfies $\hat C^2=-\hat 1$, and so that $C$ is skew-symmetric. Since $CC^*=-1_N$ we have that $\abs{\det C}^2=(-1)^N$, which can only be the case if $N$ is even. Takagi decomposition tells us there is some matrix $A$ such that $C=AJ_NA^t$. Thus there is some basis in which 
    $$C=J_N.$$

\textbf{Class AI $(001)$:} For a time-reversal symmetry $\hat T\in G_{UA}$ such that $\hat T^2=\hat 1$ we have $TT^*=1_N$, so by the same argument as for systems in the class $(010)$ there is some basis such that 
$$T=1_N.$$

\textbf{Class AII $(00\text{-}1)$:} For $\hat T\in G_{UA}$ such that $\hat T^2=-\hat 1$ we have $TT^*=-1_N$, which forces $N$ to be even. As in class $(0\text{-}10)$ there is some basis such that 
$$T=J_N.$$

Let's now consider systems in which all 3 types of non-ordinary symmetries are present. The assumption $\hat S=\hat T\hat C$ is equivalent to $S=TC^*$. We may start by choosing a basis where $S=1_{m,N-m}$. Since $T=\epsilon_T T^t$ we can write $T=\begin{bsmallmatrix}
        a&b\\\epsilon_Tb^t&d
    \end{bsmallmatrix}$, with $a\in M_m\mathds C$ satisfying $a=\epsilon_Ta^t$, and $d\in M_{N-m}\mathds C$ satisfying $d=\epsilon_Td^t$. We then have $C=T^tS^ *=\begin{bsmallmatrix}
        \epsilon_Ta&-\epsilon_Tb\\b^t&-\epsilon_Td
    \end{bsmallmatrix}$.\\

Suppose now $\epsilon_T=\epsilon_C$. The conditions $\epsilon_T 1_N=TT^*=CC^*$ give us further restraints on $a$, $b$ and $d$. The second equation simplifies to $\begin{bsmallmatrix}
        bb^\dagger&ab^*\\db^\dagger&b^tb^*
    \end{bsmallmatrix}=0$, which implies the first equation is equivalent to $\begin{bsmallmatrix}
        aa^*&bd^*\\b^ta^*&dd^*
    \end{bsmallmatrix}=\begin{bsmallmatrix}
        \epsilon_T1_m&0\\0&\epsilon_T1_{N-m}
    \end{bsmallmatrix}$.\\

\textbf{Class BDI $(111)$:} Suppose $\epsilon_T=\epsilon_C=1$. From the above equations we can conclude both $a$ and $d$ are symmetric unitary matrices, and that $b=0$. Takagi decomposition then tells us that there is a unitary matrix of the form $R=\begin{bsmallmatrix}
        r&0\\0&s
    \end{bsmallmatrix}$ such that $RR^t=T$. Since $R1_{m,N-m}R^\dagger=1_{m,N-m}$, in the basis determined by $R$ we have
    $$
        S=1_{m,N-m},\qquad C=1_{m,N-m},\qquad T=1_N.
    $$

\textbf{Class CII $(1\text{-}1\text{-}1)$:}   If $\epsilon_T=\epsilon_C=-1$ then $a$ and $d$ are skew-symmetric, and $b=0$. The fact that $aa^*=-1_m$ and $dd^*=-1_{N-m}$ implies both $m$ and $N-m$ are even. Again, Takagi decomposition gives us a basis change such that $a=J_m$ and $d=J_{N-m}$. Thus, we have a basis in which 
$$S=1_{m,N-m},\qquad 
C=\mqty[-J_m&0\\0&J_{N-m}]
,\qquad
T=\mqty[J_m&0\\0&J_{N-m}].
$$

Suppose now $\epsilon_T=-\epsilon_C$. Note that either $\epsilon_T$ or $\epsilon_C$ must be $-1$, which implies $N$ is even. In this case we have $TT^*=-CC^*$, which is equivalent to $\begin{bsmallmatrix}
        aa^*&bd^*\\b^ta^*&dd^*
    \end{bsmallmatrix}=0$. 
    From $TT^*=\epsilon_T1_N$ we now get $\begin{bsmallmatrix}
        bb^\dagger&ab^*\\db^\dagger&b^tb^*
    \end{bsmallmatrix}=\begin{bsmallmatrix}
        \epsilon_T1_m&0\\0&\epsilon_T1_{N-m}
    \end{bsmallmatrix}$. This last equation implies $m=\sfrac{N}{2}$, $a=0$ and $d=0$.\\

\textbf{Class CI $(1\text{-}11)$:} Suppose now $\epsilon_T=1$ and $\epsilon_C=-1$. The matrix $R=\begin{bsmallmatrix}
    b^\dagger&0\\0&1_{\sfrac{N}{2}}
\end{bsmallmatrix}$ gives us a change of basis, in which we now have
$$
    S=1_{\sfrac{N}{2},\sfrac{N}{2}}
    ,\qquad C=-J_N
    ,\qquad 
    T=F_N.
$$

\textbf{Class DIII $(11\text{-}1)$:} If $\epsilon_T=-1$ and $\epsilon_C=1$  then the same matrix $R=\begin{bsmallmatrix}
    b^\dagger&0\\0&1_{\sfrac{N}{2}}
\end{bsmallmatrix}$ provides a  change of basis where 
$$
    S=1_{\sfrac{N}{2},\sfrac{N}{2}}
    ,
    \qquad 
   C=F_N
   ,\qquad 
    T=J_N. 
$$

\subsection{Spaces of equivariant free Hamiltonian}

We are interested in classifying the spaces of equivariant free Hamiltonians. For $\hat U\in G$ and $\hat H\in \mathscr H_N$ we have 
$
    \hat U ad_{\hat H}\hat U^\dagger \Psi=\hat U\hat H\hat U^\dagger \Psi-\Psi\hat H,
$
thus 
$$\hat U\hat H\hat U^\dagger = \hat H\iff\hat U ad_{\hat H}\hat U^\dagger=ad_{\hat H}.$$ 

\begin{definition}
    Let $(\mathscr W_N,b,G)$ be a Nambu space with symmetries of signature $(\epsilon_S\epsilon_C\epsilon_T)$. Its \textit{space of equivariant free Hamiltonians} is
    $$
        \mathscr H_N^{\epsilon_S\epsilon_C\epsilon_T}=\{\hat H\in\mathscr H_N\mid \forall \hat U\in G:\hat U\hat H\hat U^\dagger=\hat H\}.
    $$
\end{definition}

For a usual symmetry, conjugation of the adjoint representation corresponds, in a given orthonormal basis, to conjugation of the representative matrices. Explicitly for an ordinary symmetry $\hat U\in G_{UL}$ and free Hamiltonian  $\hat H$, since 
$$
[\hat Uad_{\hat H}\hat U^\dagger]_{\widetilde{\mathcal B}}=\mqty[UHU^\dagger&0\\0&-U^*H^* U^t]$$
we have that 
$$\hat U\hat H\hat U^\dagger=\hat H
\iff UHU^\dagger =H.$$

For $\hat S\in G_{TA}$, $\hat C\in G_{TL}$ and $\hat T\in G_{UA}$ we have
\begin{gather*}
    [\hat Sad_{\hat H}\hat S^\dagger]_{\widetilde{\mathcal B}}=\mqty[
    0&S\\S^*&0]K\mqty[
    H&0\\0&-H^*]K\mqty[
    0&S^t\\S^\dagger&0]=\mqty[
    -SHS^\dagger&0\\0&S^*H^* S^t],\\
    [\hat Cad_{\hat H}\hat C^\dagger]_{\widetilde{\mathcal B}}=\mqty[0&C\\C^*&0]\mqty[
    H&0\\0&-H^*] \mqty[
    0&C^t\\C^\dagger&0]=\mqty[
    -CH^*C^\dagger&0\\0&C^*H C^t],\\
    [\hat Tad_{\hat H}\hat T^\dagger]_{\widetilde{\mathcal B}}=\mqty[T&0\\0&T^*]K\mqty[
    H&0\\0&-H^*]K\mqty[
    T^\dagger&0\\0&T^t]=\mqty[
    TH^*T^\dagger&0\\0&-T^*H T^t].
\end{gather*}
Thus
\begin{align*}
\hat S\hat H\hat S^\dagger=\hat H 
&\iff SHS^\dagger =-H,\\
\hat C\hat H\hat C^\dagger=\hat H 
&\iff CHC^\dagger =-H^*
,\\
    \hat T\hat H\hat T^\dagger=\hat H &\iff THT^\dagger =H^*.
\end{align*}

We can now characterize the matrices that represent the equivariant free Hamiltonians for each symmetry class in the convenient basis of the last section. We summarize the results from this section in the last column of Table \ref{tab:Classif1}.\\

\textbf{Class AIII $(100)$:} Let $\hat S\in G_{TA}$ be the chiral symmetry represented by the matrix $S=1_{m,N-m}$. For all $a\in M_m\mathds C$, $b\in M_{m,N-m}\mathds C$, $c\in M_{N-m,m}\mathds C$ and $d\in M_{N-m}\mathds C$ we have $1_{m,N-m}\begin{bsmallmatrix}
    a&b\\c&d
\end{bsmallmatrix}1_{m,N-m}=\begin{bsmallmatrix}
    a&-b\\-c&d
\end{bsmallmatrix}$. By self-adjointness, for a Hamiltonian $\hat H$ its representing matrix is of the form $H=\begin{bsmallmatrix}
    a&b\\b^\dagger&d
\end{bsmallmatrix}$. Since $\hat S\hat H\hat S^\dagger=\hat H$ iff $SHS=-H$, we conclude the equivariant Hamiltonians in this class are represented by matrices of the form $H=\begin{bsmallmatrix}
    0&b\\b^\dagger&0
\end{bsmallmatrix}$.\\

\textbf{Class D $(010)$:} If $\hat C\in G_{TL}$ is represented by  $C=1_N$ then a Hamiltonian $\hat H$ commutes with $\hat C$ iff $H=-H^*$.\\

\textbf{Class C $(0\text{-}10)$:}  If $\hat C\in G_{TL}$ is represented by  $C=J_N$ then a Hamiltonian $\hat H$ commutes with $\hat C$ iff $J_NHJ_N^\dagger=-H^*$. If $H=\begin{bsmallmatrix}
    a&b\\b^\dagger&d
\end{bsmallmatrix}$ for $a,b,d\in M_{\sfrac{N}{2}}\mathds C$, with $a$ and $d$ self-adjoint, then $J_NHJ_N^\dagger=\begin{bsmallmatrix}
    d&-b^\dagger\\-b&a
\end{bsmallmatrix}$. So a Hamiltonian $\hat H$ commutes with $\hat C$ iff $H=\begin{bsmallmatrix}
    a&b\\b^*&-a^*
\end{bsmallmatrix}$, $a=a^\dagger$ and $b=b^t$.\\

\textbf{Class AI $(001)$:} If $\hat T\in G_{UA}$ is represented by  $T=1_N$ then a Hamiltonian $\hat H$ commutes with $\hat T$ iff $H=H^*$.\\

\textbf{Class AII $(00\text{-}1)$:}  If $\hat T\in G_{UA}$ is represented by  $T=J_N$ then $\hat H$ commutes with $\hat T$ iff $J_NHJ_N^\dagger=H^*$. Thus $\hat H$ commutes with $\hat T$ iff $H=\begin{bsmallmatrix}
    a&b\\-b^*&a^*
\end{bsmallmatrix}$, $a=a^\dagger$ and $b=-b^t$.\\

\textbf{Class BDI $(111)$:} From the restrictions in classes $AIII$ and $AI$, $\hat H$ is equivariant in this class  iff $H=\begin{bsmallmatrix}
    0&b\\b^t&0
\end{bsmallmatrix}$ and $b=b^*$.\\

\textbf{Class DIII $(11\text{-}1)$:} From the restrictions in classes $AIII$ and $AII$  a Hamiltonian in this class is equivariant iff  $H=\begin{bsmallmatrix}
    0&b\\-b^*&0
\end{bsmallmatrix}$, with $b=-b^t$.\\

\textbf{Class CI $(1\text{-}11)$:} From the restriction in $AIII$ and equivariance with regards to time-reflection symmetry  we again have that $H=\begin{bsmallmatrix}
    0&b\\b^\dagger&0
\end{bsmallmatrix}$, and now  $F_NHF_N=H^*$ is equivalent to $b=b^t$.\\

\textbf{Class CII $(1\text{-}1\text{-}1)$:} From the restriction in $AIII$ we have $H=\begin{bsmallmatrix}
    0&B\\B^\dagger &0
\end{bsmallmatrix}$, where  $B=\begin{bsmallmatrix}
    a&b\\c&d
\end{bsmallmatrix}$ with $a,b,c,d\in M_{\sfrac{m}{2},\sfrac{N-m}{2}}\mathds C$. Equivariance with regards to time-reflection symmetry implies $J_mBJ_{N-m}^\dagger=B^*$, thus $d=a^*$ and $c=-b^*$.

\subsection{Classification by compact symmetric spaces}

The central observation of the classification of grotesque free fermion system is that $\mathscr H^{\epsilon_S\epsilon_C\epsilon_T}_N$ is the tangent space of an irreducible compact symmetric space, and that all irreducible compact symmetric space classes are realized this way.

\begin{definition}
    Let $(\mathscr W_N,b,G)$ be a Nambu space with symmetries of signature $(\epsilon_S\epsilon_C\epsilon_T)$. The \textit{space of equivariant free time-evolution operators} is
    $$
        \mathscr M^{\epsilon_S\epsilon_C\epsilon_T}_N \coloneqq \{Ad_{\exp(-i\hat H)}\in Aut(Cl(\mathscr W_N,b))\mid \hat H\in \mathscr H^{\epsilon_S\epsilon_C\epsilon_T}_N\}.
    $$
\end{definition}

Let $\mathcal M$ be a connected $d$-dimensional Riemannian manifold and $p\in \mathcal M$. A diffeomorphism $f$ on a neighborhood $U$ of $p$ is a \textit{geodesic symmetry} around $p$ if it fixes $p$ and reverses geodesics around $p$, ie if $\gamma$ is  a geodesic with $\gamma(0)=p$ then $f\gamma(\theta)=\gamma(-\theta)$. In particular this implies $df_p=-Id_{T_p\mathcal M}$. We say $\mathcal M$ is \textit{locally symmetric} if its geodesic symmetries are isometries, and that a locally symmetric manifold is \textit{(globally) symmetric} if all its geodesic symmetries can be extended to all of $\mathcal M$.

The space $\mathscr M^{\epsilon_S\epsilon_C\epsilon_T}_N$ is a compact symmetric space. The symmetric structure comes from the following construction: Let $\mathcal {G}$ be a connected compact Lie group equipped with a Cartan involution, ie an automorphism $\tau:\mathcal {G}\to \mathcal {G}$ with the involutive property $\tau^2=id_{\mathcal {G}}$. Define the subgroup of $\tau$-fixed elements as 
$$\text{Fix}_\tau\coloneqq \{u\in \mathcal {G}\mid u=\tau(u)\}.$$ 

Then the quotient space $\frac{\mathcal {G}}{\text{Fix}_\tau}$ is a compact symmetric space. To define its Riemannian structure let $\mathfrak g=T_e\mathcal {G}$ be the Lie algebra of $\mathcal {G}$, $\mathfrak g=\mathfrak k\oplus\mathfrak p$ be its decomposition into the positive and negative eigenspaces of the involution $d\tau(e)$, and let $\braket{\cdot }{\cdot}_{\mathfrak p}$ be an inner product which is invariant under the adjoint action of $\text{Fix}_\tau$ on $\mathfrak p$, ie $\braket{v}{w}_{\mathfrak p}=\braket{Ad_kv}{Ad_kw}_{\mathfrak p}$ for all $v,w\in\mathfrak p$ and $k\in \text{Fix}_\tau$. Then the Riemannian metric at $T_{u\text{Fix}_\tau}\frac{\mathcal {G}}{\text{Fix}_\tau}$ is $g_{u\text{Fix}_\tau}(v,w)\coloneqq \braket{dL_u^{-1}(v)}{dL_u^{-1}(w)}_{\mathfrak p}$, where $L_u(v\text{Fix}_\tau)=uv\text{Fix}_\tau$. The submanifold 
$$\text{Inv}_\tau\coloneqq \{u\in \mathcal {G}\mid \tau(u)=u^{-1}\}$$
of $\tau$-inverted elements is isometric to $\frac{\mathcal {G}}{\text{Fix}_\tau}$ by the Cartan embedding 
$$
    c_{\tau}:\frac{\mathcal {G}}{\text{Fix}_\tau}\to \text{Inv}_\tau\subset \mathcal {G},
    \qquad  c_{\tau}(u\text{Fix}_\tau)\coloneqq  u\tau(u^{-1}).
$$
The geodesic inversion with respect to $y\in \text{Inv}_\tau$ is $s_y(x)=yx^{-1}y$.
If $v\in \text{Inv}_\tau$ then $c_\tau(v\text{Fix}_\tau)=v\tau(v^{-1})=v^2$. Thus if $v$ is a square root of $u$ in $\text{Inv}_\tau$, then $c^{-1}_\tau(u)=v\text{Fix}_\tau$.\\

We now derive the symmetric spaces associated with each class of grotesque fermion systems. In all cases we assume the convenient basis of the previous sections, and their structure is listed in Table \ref{tab:Classif2}.\\

\begin{table}[ht]
\caption{The third column describes the structure of equivariant time evolution operators in $\mathscr M^{\epsilon_S\epsilon_C\epsilon_T}_{N}$. The fourth column lists the corresponding compact symmetric space classified by Cartan.}
\label{tab:Classif2}
    \begin{tabular}{|c|c|c|c|}\hline
     \text{CAZ}&$(\epsilon_S\epsilon_C\epsilon_T)$& $[Ad_{\hat U}]_{\mathcal B}$ & Symmetric space\\\hline\hline
     AIII&$(100)$&$
     U=\mqty[a&b\\-b^\dagger&d],\mqty{a^\dagger=a\\d^\dagger=d}$&$\frac{U(N)}{U(m)\times U(N-m)}$\\\hline 
     
     A&$(000)$&-&$U(N)$\\ \hline\hline

     BDI&$(111)$&$U=\mqty[a&b\\b^t&d],\mqty{a^*=a^t=a\\d^*=d^t=d\\b^*=-b}$&$\frac{O(N)}{O(m)\times O(N-m)}$\\\hline

     AI&$(001)$&$U=U^t$&$\frac{U(N)}{O(N)}$\\\hline

     CI&$(1\text{-}11)$&$U=\mqty[
         a&b\\-b^*&a^*],\mqty{a^\dagger=a\\b^t=b}$&$\frac{Sp(\sfrac{N}{2})}{U(\sfrac{N}{2})}$\\\hline

     C&$(0\text{-}10)$&$U=\mqty[
         a&b\\-b^*&a^*]$&$Sp(\sfrac{N}{2})$\\\hline

     CII&$(1\text{-}1\text{-}1)$&$U=\mqty[a_{0}&a_1&b_0&b_1\\-a_1^*&a_0^*&b_1^*&-b_0^*\\-b_0^\dagger&-b_{1}^t&d_{0}&d_1\\-b_{1}^\dagger&b_{0}^t&-d_1^*&d_{0}^*],\mqty{a_0^\dagger=a_0\\
     a_1^t=-a_1\\d_0^\dagger=d_0\\d_1^t=-d_1}$&$\frac{Sp(\sfrac{N}{2})}{Sp(\sfrac{m}{2})\times Sp(\sfrac{N-m}{2})}$\\\hline
     
     AII&$(00\text{-}1)$&$U=\mqty[
         a&b\\c&a^t],\mqty{b^t=-b\\c^t=-c}$&$\frac{U(N)}{Sp(\sfrac{N}{2})}$\\\hline

     DIII&$(11\text{-}1)$&$U=\mqty[
         a&b\\b^*&a^*],\mqty{a^\dagger=a\\b^t=-b}$&$\frac{O(N)}{U(\sfrac{N}{2})}$\\\hline

    D&$(010)$&$U=U^*$&$O(N)$\\\hline
    \end{tabular}
\end{table}

\textbf{Class A $(000)$:} In the absence of non-ordinary symmetries 
$
    \mathscr M^{000}_N=\mathscr M_{N}\cong U(N).
$\\

In the presence of a charge-conjugation symmetry we will have to consider the Lie subgroup of elements fixed by its action. Let $\hat C\in G_{TL}$ be a charge-conjugation symmetry such that $\hat C^2=\pm \hat 1$. Then 
$$
    \Xi:\mathscr M^{000}_N\to \mathscr M^{000}_N
    ,\qquad 
    \Xi(Ad_{\hat U})\coloneqq \hat C Ad_{\hat U}\hat C^\dagger
    ,\qquad 
    \Xi\left(\mqty[
    U&0\\0&U^*]\right)=
    \mqty[CU^*C^\dagger&0\\0&C^*UC^t]
$$ 
is a Cartan involution. A Hamiltonian $\hat H$ commutes with $\hat C$ iff $Ad_{e^{-i\hat H}}$ is $\Xi$-fixed. At the matrix level the $\Xi$-fixed elements satisfy $C UC^\dagger=U^*$. 
\\

\textbf{Class D $(010)$:} $Ad_{\hat U}\in \mathscr M^{010}_N= \text{Fix}_{\Xi}$ iff $U=U^*$. Thus $\mathscr M^{010}_N\cong O(N)$.\\

\textbf{Class C $(0\text{-}10)$:} $Ad_{\hat U}\in \mathscr M^{0-10}_N= \text{Fix}_{\Xi}$ iff $U=\begin{bsmallmatrix}
    a&b\\-b^*&a^*
\end{bsmallmatrix}$. Thus $\mathscr M^{0-10}_N\cong Sp(\sfrac{N}{2})$.\\

Let $\hat S\in G_{TA}$ be a chiral symmetry such that $\hat S^2= \hat 1$. Then
$$
    \Sigma:\mathscr M^{000}_N\to \mathscr M^{000}_N
    ,\qquad 
    \Sigma(Ad_{\hat U})\coloneqq
     \hat S Ad_{\hat U}\hat S^\dagger
    ,\qquad 
    \Sigma\left(\mqty[U&0\\0&U^*]\right)=
    \mqty[SUS^\dagger&0\\0&S^*U^*S^t]
$$ 
is a Cartan involution. Due to $\hat S$ being anti-linear we now have that $\hat H$ commutes with $\hat S$ iff  $ Ad_{\exp(-i\hat H)}\in \text{Inv}_{\Sigma}$. The $\Sigma$-fixed elements are those such that $SUS^\dagger=U$, and the $\Sigma$-inverted elements satisfy $SUS^\dagger=U^\dagger$. The associated Cartan embedding is 
$$
    c_\Sigma(Ad_{\hat U}\text{Fix}_{\Sigma})=[Ad_{\hat U},\hat S].
$$

\textbf{Class AIII $(100)$:} In our choice of basis we have $S=1_{m,N-m}$ for some $m\leq N$. Let $Ad_{\hat U}\in \mathscr M^{000}_N$, $a\in M_m\mathds C$, $b\in M_{m,N-m}\mathds C$, $c\in M_{N-m,m}\mathds C$ and $d\in M_{N-m}\mathds C$ be such that
$U=\begin{bsmallmatrix}
    a&b\\c&d
\end{bsmallmatrix}$. Since $Ad_{\hat U}\in \text{Fix}_{\Sigma}$ iff $b=0$ and $c=0$, so that $\text{Fix}_{\Sigma}\cong U(m)\times U(N-m)$, then $\mathscr M^{100}_N=\text{Inv}_{\Sigma}\cong\frac{\mathscr M^{000}_N}{\text{Fix}_{\Sigma}}\cong \frac{U(N)}{U(m)\times U(N-m)}$. Thus $Ad_{\hat U}\in \mathscr M^{100}_N$ iff $U=\begin{bsmallmatrix}
    a&b\\-b^\dagger&d
\end{bsmallmatrix}$, $a=a^\dagger$ and $d=d^\dagger$.
Note we get different spaces depending on $m\leq N$, so we will denote the spaces of time-evolution operators in class $AIII$ as $\mathscr M^{100}_{N,m}$ to differentiate between them.\\

Let $\hat T\in G_{UA}$ be a chiral symmetry such that $\hat T^2= \pm\hat 1$. Then
$$
    \Theta:\mathscr M^{000}_N\to \mathscr M^{000}_N
    ,\qquad 
    \Theta(Ad_{\hat U})\coloneqq
     \hat T Ad_{\hat U}\hat T^\dagger
    ,\qquad 
    \Theta\left(\mqty[
    U&0\\0&U^*]\right)=
    \mqty[TU^*T^\dagger&0\\0&T^*UT^t]
$$ 
is a Cartan involution. Due to $\hat T$ being anti-linear we now have that $\hat H$ commutes with $\hat T$ iff  $Ad_{\exp(-i\hat H)}\in \text{Inv}_{\Theta}$. The $\Theta$-fixed elements are those  such that $TUT^\dagger=U^*$, and the $\Theta$-inverted elements satisfy $TUT^\dagger=U^t$. 
The associated Cartan embedding is 
$$
    c_\Theta(Ad_{\hat U}\text{Fix}_\Theta)=[Ad_{\hat U},\hat T].
$$

\textbf{Class AI $(001)$:} Since $\text{Fix}_{\Theta}\cong O(N)$ we have $\mathscr M^{001}_N=\text{Inv}_{\Theta}
% \cong \frac{\mathscr M^{000}_N}{\text{Fix}_{\Theta}}
\cong\frac{U(N)}{O(N)}$. Thus $Ad_{\hat U}\in \mathscr M^{001}_N$ iff $U^t=U$.\\

\textbf{Class AII $(00\text{-}1)$:} Since $\text{Fix}_{\Theta}\cong Sp(\sfrac{N}{2})$ we have $\mathscr M^{00-1}_N=\text{Inv}_{\Theta}
% \cong \frac{\mathscr M^{000}_N}{\text{Fix}_{\Theta}}
\cong\frac{U(N)}{Sp(\sfrac{N}{2})}$. Thus $Ad_{\hat U}\in \mathscr M^{00-1}_N$ iff $U=\begin{bsmallmatrix}
    a&b\\c&a^t
\end{bsmallmatrix}$, $b^t=-b$ and $c^t=-c$.\\

When all 3 types of non-ordinary symmetries are present, we consider the subgroup of $\Xi$-fixed elements. In all cases, this group is closed under the involution $\Sigma$, and the space of time-evolution operators is the space of $\Xi$-fixed elements that are $\Sigma$-inverted. The same results would be obtained if we considered the involution $\Theta$ instead of $\Sigma$.\\

\textbf{Class BDI $(111)$:}  The $\Xi$-fixed elements are represented by unitary matrices $U=\begin{bsmallmatrix}
    a&b\\c&d
\end{bsmallmatrix}$ such that $CUC=U^*$ for $C=1_{m,N-m}$, which means $\begin{bsmallmatrix}
    a&-b\\-c&d
\end{bsmallmatrix}=\begin{bsmallmatrix}
    a^*&b^*\\c^*&d^*
\end{bsmallmatrix}$, ie those where $a\in M_m\mathds R$, $d\in M_{N-m}\mathds R$, $b\in M_{m,N-m}i\mathds R$ and $c\in M_{N-m,m}i\mathds R$. We then have an isomorphism
$$
\text{Fix}_{\Xi}\xrightarrow{\cong}O(N),\qquad 
U\mapsto 
\mqty[1_m&0\\0&-i1_{N-m}]
\mqty[a&b\\c&d]
\mqty[1_m&0\\0&i1_{N-m}]
=
\mqty[a&ib\\-ic&d].
$$
The $\Sigma$-fixed elements are those such that $b=0$ and $c=0$, so $\mathscr M^{111}_{N,m}=\text{Fix}_\Xi\cap \text{Inv}_\Sigma
% \cong\frac{\text{Fix}_\Xi}{\text{Fix}_\Sigma}
\cong\frac{O(N)}{O(m)\times O(N-m)}$. Thus $Ad_{\hat U}\in \mathscr M^{111}_{N,m}$ iff $U=\begin{bsmallmatrix}
    a&b\\b^t&d
\end{bsmallmatrix}$, $a=a^t=a^*$, $b^*=-b$ and $d=d^t=d^*$.\\

\textbf{Class DIII $(11\text{-}1)$:} Since in this class $C=F_N$ the $\Xi$-fixed elements are represented by matrices $U=\begin{bsmallmatrix}
    a&b\\b^*&a^*
\end{bsmallmatrix}$. We then have an isomorphism
\begin{gather*}
    \text{Fix}_{\Xi}\xrightarrow{\cong}O(N),\\ 
U\mapsto 
\frac{1}{2}\mqty[1_{\sfrac{N}{2}}&1_{\sfrac{N}{2}}\\-i1_{\sfrac{N}{2}}&i1_{\sfrac{N}{2}}]
\mqty[a&b\\b^*&a^*]\mqty[1_{\sfrac{N}{2}}&i1_{\sfrac{N}{2}}\\1_{\sfrac{N}{2}}&-i1_{\sfrac{N}{2}}]=\mqty[
    \Re(a+b)&-\Im(a-b)\\\Im(a+b)&\Re(a-b)].
\end{gather*}
The $\Sigma$-fixed elements are those such that $b=0$. The above isomorphism maps $\text{Fix}_\Sigma$ to the subgroup of $O(N)$ composed of matrices of the form $\begin{bsmallmatrix}
    \Re(a)&-\Im(a)\\\Im(a)&\Re(a)
\end{bsmallmatrix}$, which is isomorphic to $U(\sfrac{N}{2})$, so $\mathscr M^{11-1}_N=\text{Fix}_\Xi\cap \text{Inv}_\Sigma
% \cong\frac{\text{Fix}_\Xi}{\text{Fix}_\Sigma}
\cong\frac{O(N)}{U(\sfrac{N}{2})}$. Thus $Ad_{\hat U}\in \mathscr M^{11-1}_{N,m}$ iff $U=\begin{bsmallmatrix}
    a&b\\b^*&a^*
\end{bsmallmatrix}$, $a=a^\dagger$ and $b=-b^t$.\\

\textbf{Class CI $(1\text{-}11)$:} In this class we have $\text{Fix}_\Xi\cong Sp(\sfrac{N}{2})$, as in class C. The subgroup of $\Sigma$-fixed elements are represented by matrices of the form $\begin{bsmallmatrix}
    a&0\\0&a^*
\end{bsmallmatrix}$, which is isomorphic to $U(\sfrac{N}{2})$, so $\mathscr M^{1-11}_N=\text{Fix}_\Xi\cap \text{Inv}_\Sigma
% \cong\frac{\text{Fix}_\Xi}{\text{Fix}_\Sigma}
\cong \frac{Sp(\sfrac{N}{2})}{U(\sfrac{N}{2})}$. 
Thus $Ad_{\hat U}\in \mathscr M^{1-11}_{N}$ iff $U=\begin{bsmallmatrix}
    a&b\\-b^*&a^*
\end{bsmallmatrix}$, $a=a^\dagger$ and $b=b^t$.\\

\textbf{Class CII $(1\text{-}1\text{-}1)$:} In this class we have that the $\Xi$-fixed elements are of the form 

$$
    U=\mqty[
    a_{0}&a_{1}&b_{0}&b_{1}\\-a_{1}^*&a_{0}^*&b_{1}^*&-b_{0}^*\\c_{0}&c_{1}&d_{0}&d_{1}\\c_{1}^*&-c_{0}^*&-d_{1}^*&d_{0}^*]
,\qquad\mqty{ 
a_{i}\in M_{\sfrac{m}{2}}\mathds C,\\
b_{i}\in M_{\sfrac{m}{2},\sfrac{N-m}{2}}\mathds C,\\
c_{i}\in M_{\sfrac{N-m}{2},\sfrac{m}{2}}\mathds C,\\
d_{i}\in M_{\sfrac{N-m}{2}}\mathds C.}
$$
Setting
\begin{equation}
    V=\mqty[1_{\sfrac{m}{2}}&0&0&0\\0&0&-i1_{\sfrac{m}{2}}&0\\0&1_{\sfrac{N-m}{2}}&0&0\\0&0&0&-i 1_{\sfrac{N-m}{2}}]
\label{V1-1-1}
\end{equation}
we then have an isomorphism
$$
\text{Fix}_{\Xi}\xrightarrow{\cong}Sp(\sfrac{N}{2}),\qquad 
U\mapsto VUV^\dagger=\mqty[
    a_{0}&ib_{0}&a_{1}&ib_{1}\\-ic_{0}&d_{0}&-ic_{1}&d_{1}\\-a_{1}^*&ib_{1}^*&a_{0}^*&-ib_{0}^*\\-ic_{1}^*&-d_{1}^*&ic_{0}^*&d_{0}^*].
$$
The $\Sigma$-fixed elements are those such that $b_{i}=0$ and $c_{i}=0$, so $\mathscr M^{1-1-1}_{N,m}=\text{Fix}_\Xi\cap \text{Inv}_\Sigma\cong\frac{Sp(\sfrac{N}{2})}{Sp(\sfrac{m}{2})\times Sp(\sfrac{N-m}{2})}$. Thus $Ad_{\hat U}\in \mathscr M^{1-1-1}_{N,m}$ iff $U=\begin{bsmallmatrix}
    a_0&a_1&b_0&b_1\\-a_1^*&a_0^*&b_1^*&-b_0^*\\-b_0^\dagger&-b_1^t&d_0&d_1\\-b_1^\dagger&b_0^t&-d_1^*&d_0^*
\end{bsmallmatrix}$, $a_0=a_0^\dagger$, $a_1^t=-a_1$, $d_0^\dagger=d_0$ and $d_1^t=-d_1$.
%%%%%%%%%%%%%%%%%
%%% S % S % S %%%
%%%%%%%%%%%%%%%%%

\section{K-theoretical classification}\label{sec3}

\subsection{Cohomology theories and representing spectra}

A cohomology theory is a contravariant functor $E^n:\texttt{Top}_{*}^{op}\to\texttt{AbGrp}^{\mathds Z}$ equipped with degree -1 natural isomorphisms $\sigma^n:E^{n+1}(X\wedge \mathds S^1)\to E^n X$ satisfying the Eilenberg-Steenrod homotopy invariance, exactness and additive axioms.

Complex and real topological K-theory are cohomology theories induced by the algebraic structure of isomorphism classes of vector bundles over spaces \cite{atiyah2018k}. Topological K-theory can be defined from the Murray-von Neumann category functor \cite{vasconcellos2022operator}. Let $\mathds F=\mathds C,\mathds R$ and $X\in\texttt{Top}$. The Murray-von Neumann category $\texttt{pr}_{\mathds F}X$ has as objects orthogonal projection matrices over $C(X,\mathds F)$, and morphisms are matrices witnessing the Murray-von Neumann relation:
\begin{gather*}\textstyle
    \text{Ob }\texttt{pr}_{\mathds F}X\coloneqq\coprod_{p\in\mathds N}\{P\in M_pC(X,\mathds F)\mid P=P^\dagger=P^2\},\\
    \texttt{pr}_{\mathds F}X(P,Q)\coloneqq\{U\in M_{q,p}C(X,\mathds F)\mid P=U^\dagger U, Q=UU^\dagger\}.
\end{gather*}
Composition is defined by matrix multiplication, and $id_P=P$. The category $\texttt{pr}_{\mathds F}X$ is equivalent to the category of vector bundles over $X$ and bundle isomorphisms. For every $P\in \texttt{pr}_{\mathds F}X$ we have that $\coprod_{x\in X}\Im P(x)\subset X\times \mathds F^p$ is a vector bundle over $X$, and if $U\in \texttt{pr}_{\mathds F}X(P,Q)$ then $\tilde U:\coprod_{x\in X}\Im P(x)\to \coprod_{x\in X}\Im Q(x)$, with $\tilde U(x,v)\coloneqq  (x,U(x)v)$, is a bundle isomorphism. Every vector bundle over $X$ is isomorphic to one induced by an object of $\texttt{pr}_{\mathds F}X$.

The fiberwise direct sum gives us a permutative structure on the category of vector bundles over $X$. This permutative structure is encoded in $\texttt{pr}_{\mathds F}X$ by the direct sum of matrices $U\oplus V
        \coloneqq
        \begin{bsmallmatrix}
            U&0\\0&V
        \end{bsmallmatrix}$. The neutral element is the empty matrix $\boldsymbol 0\coloneqq \mqty[\ \ \ ]$, and the symmetry braiding is $\tau_{P,Q}\coloneqq \begin{bsmallmatrix}
            0&Q\\P&0
        \end{bsmallmatrix}:P\oplus Q\to Q\oplus P$. This permutative structure induces a commutative monoid structure on the decategorification $\pi_0\abs{\texttt{pr}_{\mathds F}X}$, ie the set of isomorphism classes of the objects in $\texttt{pr}_{\mathds F}X$. Applying Grothendieck's abelian group completion we get a contravariant functor, which induces a cohomology theory.

\begin{definition}
    The $K^{00 }_{\mathds F}$-group contravariant functor is
    $$
        K^{00}_{\mathds F}:\texttt{Top}_\ast^{op}\to \texttt{AbGrp}
            ,\qquad K^{00}_{\mathds F}X\coloneqq Gr \pi_0\abs{\texttt{pr}_{\mathds F}X}.
    $$
    Topological K-theory over $\mathds F$ is determined by the functor
    $$
        K_{\mathds F}^n:\texttt{Top}^{op}_\ast\to\texttt{AbGrp}^{\mathds Z_{\leq 0}}
        ,\qquad 
        K_{\mathds F}^n X\coloneqq \begin{cases}
            \ker(i^*_{x_0}:K^{00}_{\mathds F}X\to K^{00}\{x_0\}\cong \mathds F),&n=0\\
            K_{\mathds F}^0(X\wedge \mathds S^{\abs{n}}),&n<0
        \end{cases}
    $$
    equipped with the natural isomorphisms $\sigma^n:K_{\mathds F}^{n+1}(X\wedge \mathds S^1)\to K_{\mathds F}^nX$, induced by the natural isomorphisms $\mathds S^1\wedge\mathds S^{\abs{n}-1}\cong \mathds S^{\abs{n}}$. 

    Bott's periodicity theorem \cite{atiyah1968bott,bott1959stable} tells us complex topological K-theory has periodicity 2, meaning we have a natural isomorphism $K_\mathds C^{n+2}X\cong K_\mathds C^nX$, and that real topological K-theory has periodicity 8, ie $K_\mathds R^{n+8}X \cong K_\mathds R^nX$. This periodicity allows us to define the positive degree K-groups.
\end{definition}

Brown's representability theorem tells us that cohomology theories can be represented by objects called spectra. There are many versions of the category of spectra, but for our purposes the following simple definition will suffice.

\begin{definition}
    A (sequential pre-)spectrum $Y$ is a sequence of pointed topological spaces $\{Y_\bullet\}_{\bullet\in\mathds N}\in \texttt{Top}_*^{\mathds N}$ equipped with structural suspension maps $\{\sigma^Y_\bullet:Y_\bullet\wedge \mathds S^1\to Y_{\bullet+1}\}_{\bullet\in\mathds N}$.

    A spectrum map $\phi:Y\to Z$ is a sequence of pointed maps $\{\phi_\bullet:Y_\bullet\to Z_\bullet\}_{\bullet\in\mathds N}$ such that $\sigma^Z_\bullet(\phi_\bullet(y),\theta)=\phi_{\bullet+1}(\sigma^Y_{\bullet}(y,\theta))$.
\end{definition}

For $n\in\mathds Z$ the $n$-th stable homotopy group of a spectrum $Y$ is
$$\pi^S_n Y=\text{colim}_{\bullet\to \infty}\pi_{n+\bullet}Y_\bullet,
$$
with the colimit induced by the  duals of the structural suspension maps. For $X\in \texttt{Top}_{\ast}$ the mapping spectrum $Y^X$ is defined by
$$
    (Y^X)_\bullet\coloneqq Y^X_\bullet,
    \qquad \sigma^{Y^X}_\bullet(f,\theta)(x)\coloneqq \sigma^Y_\bullet(f(x),\theta). 
$$

Every cohomology theory $E^\bullet$ is represented by a spectrum $Y$, in the sense that for all $X\in\texttt{Top}_*$ we have a natural graded isomorphism 
$$
    E^{n} X\cong\pi^S_{n}Y^X.
$$

A homotopy between spectra maps $f,g:Y\to Z$ is a map $h:Y\wedge I_+\to Z$ such that $h_\bullet(y,0)=f_\bullet(y)$ and $h_\bullet(y,1)=g_\bullet(y)$. A homotopy equivalence is a spectra map that admits an inverse up to homotopy, and spectra that are homotopy equivalent represent the same cohomology theory. In particular, if a subspectrum $Y\subset Z$ is a deformation retract, ie if there is a map $r:Z\to Y$ homotopy equivalent to the identity, then $Y$ and $Z$ represent the same cohomology theory.

Complex and real topological K-theory are represented by spectra $KU$ and $KO$. There are many realizations of these spectra, see for instance \cite{MaEinftyRingSpcsSpectra,schwede2012symmetric}. We now give a definition of $KU$ and $KO$ in terms of spaces of free time evolution operators of grotesque fermion systems, with structural suspension maps induced by the homotopy equivalences in Bott's proof of the periodicity theorem \cite{bott1959stable}. To the best of our knowledge such explicit expressions have not up to now appeared in the literature. In section 4 we will see that spectra $KU^{wi}$ and $KO^{wi}$ constructed from weakly interacting time evolution operators deformation retract to $KU$ and $KO$. Thus, the tenfold way is stable to weak interactions.

\subsection{Construction of $KU$ and $KO$ in terms of time evolution operators}

We first show how the complex topological K-theory spectrum $KU$ can be described in terms of time evolution operators of systems without time reversal or charge-conjugation symmetries. Bott's periodicity theorem for complex K-theory is reflected in the fact that its representing spectrum $KU$ is composed of a 2-periodic sequence of spaces, alternating between classes with only chiral symmetry and without any non-regular symmetries.\\

Any choice of basis for each Hilbert space $\mathscr V_N$ gives us inclusions $\iota_N:\mathscr V_N\hookrightarrow \mathscr V_{N+1}$, with $\iota_N(\ket{i})=\ket{i}$. These induce the maps
$$
    \iota_N:\mathscr M^{000}_N\hookrightarrow\mathscr M^{000}_{N+1},
    \qquad
    \iota_N(a)=\mqty[a&0\\0&1].
$$

We can define 
$$\mathscr M^{000}\coloneqq\text{colim}_{N\in\mathds N}\mathscr M^{000}_N\cong U(\infty).
$$

For each $N,m\in\mathds N$ with $m\leq N$ the inclusions $\iota_N$ induce maps $\iota_N:\mathscr M^{100}_{N,m}\hookrightarrow\mathscr M^{100}_{N+1,m}$. We can then define $\mathscr M^{100}_{\infty,m}\coloneqq \text{colim}_{N\in\mathds N_{\geq m}}\mathscr M^{100}_{N,m}\cong BU(m)$. We also have inclusions $\kappa_{N,m}:\mathscr V_N\hookrightarrow \mathscr V_{N+1}$ with $\kappa_{N,m}(\ket{i})=\ket{i}$ if $i\leq m$, and $\kappa_{N,m}(\ket{i})=\ket{i+1}$ if $i> m$. These induce maps
$$
\kappa_{N,m}:\mathscr M^{100}_{N,m}\hookrightarrow \mathscr M^{100}_{N+1,m+1},
    \qquad
    \kappa_{N,m}\left(\mqty[a&b\\-b^\dagger&d]\right)
    =\mqty[a&0&b\\0&1&0\\-b^\dagger&0&d],
$$
which gives us inclusions $\kappa_m:\mathscr M^{100}_{\infty,m}\hookrightarrow \mathscr M^{100}_{\infty,m+1}$. We can define 
$$\mathscr M^{100}\coloneqq \text{colim}_{m\in\mathds N}\mathscr M^{100}_{\infty,m}\cong BU(\infty).
$$

The structural suspension maps of $KU$ can be constructed by adapting the homotopy equivalences  used by Bott to prove his periodicity theorem \cite{bott1959stable}, which is a consequence of the following proposition:

\begin{proposition}\label{BottProp}
    Let $G$ be a compact Lie group equipped with a Cartan involution $\tau$. Let $s$ be a geodesic segment on the symmetric space $\frac{G}{\text{Fix}_\tau}$ from the coset $\text{Fix}_\tau$ to a coset $g\text{Fix}_\tau$, with $g$ in the normalizer of $\text{Fix}_\tau$. Let $K_s$ be the centralizer of $s$, and $\Omega_s\frac{G}{\text{Fix}_\tau}$ be the component of $s$ in the space of paths from $\text{Fix}_\tau$ to $g\text{Fix}_\tau$. Define the map
    $$ \textstyle 
        f_s:\frac{\text{Fix}_\tau}{K_s}\to \Omega_s\frac{G}{\text{Fix}_\tau},\qquad f_s(x K_s)\coloneqq (\theta\mapsto xs(\theta)x^{-1} \text{Fix}_\tau). 
    $$
    If $s$ contains no conjugate point of $e$ in its interior, then the induced homomorphism
    $$
        \textstyle f_s^\ast:H^\bullet(\Omega_s\frac{G}{\text{Fix}_\tau},\mathds Z_2)\to H^\bullet(\frac{\text{Fix}_\tau}{K_s},\mathds Z_2)
    $$
    is surjective.
\end{proposition}

Bott used this proposition to show that for all $N$ there is a geodesic $s$ in $U(2N)$ such that $f_s:\frac{U(2N)}{U(N)\times U(N)}\to \Omega_s(2N)$ induces an isomorphism of homotopy groups up to dimension $2N$. Taking colimits gives us a homotopy equivalence $BU(\infty)\simeq \Omega_sU(\infty)$.

We can adapt Bott's construction to define the even structural suspensions of $KU$. Consider for $N\in \mathds N$ the geodesic
\begin{gather}
    Ad_{\hat s_{0}(\theta)}\in \mathscr M^{000}_{2N}
    ,\qquad \textstyle
    \hat s_{0}(\theta)
    \coloneqq 
    \exp(-i\pi \theta (\sum_{j=1}^{N}a_j^\dagger a_j - a_{N+j}^\dagger a_{N+j})   )
    ,\label{s0}\\
    [Ad_{\hat s_{0}(\theta)}]_{\mathcal B}=\exp(-i\pi \theta  1_{N,N})=\mqty[e^{-i \pi \theta}1_{N}&0\\0&e^{i \pi \theta}1_{N}].
    \nonumber
\end{gather}

Since the centralizer of $\hat s_0$ is $\text{Fix}_\Sigma$ the map
\begin{gather*}
    \sigma'_{0}:\frac{\mathscr M^{000}_{2N}}{\text{Fix}_\Sigma}\wedge \mathds S^1\to \mathscr M^{000}_{2N}
    ,\\
    \sigma'_{0}(Ad_{\exp(-i\hat H)}\text{Fix}_\Sigma,\theta)\coloneqq [Ad_{\exp(-i\hat H)},Ad_{\hat s_{0}(\theta)}]=Ad_{[\exp(-i\hat H),\hat s_{0}(\theta)]}
\end{gather*}
is well defined. Taking colimits we can define
\begin{gather*}
    \sigma_0:\mathscr M^{100}\wedge \mathds S^1\to \mathscr M^{000},
    \nonumber\\ \sigma_{0}(Ad_{\exp(-i\hat H)},\theta)\coloneqq \sigma_0'(c_\Sigma^{-1}Ad_{\exp(-i\hat H)},\theta)=Ad_{[\exp(\sfrac{-i}{2}\hat H),\hat s_0(\theta)]}.
\end{gather*}
Under the convenient basis of the previous section we have
$$
    \sigma_0\left(\exp(-i\mqty[0&b\\b^\dagger&0]),\theta\right)=\exp(\frac{-i}{2}\mqty[0&b\\b^\dagger&0])\exp(\frac{i}{2}\mqty[0&e^{-i2\pi\theta}b\\e^{i2\pi\theta}b^\dagger&0]).
$$
This structural suspension induces a homotopy equivalence of $\mathscr M^{100}$ with the connected component of the trivial loop in $\Omega\mathscr M^{000}$.\\

The total spaces $EU(m)$ of the universal $U(m)$-principal bundle are contractible, and these contractions induce a homotopy equivalence $U(\infty)\simeq \Omega BU(\infty)$. This gives us the odd structural suspensions of $KU$.

Recall that the chiral symmetry $\hat S$ induces orthogonal decomposition $\mathscr V_N=\mathscr V_N^+\oplus\mathscr V_N^-$. Let $\hat N_+\coloneqq \sum_{j=1}^ma_j^\dagger a_j$ and $\hat N_-\coloneqq\sum_{j=m+1}^Na_j^\dagger a_j$, which restricted to $\mathscr V_N$ give us the orthogonal projections onto $\mathscr V_N^+$ and $\mathscr V_N^-$, respectively. Consider the subgroups 
\begin{gather*}
    \mathscr M_{N,m}^+\coloneqq\{Ad_{\hat U}\in\mathscr M^{000}_N\mid \hat N_+Ad_{\hat U}\hat N_+ +\hat N_-=Ad_{\hat U}\},\\
    \mathscr M_{N,m}^-\coloneqq\{Ad_{\hat U}\in\mathscr M^{000}_N\mid \hat N_++\hat N_-Ad_{\hat U}\hat N_-=Ad_{\hat U}\}.
\end{gather*}

Note that $\text{Fix}_{\Sigma}=\mathscr M_{N,m}^+\mathscr M_{N,m}^-=\mathscr M_{N,m}^-\mathscr M_{N,m}^+$. Under the isomorphism $\mathscr M^{000}_N\cong U(N)$, we have $\mathscr M_{N,m}^+\cong U(m)\times\{1_{N-m}\}$ and $\mathscr M_{N,m}^-\cong \{1_m\}\times U(N-m)$. We can define $\mathscr E^{100}_{N,m}\coloneqq \frac{\mathscr M^{000}_N}{\mathscr M_{N,m}^-}$. The inclusions $\iota_N$ induce maps $\mathscr E^{100}_{N,m}\hookrightarrow \mathscr E^{100}_{N+1,m}$, and we define $\mathscr E^{100}_{\infty,m}\coloneqq \text{colim}_{N\in\mathds N_{\geq m}}\mathscr E^{100}_{N,m}\cong EU(m)$.

We then have fibrations $\pi_m:\mathscr E^{100}_{\infty,m}\twoheadrightarrow \mathscr M^{100}_{\infty,m}$ with fibers $\mathscr M^+_{\infty,m}\coloneqq \text{colim}_{N\in\mathds N_{\geq m}}\mathscr M^+_{N,m}\cong \mathscr M^{000}_m\cong U(m)$. The inclusions $\kappa_{m,N}$ give us maps $\kappa_m:\mathscr E_{\infty,m}^{100}\to \mathscr E_{\infty,m+1}^{100}$, and we define $\mathscr E^{100}\coloneqq \text{colim}_{m\in\mathds N}\mathscr E_{\infty,m}^{100}\cong EU(\infty)$. The fibrations $\pi_m$ then induce a fibration $\pi:\mathscr E^{100}\twoheadrightarrow\mathscr M^{100}$, with fibers $\mathscr M^+\coloneqq\text{colim}_{m\in\mathds N}\mathscr M^+_{\infty,m}\cong \mathscr M^{000}\cong U(\infty)$, which is diffeomorphic to the universal $U(\infty)$-principal bundle. We then set the maps
$$
\iota_1:\mathscr H^{000}_N\to \mathscr H^{000}_{2N},\qquad \iota_1(H)\coloneqq \mqty[H&0\\0&0],
$$
which induce a fiber inclusion $\mathscr M^{000}\hookrightarrow\mathscr E^{100}$.
For $N,m\in\mathds N$ with $m\leq N$ consider the geodesic
\begin{gather*}\textstyle
    Ad_{\hat s_{1}(\theta)}\in\mathscr M^{000}_{N+m},
    \\ \hat s_{1}(\theta)\coloneqq 
    \exp(\sfrac{-i\pi \theta}{2}(\textstyle\sum_{j=1}^m a_j^\dagger a_j-a_{j}^\dagger a_{N+j}-a_{N+j}^\dagger a_j+a_{N+j}^\dagger a_{N+j})),\\
    [Ad_{\hat s_{1}(\theta)}]_{\mathcal B}=\exp(\frac{-i\pi \theta}{2} \mqty[1_m&0&-1_m\\0&0&0\\-1_m&0&1_m])
    =\mqty[
        \frac{1+e^{-i\pi \theta}}{2}1_m&0&\frac{1-e^{-i\pi \theta}}{2}1_m\\
        0&1_{N-m}&0\\
        \frac{1-e^{-i\pi \theta}}{2}1_m&0&\frac{1+e^{-i\pi \theta}}{2}1_m]
\end{gather*}

We can then define
\begin{gather*}
    \tilde \sigma'_{1}:\mathscr E^{100}_{N,m}\wedge \mathds S^1 \to \frac{\mathscr M^{000}_{N+m,m}}{\text{Fix}_\Sigma}
,\\ \tilde\sigma'_{1}(Ad_{\exp(-i\hat H)}\mathscr M_{N,m}^-,\theta)\coloneqq Ad_{\exp(-i\hat s_{1}(\theta)\iota_1(\hat H)\hat s_{1}(\theta)^\dagger)}\text{Fix}_\Sigma,
\end{gather*}
since elements of $\mathscr M^-_{N,m}\subset \mathscr M^-_{N+m,m}$ commute with $\hat s_{1}(\theta)$. This map induces a homotopy from $\iota \pi:\mathscr E^{100}_{N,m}\twoheadrightarrow \mathscr M^{100}_{N+m,m}$ to the constant map at the base point. 
Considering that under our choice of basis $\mathscr M^{000}_{N}$ is identified with the fiber of $\pi_N:\mathscr E^{100}_{\infty,N}\twoheadrightarrow \mathscr M^{100}_{\infty,N}$, we can define
\begin{gather*}
    \sigma_1:\mathscr M^{000}\wedge \mathds S^1\to  \mathscr M^{100}, \qquad \sigma_1(Ad_{\exp(-i\hat H)},\theta)\coloneqq [Ad_{\exp(-i\hat s_{1}(\theta)\iota_1(\hat H)\hat s_{1}(\theta)^\dagger)},\hat S]
    \\
    \sigma_1\left(\exp(-iH),\theta\right)=\left[\exp(-i\mqty[\cos(\frac{\pi \theta}{2})^2H&\frac{i}{2}\sin(\pi\theta)H\\\frac{-i}{2}\sin(\pi\theta) H&\sin(\frac{\pi \theta}{2})^2 H]),\hat S\right].
\nonumber
\end{gather*}

\begin{definition}
    The \textit{complex topological K-theory spectrum $KU$} is
    \begin{gather*}
        KU_{\bullet}\coloneqq \begin{cases}
        \mathscr M^{100}\cong BU(\infty),&\bullet=0\mod 2\\
        \mathscr M^{000}\cong U(\infty),&\bullet =1\mod 2
    \end{cases},
    \\
    \sigma_\bullet(Ad_{\exp(-i\hat H)},\theta)\coloneqq \begin{cases}
    Ad_{[\exp(\sfrac{-i}{2}\hat H),\hat s_0(\theta)]},&\bullet=0\mod 2\\
        [Ad_{\exp(-i\hat s_{1}(\theta)\iota_1(\hat H)\hat s_{1}(\theta)^\dagger)},\hat S],&\bullet=1\mod 2
    \end{cases}.
    \end{gather*}
\end{definition}

We now describe the real topological K-theory spectrum $KO$ in terms of time evolution operators of systems with time-inversion and/or charge-conjugation symmetries.\\

The maps $\iota_N$ induce inclusions $\mathscr M^{010}_N\hookrightarrow \mathscr M^{010}_{N+1}$ and $\mathscr M^{001}_N\hookrightarrow \mathscr M^{001}_{N+1}$. We can then define 
$$\mathscr M^{010}\coloneqq \text{colim}_{N\in\mathds N}\mathscr M^{010}_N\cong O(\infty)
\quad \text{ and }\quad 
\mathscr M^{001}\coloneqq \text{colim}_{N\in\mathds N}\mathscr M^{001}_N\cong\frac{U(\infty)}{O(\infty)}.
$$

The maps $\iota_N$ and $\kappa_{m,N}$ induce inclusions $\mathscr M^{111}_{N,m}\hookrightarrow \mathscr M^{111}_{N+1,m}$ and $\mathscr M^{111}_{N,m}\hookrightarrow \mathscr M^{111}_{N+1,m+1}$. We then define 
$$\mathscr M^{111}\coloneqq \text{colim}_{m\in\mathds N,N\in \mathds N_{\geq m}}\mathscr M^{111}_{N,m}\cong BO(\infty).
$$

The maps $\iota_{N+1}\kappa_{m+\sfrac{N-m}{2},N}=\kappa_{m+\sfrac{N-m}{2},N+1}\iota_N$ and $\kappa_{m+1,N+1}\kappa_{\sfrac{m}{2},N}=\kappa_{\sfrac{m}{2},N+1}\kappa_{m,N}$ induce inclusions $\mathscr M^{1-1-1}_{N,m}\hookrightarrow \mathscr M^{1-1-1}_{N+2,m}$ and $\mathscr M^{1-1-1}_{N,m}\hookrightarrow \mathscr M^{1-1-1}_{N+2,m+2}$. We then define 
$$\mathscr M^{1-1-1}\coloneqq \text{colim}_{m\in2\mathds N,N\in 2\mathds N_{\geq m}}\mathscr M^{1-1-1}_{N,m}\cong BSp(\sfrac{\infty}{2}).
$$

The maps $\iota_{N+1}\kappa_{\sfrac{N}{2},N}=\kappa_{\sfrac{N}{2},N+1}\iota_{N}$ induce inclusions $\mathscr M^{0-10}_N\hookrightarrow\mathscr M^{0-10}_{N+2}$, $\mathscr M^{00-1}_N\hookrightarrow\mathscr M^{00-1}_{N+2}$, $\mathscr M^{1-11}_N\hookrightarrow\mathscr M^{1-11}_{N+2}$ and $\mathscr M^{11-1}_N\hookrightarrow\mathscr M^{11-1}_{N+2}$.  We then define 
\begin{align*}
    \mathscr M^{0-10}&\coloneqq \text{colim}_{N\in 2\mathds N}\mathscr M^{0-10}_{N}\cong Sp(\sfrac{\infty}{2})
    ,&
    \mathscr M^{00-1}
    &\coloneqq \text{colim}_{N\in 2\mathds N}\mathscr M^{00-1}_{N}\cong\frac{U(\infty)}{Sp(\sfrac{\infty}{2})},\\
    \mathscr M^{1-11}
    &\coloneqq \text{colim}_{N\in 2\mathds N}\mathscr M^{1-11}_{N}\cong\frac{Sp(\sfrac{\infty}{2})}{U(\sfrac{\infty}{2})}
    ,&
    \mathscr M^{11-1}
    &\coloneqq \text{colim}_{N\in 2\mathds N}\mathscr M^{11-1}_{N}\cong\frac{O(\infty)}{U(\sfrac{\infty}{2})}.
\end{align*}

Setting $\hat s_0(\theta)$ as in (\ref{s0}), the 0th structural suspension map is
\begin{gather*}
    \sigma_0:\mathscr M^{111}\wedge \mathds S^1\to \mathscr M^{001}
    ,\qquad
    \sigma_0(Ad_{\exp(-i\hat H)},\theta)\coloneqq [Ad_{[\exp(\sfrac{-i}{2}\hat H),\hat s_0(\theta)]},\hat T],\\
    \sigma_0\left(\exp(-i\mqty[0&b\\b^t&0]),\theta\right)=\left[\exp(\frac{-i}{2}\mqty[0&b\\b^t&0])\exp(\frac{i}{2}\mqty[
        0&e^{-i2\pi\theta}b\\e^{i2\pi\theta}b^t&0]),\hat T\right].
\end{gather*}

In order to define the 1st structural suspension map, note first we have an inclusion
$$
    \iota_1:\mathscr H^{001}_N\to \mathscr H^{000}_{2N}
    ,\qquad \iota_1(H)\coloneqq \mqty[
    H&0\\0&-H]
$$ 
such that $\iota_1(\hat H)$ commutes with $\hat C$ in the class CI $(1\text{-}11)$.
Setting 
\begin{gather*}
    \textstyle Ad_{\hat s_1(\theta)}\in \mathscr M^{1-11}_{2N}
,
\qquad 
\hat s_1(\theta)\coloneqq \exp(-i\pi \theta(\sum_{j=1}^Na_{j}^\dagger a_{N+j}+a_{N+j}^\dagger a_{j})),\\ 
[\textstyle Ad_{\hat s_1(\theta)}]_{\mathcal B}=\exp(-i\pi \theta F_{2N})
=\mqty[\cos(\pi \theta)1_N&-i\sin(\pi \theta)1_N\\-i\sin(\pi \theta)1_N&\cos(\pi \theta)1_N]
\end{gather*}
we can define
\begin{gather*}
    \sigma_1:\mathscr M^{001}\wedge \mathds S^1\to \mathscr M^{1-11}
    ,\qquad
    \sigma_1(Ad_{\exp(-i\hat H)},\theta)\coloneqq [Ad_{[\exp(\sfrac{-i}{2}\iota_1(\hat H)),\hat s_1(\theta)]},\hat T]\\
    \sigma_1\left(\exp(-iH),\theta\right)
    =\left[\exp(\frac{-i}{2}\mqty[H&0\\0&-H])\exp(\frac{i}{2}\mqty[\cos(2\pi\theta)H
    &i \sin(2\pi\theta)H
    \\
    -i \sin(2\pi\theta)H&-\cos(2\pi\theta)H]),\hat T\right].
\end{gather*}

Setting 
\begin{gather*}
    \textstyle
    Ad_{\hat s_2(\theta)}\in\mathscr M^{0-10}_N
    ,
    \qquad 
    \hat s_2(\theta)\coloneqq \exp(-i\pi \theta(\sum_{j=1}^{\sfrac{N}{2}}a_{j}^\dagger a_{j}-a_{\sfrac{N}{2}+j}^\dagger a_{\sfrac{N}{2}+j})),\\
    [Ad_{\hat s_2(\theta)}]_{\mathcal B}=\exp(-i\pi \theta 1_{\sfrac{N}{2},\sfrac{N}{2}})=\mqty[e^{-i\pi \theta}1_{\sfrac{N}{2}}&0\\0&e^{i\pi \theta}1_{\sfrac{N}{2}}]
\end{gather*}
the 2nd structural suspension map is
\begin{gather*}
    \sigma_2:\mathscr M^{1-11}\wedge \mathds S^1\to \mathscr M^{0-10}
    ,\qquad
    \sigma_2(Ad_{\exp(-i\hat H)},\theta)\coloneqq Ad_{[\exp(\sfrac{-i}{2}\hat H),\hat s_2(\theta)]},\\
    \sigma_2\left(\exp(-i\mqty[0&b\\b^*&0]),\theta\right)=\exp(\frac{-i}{2}\mqty[0&b\\b^*&0])\exp(\frac{i}{2}\mqty[
        0&e^{-i2\pi\theta}b\\e^{i2\pi\theta}b^*&0]).
\end{gather*}

To define the 3rd structural suspension  we set
\begin{gather*}\textstyle
    Ad_{\hat s_3(\theta)}\in\mathscr M^{000}_{2N},
    \\ 
    \hat s_{3}(\theta)\coloneqq 
    \exp(\sfrac{-i\pi \theta}{2}(\textstyle\sum_{j=1}^N a_j^\dagger a_j-a_{j}^\dagger a_{N+j}-a_{N+j}^\dagger a_j+a_{N+j}^\dagger a_{N+j})),\\
    [Ad_{\hat s_3(\theta)}]_{\mathcal B}=\exp(\frac{-i\pi \theta}{2} \mqty[1_N&-1_N\\-1_N&1_N])
    =\mqty[
        \frac{1+e^{-i\pi \theta}}{2}1_N&\frac{1-e^{-i\pi \theta}}{2}1_N\\
        \frac{1-e^{-i\pi \theta}}{2}1_N&\frac{1+e^{-i\pi \theta}}{2}1_N]
\end{gather*}
as in the 1st suspension map of $KU$. We also set the inclusion
$$
\iota_3:\mathscr H^{0-10}_N\to \mathscr H^{000}_{2N},\qquad \iota_3(H)\coloneqq \mqty[H&0\\0&0],
$$
such that $\iota_3(\hat H)$ commutes with $\hat C$ in the class CII $(1\text{-}1\text{-}1)$.
This lets us define
\begin{gather*}
    \sigma_3:\mathscr M^{0-10}\wedge \mathds S^1\to \mathscr M^{1-1-1}, \qquad \sigma_3(Ad_{\exp(-i\hat H)},\theta)\coloneqq [Ad_{\exp(-i\hat s_{3}(\theta)\iota_3(\hat H)\hat s_{3}(\theta)^\dagger)},\hat S],\\
    \sigma_3\left(\exp(-iH),\theta\right)=\left[\exp(-i\mqty[\cos(\frac{\pi \theta}{2})^2 H&\frac{i}{2}\sin(\pi\theta) H\\\frac{-i}{2}\sin(\pi\theta) H&\sin(\frac{\pi \theta}{2})^2 H]),\hat S\right].
\end{gather*}

Letting $V$ be as in (\ref{V1-1-1}) then conjugation by $V$ induces the inclusion
$$
    \iota_4:\mathscr H^{1-1-1}_{2N,N}\to\mathscr H_N^{000},\qquad \iota_4\left(\mqty[
         0&0&b_0&b_1\\
         0&0&-b_1^*&b_0^*\\
             b_0^\dagger&-b_1^t&0&0\\b_1^\dagger&b_0^t&0&0]\right)
    \coloneqq\mqty[
        0&ib_0&0&ib_1\\-ib_0^\dagger&0&ib_1^t&0\\0&-ib_1^*&0&ib_0^*\\-ib_1^\dagger&0&-ib_0^t&0].
$$
Setting
\begin{gather*}\textstyle
    Ad_{\hat s_4(\theta)}\in\mathscr M^{00-1}_{2N},
    \\ 
    \hat s_4(\theta)\coloneqq \exp(-i\pi\theta(\textstyle\sum_{j=1}^{\sfrac{N}{2}}a_j^\dagger a_j-a_{\sfrac{N}{2}+j}^\dagger a_{\sfrac{N}{2}+j}+a_{N+j}^\dagger a_{N+j}-a_{\sfrac{3N}{2}+j}^\dagger a_{\sfrac{3N}{2}+j})),\\
    [Ad_{\hat s_4(\theta)}]_{\mathcal B}=\exp(-i\pi\theta\mqty[1_{\sfrac{N}{2},\sfrac{N}{2}}&0\\0&1_{\sfrac{N}{2},\sfrac{N}{2}}])
    =\mqty[e^{-i\pi\theta}1_{\sfrac{N}{2}}&0&0&0\\0&e^{i\pi\theta}1_{\sfrac{N}{2}}&0&0\\0&0&e^{-i\pi\theta}1_{\sfrac{N}{2}}&0\\0&0&0&e^{i\pi\theta}1_{\sfrac{N}{2}}]
\end{gather*}
the 4th structural suspension map is
\begin{gather*}
    \sigma_4:\mathscr M^{1-1-1}\wedge\mathds S^1\to \mathscr M^{00-1}
    ,\qquad
    \sigma_4(Ad_{\exp(-i\hat H)},\theta)\coloneqq [Ad_{[\exp(\sfrac{-i}{2}\iota_4(\hat H)),\hat s_4(\theta)]},\hat T],\\
    \sigma_4\left(\exp(-i\mqty[
         0&0&b_0&b_1\\
         0&0&-b_1^*&b_0^*\\
             b_0^\dagger&-b_1^t&0&0\\b_1^\dagger&b_0^t&0&0]),\theta\right)=\\\left[\exp(\frac{-i}{2}\begin{bsmallmatrix}
        0&ib_0&0&ib_1\\-ib_0^\dagger&0&ib_1^t&0\\0&-ib_1^*&0&ib_0^*\\-ib_1^\dagger&0&-ib_0^t&0\end{bsmallmatrix})\exp(\frac{i}{2}\begin{bsmallmatrix}
        0&e^{-i2\pi\theta}ib_0&0&e^{-i2\pi\theta}ib_1\\-e^{i2\pi\theta}ib_0^\dagger&0&e^{i2\pi\theta}ib_1^t&0\\0&-e^{-i2\pi\theta}ib_1^*&0&e^{-i2\pi\theta}ib_0^*\\-e^{i2\pi\theta}ib_1^\dagger&0&-e^{i2\pi\theta}ib_0^t&0
    \end{bsmallmatrix}),\hat T\right].
\end{gather*}

To define the 5th structural map we have to consider the mappings
$$
\iota_5:\mathscr H^{00-1}_N\to \mathscr H^{000}_{2N}
,\qquad \iota_5\left(\mqty[a&b\\-b^*&a^*]\right)
\coloneqq \mqty[
    a&b&0&0\\-b^*&a^*&0&0\\0&0&-a^*&-b^*\\0&0&b&-a],
$$
such that $\iota_5(\hat H)$ commutes with $\hat C$ in the class DIII $(11\text{-}1)$. Definig 
\begin{gather*}\textstyle
    Ad_{\hat s_5(\theta)}\in\mathscr M^{11-1}_{2N},
    \\
    \hat s_5(\theta)\coloneqq \exp(-i\pi \theta(\textstyle\sum_{j=1}^{\sfrac{N}{2}}  a_{j}^\dagger a_{\sfrac{N}{2}+j}-a_{\sfrac{N}{2}+j}^\dagger a_{j}- a_{N+j}^\dagger a_{\sfrac{3N}{2}+j}+a_{\sfrac{3N}{2}+j}^\dagger a_{N+j})),\\
    [Ad_{\hat s_5(\theta)}]_{\mathcal B}=\exp(-i\pi \theta \mqty[0&J_N\\-J_N&0])
    =\begin{bsmallmatrix}
    \cos(\pi \theta)1_{\sfrac{N}{2}}&0&0&-i\sin(\pi \theta)1_{\sfrac{N}{2}}\\0&\cos(\pi \theta)1_{\sfrac{N}{2}}&i\sin(\pi \theta)1_{\sfrac{N}{2}}&0\\
    0&i\sin(\pi \theta)1_{\sfrac{N}{2}}&\cos(\pi \theta)1_{\sfrac{N}{2}}&0\\
    -i\sin(\pi \theta)1_{\sfrac{N}{2}}&0&0&\cos(\pi \theta)1_{\sfrac{N}{2}}
    \end{bsmallmatrix}
\end{gather*}
we have
\begin{gather*}
    \sigma_5:\mathscr M^{00-1}\wedge\mathds S^1\to \mathscr M^{11-1}
    ,\qquad
    \sigma_5(Ad_{\exp(-i\hat H)},\theta)\coloneqq [Ad_{[\exp(\sfrac{-i}{2}\iota_5(\hat H)),\hat s_5(\theta)]},\hat T]\\
    \sigma_5\left(\exp(-i\mqty[a&b\\-b^*&a^*]),\theta\right)
    =\\
    \left[\exp(\frac{-1}{2}\begin{bsmallmatrix}
        a&b&0&0\\-b^*&a^*&0&0\\0&0&-a^*&-b^*\\0&0&b&-a
    \end{bsmallmatrix})\exp(\frac{1}{2}\begin{bsmallmatrix}
        \cos(2\pi\theta)a&\cos(2\pi\theta)b&-i\sin(2\pi\theta)b&i\sin(2\pi\theta)a\\-\cos(2\pi\theta)b^*&\cos(2\pi\theta)a^*&-i\sin(2\pi\theta)a^*&-i\sin(2\pi\theta)b^*\\-i\sin(2\pi\theta)b^*&i\sin(2\pi\theta)a^*&\cos(2\pi\theta)a^*&\cos(2\pi\theta)b^*\\-i\sin(2\pi\theta)a&-i\sin(2\pi\theta)b&-\cos(2\pi\theta)b&\cos(2\pi\theta)a
    \end{bsmallmatrix}),\hat T\right].
\end{gather*}

Defining
$$
\iota_6:\mathscr H_N^{11-1}\to \mathscr H_N^{010},\qquad \iota_6\left(\mqty[0&b\\-b^*&0]\right)\coloneqq\mqty[i\Im(b)&-i\Re(b)\\-i\Re(b)&-i\Im(b)]
$$
and
\begin{gather*}
    \textstyle
    Ad_{\hat s_6(\theta)}\in\mathscr M^{010}
    ,
    \qquad 
    \hat s_6(\theta)\coloneqq \exp(-i\pi \theta(\sum_{j=1}^{\sfrac{N}{2}}ia_j^\dagger a_{\sfrac{N}{2}+j}-i a_{\sfrac{N}{2}+j}^\dagger a_j)),\\
    [Ad_{\hat s_6(\theta)}]_{\mathcal B}=\exp(\pi \theta J_N)=\mqty[
    \cos(\pi \theta)1_{\sfrac{N}{2}}&\sin(\pi \theta)1_{\sfrac{N}{2}}\\-\sin(\pi \theta)1_{\sfrac{N}{2}}&\cos(\pi \theta)1_{\sfrac{N}{2}}
]
\end{gather*}
the 6th structural suspension map is
\begin{gather*}
    \sigma_6:\mathscr M^{11-1}_N\wedge\mathds S^1\to \mathscr M^{010}_N,\qquad \sigma_6(Ad_{\exp(-i\hat H)},\theta)\coloneqq Ad_{[\exp(-\sfrac{i}{2}\iota_6(\hat H)),\hat s_6(\theta)]}\\
    \sigma_6\left(\exp(-i\mqty[0&b\\-b^*&0]),\theta\right)
    =\\
    \exp(\frac{1}{2}\begin{bsmallmatrix}
        \Im(b)&-\Re(b)\\-\Re(b)&-\Im(b)
    \end{bsmallmatrix})\exp(\frac{-1}{2}\begin{bsmallmatrix}
        -\sin(2\pi\theta)\Re(b)+\cos(2\pi\theta)\Im(b)&-\cos(2\pi\theta)\Re(b)-\sin(2\pi\theta)\Im(b)\\-\cos(2\pi\theta)\Re(b)-\sin(2\pi\theta)\Im(b)&\sin(2\pi\theta)\Re(b)-\cos(2\pi\theta)\Im(b)
    \end{bsmallmatrix})
\end{gather*}

Defining $\hat s_7(\theta)$ and $\iota_7$ as as in the 1st suspension map of $KU$ the 7th structural suspension map is
\begin{gather*}
    \sigma_7:\mathscr M^{010}\wedge\mathds S^1\to \mathscr M^{111}, \qquad \sigma_7(Ad_{\exp(-i\iota_7(\hat H))},\theta)\coloneqq [Ad_{\exp(-i\hat s_{7}(\theta)\iota_7(\hat H)\hat s_{7}(\theta)^\dagger)},\hat S],\\
    \sigma_7\left(\exp(-iH),\theta\right)=\left[\exp(-i\mqty[\cos(\frac{\pi \theta}{2})^2H&\frac{i}{2}\sin(\pi\theta)H\\\frac{-i}{2}\sin(\pi\theta) H&\sin(\frac{\pi \theta}{2})^2 H]),\hat S\right].
\end{gather*}

\begin{definition}
    The real topological K-theory spectrum $KO$ is
    \begin{gather*}
    KO_{\bullet}\coloneqq \begin{cases}
        \mathscr M^{111}\cong BO(\infty),&\bullet=0\mod 8\\
        \mathscr M^{001}\cong \frac{U(\infty)}{O(\infty)},&\bullet =1\mod 8\\
        \mathscr M^{1-11}\cong \frac{Sp(\sfrac{\infty}{2})}{U(\sfrac{\infty}{2})},&\bullet=2\mod 8\\
        \mathscr M^{0-10}\cong Sp(\sfrac{\infty}{2}),&\bullet =3\mod 8\\
        \mathscr M^{1-1-1}\cong BSp(\sfrac{\infty}{2}),&\bullet=4\mod 8\\
        \mathscr M^{00-1}\cong \frac{U(\infty)}{Sp(\sfrac{\infty}{2})},&\bullet =5\mod 8\\
        \mathscr M^{11-1}\cong \frac{O(\infty)}{U(\sfrac{\infty}{2})},&\bullet=6\mod 8\\
        \mathscr M^{010}\cong O(\infty),&\bullet =7\mod 8
    \end{cases},\\
    \sigma_\bullet(Ad_{\exp(-i\hat H)},\theta)\coloneqq \begin{cases}
        [Ad_{[\exp(\sfrac{-i}{2}\iota_\bullet(\hat H)),\hat s_\bullet(\theta)]},\hat T],&\bullet=0,1,4,5\mod 8\\
        Ad_{[\exp(\sfrac{-i}{2}\iota_\bullet(\hat H)),\hat s_\bullet(\theta)]},&\bullet=2,6\mod 8\\
        [Ad_{\exp(-i\hat s_{\bullet}(\theta)\iota_\bullet(\hat H)\hat s_{\bullet}(\theta)^\dagger)},\hat S],&\bullet=3,7\mod 8
    \end{cases},
    \end{gather*}
    where we assume $\iota_0$ and $\iota_2$ to be the identities.
\end{definition}

%%%%%%%%%%%%%%%%%
%%% S % S % S %%%
%%%%%%%%%%%%%%%%%

\section{Weakly interacting systems}\label{sec4}

As is standard, to model interacting systems we must consider operators in the Clifford subalgebra $Cl^+(\mathscr W_N,b)=\bigoplus_{n\in \mathds N}Cl^{2n}(\mathscr W_N,b)$ generated  by even degree monomials of creation and annihilation operators. These are the operators that commute with the parity operator $\hat P\coloneqq \exp(i\pi \hat N)$, thus preserve the parity in the number of fermions. See \cite{friis2016reasonable,book-bru-pedra,bratteli-dois,wick1952intrinsic} for justifications for this assumption.

\begin{definition}
    Let $(\mathscr W_N,b)$ be a Nambu space. The \textit{space of interacting Hamiltonians} is
$$\textstyle
\mathscr H_{i,N}\coloneqq \{\hat H\in Cl^+(\mathscr W_N,b) \mid \hat H^\dagger=\hat H\}.
$$
The \textit{space of interacting time evolution operators} is
$$
    \mathscr M_{i,N}\coloneqq \{Ad_{\exp(-i \hat H)}\in Aut(Cl(\mathscr W_N,b))\mid \hat H\in \mathscr H_{i,N}\}.
$$
\end{definition}

% In \cite{kitaev2009periodic,kennedy2015bott} the free Hamiltonians representing topological phases are assumed gapped, and weakly interacting systems are those where the interaction is too weak to close the gap. To illustrate why strong interactions break the tenfold way classification Kitaev provides the following example of continuous path of interacting Hamiltonians:
% $$\textstyle
%     \hat H(t)=\cos(\pi t)\sum_{j=1}^4a^\dagger_ja_j+\sin(\pi t)(a_1a_2a_3a_4+a_4^\dagger a_3^\dagger a_2^\dagger a_1^\dagger)\in \mathscr H_{i,4}.
% $$
% This gives a path of non-degenerate ground states
% starting at the vacuum state generated by $\ket{0}$ and ending ate the fully-occupied state generated by $\ket{1234}$. It is impossible for a path of free Hamiltonians to connect

% In a free system such a path would only be possible if the energy gap closes, but in the above example the interaction term allows\todo{terminar} the energy gap at some point.

In the context of time evolution operators we can give a geometric definition of weak interactions. We will consider an operator weakly interacting if there is an unambiguous closest free operator to it, and a unique distance minimizing geodesic between them. This condition is satisfied by elements in the \textit{complement of the cut locus} of $\mathscr M_N$ in $\mathscr M_{i,N}$.
We recall here the definition of the cut locus of a submanifold
\cite{prasad2022cut,thom1972cut}.

\begin{definition}
    Let $\mathcal N$ be a compact submanifold of a complete manifold $\mathcal M$. The \textit{separating set $Se(\mathcal N)\subset\mathcal M$ of $\mathcal N$} consists of all points $u$ such that at least two distance minimizing geodesics from $\mathcal N$ to $u$ exist. 
    
    A point $u$ is in the \textit{cut locus $Cu(\mathcal N)\subset \mathcal M$ of $\mathcal N$} if there is some distance minimizing geodesic joining $\mathcal N$ to $u$  such that any extension of it beyond $u$ is not a distance minimizing geodesic.
\end{definition}

Since $Cu(\mathcal N)=\overline{Se(\mathcal N)}$ the cut locus is closed. To express our geometric definition of weak interactions we need to consider the orthogonal complement $\mathscr X_{N}\subset \mathscr H_{i,N}$ of $\mathscr H_N$. 
    
    \begin{definition}
        Let $(\mathscr W_N,b)$ be a Nambu space. The \textit{space of weakly interacting time evolution operators} is
        \begin{align*}
            &\mathscr M_{wi,N}\coloneqq\\
            &
        \left\{Ad_{\hat U}\in \mathscr M_{i,N}\mid \exists (Ad_{\hat U_0},\hat X)\in \mathscr M_N\times \mathscr X_{N}
            ,\forall t\in[0,1]:
            \mqty{Ad_{\hat U}=Ad_{\hat U_0\exp(-i\hat X)}
            ,\\
            Ad_{\hat U_0\exp(-it\hat X)}\in \mathscr M_{i,N}\setminus Cu(\mathscr M_N)}
            \right\}.
        \end{align*}
    \end{definition}

    \begin{example}
        Consider
        $$
            \hat U(t)\coloneqq \exp(i\pi (a_1^\dagger a_1+a_2^\dagger a_2))\exp(-i2\pi t(a_1a_2+ a_2^\dagger a_1^\dagger)).
        $$ 

            In the basis $\mathcal C=\{\ket{0},\ket{1},\ket{2},\ket{12}\}$ of $\bigwedge\mathscr V_2$  we have
    $$
        [Ad_{\hat U(t)}]_{\mathcal C}=\mqty[\cos(2\pi t)&0&0&-i\sin(2\pi t)\\
        0&-1&0&0\\
        0&0&-1&0\\
        -i\sin(2\pi t)&0&0&\cos(2\pi t)].
    $$
    Since 
    $Ad_{\hat U(0)}=Ad_{\hat U(1)}=Ad_{\hat P}\in\mathscr M_{2}$ and $Ad_{\hat U(t)}\in \mathscr M_{i,2}\setminus \mathscr M_{2}$ for $t\in (0,1)$, then $Ad_{\hat U(\sfrac{1}{2})}=-\hat 1_{\bigwedge \mathscr V_2}\in Cu(\mathscr M_{2})$. Thus $-\hat 1_{\bigwedge \mathscr V_2}$ is an interacting time evolution operator that is not weakly interacting.
    \end{example}

    If $Ad_{\hat U}=Ad_{\hat U_0\exp(-i\hat X)}$ is weakly interacting then $Ad_{\hat U_0}$ is the element of $\mathscr M_{N}$ that is closest to $Ad_{\hat U}$ in the geodesic metric. The cut locus assumption guarantees $Ad_{\hat U_0}$ is unique, and further that $\mathscr M_{wi,N}$ strongly deformation retracts to $\mathscr M_{N}$ by the homotopy
    \begin{equation}
        h:\mathscr M_{wi,N}\wedge I_+\to \mathscr M_{wi,N}
        ,\qquad
        h(Ad_{\hat U},t)\coloneqq 
        Ad_{\hat U_0\exp(-it\hat X)}.
        \label{DefRetr}
    \end{equation}

We now want to consider equivariant interacting time evolution operators over grotesque Nambu spaces with symmetries $(\mathscr W_N,b,G)$. Assuming equivariance with regard to the full subgroup $G_{UL}\cong U(1)$ would restrict terms like $a_ja_k+a^\dagger_ka^\dagger_j$ from appearing in the generating Hamiltonians. This precludes consideration of superconducting models, like in the Bogoliubov-de Gennes formalism. We will thus consider Hamiltonians equivariant under the subgroup $G'\subset G$ generated by non-ordinary symmetries that square to $\pm \hat 1$ and by the parity operator $\hat P$, so that $G'_{UL}\cong O(1)$.

\begin{definition}
    Let $(\mathscr W_N,b,G)$ be a grotesque Nambu space with symmetries with signature $(\epsilon_S\epsilon_C\epsilon_T)$. 
    The \textit{space of equivariant interacting Hamiltonians} is
    $$
        \mathscr H_{i,N}^{\epsilon_S\epsilon_C\epsilon_T}=\{\hat H\in\mathscr H_{i,N}\mid \forall \hat U\in G':\hat U\hat H\hat U^\dagger=\hat H\}.
    $$

    The \textit{space of equivariant interacting time evolution operators} is
    $$
        \mathscr M_{i,N}^{\epsilon_S\epsilon_C\epsilon_T}=\{Ad_{\exp(-i\hat H)}\in Aut(Cl(\mathscr W,b)) \mid \hat H\in \mathscr H_{i,N}^{\epsilon_S\epsilon_C\epsilon_T}\}.
    $$
    
    The \textit{space of equivariant weakly interacting time evolution operators} is
    $$
        \mathscr M_{wi,N}^{\epsilon_S\epsilon_C\epsilon_T} \coloneqq \mathscr M_{i,N}^{\epsilon_S\epsilon_C\epsilon_T}\cap \mathscr M_{wi,N}.
    $$

    Given convenient choices of basis we denote by $\mathscr M_{wi}^{\epsilon_S\epsilon_C\epsilon_T}$ the colimit as in the definition of $\mathscr M^{\epsilon_S\epsilon_C\epsilon_T}$.
\end{definition}
    
We are now ready to define weakly interacting versions of the topological K-theory spectra. The cut locus assumption allows us to define the structural suspension maps by projecting to the free subspectra $KU$ and $KO$, and then applying their suspension maps.

\begin{definition}
    The \textit{weakly interacting complex topological K-theory spectrum} $KU^{wi}$ is
    \begin{gather*}
    KU^{wi}_{\bullet}\coloneqq \begin{cases}
        \mathscr M_{wi}^{100},&\bullet=0\mod 2\\
        \mathscr M_{wi}^{000},&\bullet =1\mod 2
    \end{cases},\qquad
    \sigma^{wi}_\bullet(Ad_{\hat U},\theta)\coloneqq \sigma_\bullet(Ad_{\hat U_0},\theta).
    \end{gather*}
    
    The \textit{weakly interacting real topological K-theory spectrum}
    $KO^{wi}$ is 
    \begin{gather*}
    KO^{wi}_{\bullet}\coloneqq \begin{cases}
        \mathscr M_{wi}^{111},&\bullet=0\mod 8\\
        \mathscr M_{wi}^{001},&\bullet =1\mod 8\\
        \mathscr M_{wi}^{1-11},&\bullet=2\mod 8\\
        \mathscr M_{wi}^{0-10},&\bullet =3\mod 8\\
        \mathscr M_{wi}^{1-1-1},&\bullet=4\mod 8\\
        \mathscr M_{wi}^{00-1},&\bullet =5\mod 8\\
        \mathscr M_{wi}^{11-1},&\bullet=6\mod 8\\
        \mathscr M_{wi}^{010},&\bullet =7\mod 8
    \end{cases},\qquad
    \sigma^{wi}_\bullet(Ad_{\hat U},\theta)\coloneqq \sigma_\bullet(Ad_{\hat U_0},\theta).
    \end{gather*}
\end{definition}

\begin{lemma}\label{DefRetPresCartInv}
        Let $\tau$ be a Cartan involution of $\mathscr M_{i,N}$ such that $\tau(\mathscr M_N)=\mathscr M_N$. Then the deformation retract $h$ in (\ref{DefRetr}) leaves the subspaces $\text{Fix}_\tau$ and $\text{Inv}_\tau$ invariant.
    \end{lemma}

    \begin{proof}
        Suppose $Ad_{\hat U}\in \text{Fix}_\tau$. Since Cartan involutions preserve geodesics, $\tau(Ad_{\hat U_0})$ is the closest element of $\mathscr M_N$ to $Ad_{\hat U}$, thus $\tau(Ad_{\hat U_0})=Ad_{\hat U_0}$.
        Similarly if $Ad_{\hat U}\in \text{Inv}_\tau$ then $\tau(Ad_{\hat U_0})=Ad_{\hat U_0^\dagger}$. Since $h$ is defined in terms of geodesics this guarantees that $\text{Fix}_\tau$ and $\text{Inv}_\tau$ are invariant under $h$.
    \end{proof}

\begin{theorem}\label{TeoA}
    The spectra $KU^{wi}$ and $KO^{wi}$ deformation retract to $KU$ and $KO$. 
\end{theorem}

\begin{proof}
    By lemma \ref{DefRetPresCartInv} the deformation retract $h$ determines well defined homotopy retracts of each $\mathscr M^{\epsilon_S\epsilon_C\epsilon_T}_{wi}$ onto $\mathscr M^{\epsilon_S\epsilon_C\epsilon_T}$. Since
$$
h(\sigma^{wi}_\bullet(Ad_{\hat U},\theta),t)=h(\sigma_\bullet(Ad_{\hat U_0},\theta),t)=\sigma_\bullet(Ad_{\hat U_0},\theta)=\sigma_{\bullet}^{wi}(h(Ad_{\hat U},t),\theta),
$$
these homotopies are compatible with the structural suspension maps, thus determine spectra deformation retracts.    
\end{proof}

%%%%%%%%%%%%%%%%%
%%% S % S % S %%%
%%%%%%%%%%%%%%%%%

\section{Concluding remarks}\label{sec5}

\subsection{Twisted equivariant K-theory}

As mentioned in the introduction, the topological phases of crystalline fermion systems are classified by twisted equivariant K-theory \cite{freed2013twisted}. Every equivariant cohomology theory for a compact Lie group $P$ is represented by some $P$-spectrum $Y$, which is composed of a collection of pointed $P$-spaces $\{Y_V\}$ indexed by $P$-representations, equipped for each subrepresentation $V\hookrightarrow W$ with an equivariant structural suspension map $\sigma_{V,W}:Y_V\wedge \mathds S^{W-V}\to E_W$ \cite{LMS,may1996equivariant}. Freed and Moore state their arguments in \cite{freed2013twisted} can be adapted to generalize the classifying spaces used by Kitaev in \cite{kitaev2009periodic} to the equivariant context. This suggest our construction of $KU$ and $KO$ in terms of time evolution operators can be extended to $P$-spectra.  Given a full description of twisted  equivariant K-theory in terms of spectra of time evolution operators, we expect our arguments extend to a stable homotopy theoretical proof that the classification of crystalline topological insulators and superconductors is also stable to weak interactions.

\subsection{Classification of interacting fermion systems by cobordism}

The formulas for the structural suspension maps of $KU$ and $KO$ also apply for the spaces of equivariant interacting time evolution operators $\mathscr M^{\epsilon_S\epsilon_C\epsilon_T}_{i,N}$, which means they also form spectra $MU^i$ and $MO^i$. In \cite{freed2021reflection} Freed and Hopkins produce a general formula for symmetry protected phases in terms of Thom's bordism spectra, showing that in the interacting regime topological phases are classified by cobordism, as conjectured in \cite{Kapustin_2015}. This suggests the interacting spectra $MU^i$ and $MO^i$ might be weakly equivalent to Thom's spectra.
We suspect the filtered algebra structure of $Cl^+(\mathscr W_N,b)$ induces a spectra filtration interpolating between $KU$ and $MU^i$, and an analogous filtration between $KO$ and $MO^i$. In  \cite{cornfeld2021tenfold} the authors hypothesize this filtration might be related to the chromatic filtration interpolating between K-theory and cobordism. Our definitions in terms of spaces of time evolution operators may provide a framework to search for evidence of such relation.

%%%%%%%%%%%%%%%%%%%%%%%%%%%%%%%%%%%%%%%%%%%%%%%%%%%%%%%%%%%%%%%%%%%%%%%%%

%%%%%% BIBLIOGRAPHY
% \newpage
%\nocite{*}
\bibliographystyle{plain}
\bibliography{refs}

@book{may2006parametrized,
  title={Parametrized Homotopy Theory},
  author={May, J. Peter and Sigurdsson, Johann},
  series={Mathematical Surveys and Monographs},
  volume={132},
  year={2006},
  publisher={American Mathematical Society},
  address={Providence, RI},
  isbn={978-0-8218-3922-5}
}

@book{adams1974stable,
  title={Stable Homotopy and Generalised Homology},
  author={Adams, John Frank},
  series={Chicago Lectures in Mathematics},
  year={1974},
  publisher={University of Chicago Press},
  address={Chicago, IL},
  isbn={978-0-226-00524-9}
}

@book{MaEinftyRingSpcsSpectra,
  title     = {{$E_\infty$} ring spaces and {$E_\infty$} ring spectra},
  author    = {May, J. Peter},
  year      = {1977},
  publisher = {Springer-Verlag},
  address   = {Berlin},
  series    = {Lecture Notes in Mathematics},
  volume    = {577}
}

@book{LMS,
  author    = {Lewis, Jr., L. Gaunce and May, J. Peter and Steinberger, Mark},
  title     = {Equivariant Stable Homotopy Theory},
  series    = {Lecture Notes in Mathematics},
  volume    = {1213},
  publisher = {Springer-Verlag},
  address   = {Berlin},
  year      = {1986},
  note      = {With contributions by J. E. McClure}
}

@inproceedings{kitaev2009periodic,
  title={Periodic table for topological insulators and superconductors},
  author={Kitaev, Alexei},
  booktitle={AIP Conference Proceedings},
  volume={1134},
  pages={22--30},
  year={2009},
  organization={American Institute of Physics},
  doi={10.1063/1.3149495}
}

@article{agarwala2017tenfold,
  title={The tenfold way redux: Fermionic systems with {$N$}-body interactions},
  author={Agarwala, Adhip and Haldar, Arijit and Shenoy, Vijay B.},
  journal={Annals of Physics},
  volume={385},
  pages={469--511},
  year={2017},
  publisher={Elsevier},
  doi={10.1016/j.aop.2017.08.003}
}

@article{cornfeld2021tenfold,
  title={Tenfold topology of crystals: Unified classification of crystalline topological insulators and superconductors},
  author={Cornfeld, Eyal and Carmeli, Shachar},
  journal={Physical Review Research},
  volume={3},
  number={1},
  pages={013052},
  year={2021},
  publisher={American Physical Society},
  doi={10.1103/PhysRevResearch.3.013052}
}

@article{bott1959stable,
  title={The stable homotopy of the classical groups},
  author={Bott, Raoul},
  journal={Annals of Mathematics},
  volume={70},
  number={2},
  pages={313--337},
  year={1959},
  publisher={JSTOR}
}

@article{dyson1962threefold,
  title={The threefold way: {A}lgebraic structure of symmetry groups and ensembles in quantum mechanics},
  author={Dyson, Freeman J.},
  journal={Journal of Mathematical Physics},
  volume={3},
  number={6},
  pages={1199--1215},
  year={1962},
  publisher={American Institute of Physics},
  doi={10.1063/1.1703863}
}

@article{baez2020tenfold,
  title={The Tenfold Way},
  author={Baez, John C.},
  journal={arXiv preprint arXiv:2011.14234},
  year={2020}
}

@article{cartan1926classe,
  title={Sur une classe remarquable d'espaces de {R}iemann},
  author={Cartan, {\'E}lie},
  journal={Bulletin de la Soci{\'e}t{\'e} Math{\'e}matique de France},
  volume={54},
  pages={214--264},
  year={1926}
}

@article{kennedy2015bott,
  title={Bott--{K}itaev periodic table and the diagonal map},
  author={Kennedy, Ricardo and Zirnbauer, Martin R.},
  journal={Physica Scripta},
  volume={2015},
  number={T164},
  pages={014010},
  year={2015},
  publisher={IOP Publishing},
  doi={10.1088/0031-8949/2015/T164/014010}
}

@article{freed2021reflection,
  title={Reflection positivity and invertible topological phases},
  author={Freed, Daniel S. and Hopkins, Michael J.},
  journal={Geometry \& Topology},
  volume={25},
  number={3},
  pages={1165--1330},
  year={2021},
  publisher={Mathematical Sciences Publishers},
  doi={10.2140/gt.2021.25.1165}
}

@article{thom1972cut,
  title={Sur le cut-locus d'une vari{\'e}t{\'e} plong{\'e}e},
  author={Thom, Ren{\'e}},
  journal={Journal of Differential Geometry},
  volume={6},
  number={4},
  pages={577--586},
  year={1972},
  publisher={Lehigh University}
}

@book{atiyah2018k,
  title={K-theory},
  author={Atiyah, Michael F.},
  year={1967},
  publisher={W. A. Benjamin},
  address={New York},
  note={Reprinted by CRC Press in 2018}
}

@unpublished{schwede2012symmetric,
  author = {Schwede, Stefan},
  title  = {Symmetric Spectra},
  note   = {Book project, available at \url{https://www.math.uni-bonn.de/~schwede/SymSpec-v3.pdf}},
  year   = {2012}
}

@article{wick1952intrinsic,
  title={The intrinsic parity of elementary particles},
  author={Wick, Gian Carlo and Wightman, Arthur S. and Wigner, Eugene P.},
  journal={Physical Review},
  volume={88},
  number={1},
  pages={101--105},
  year={1952},
  publisher={American Physical Society},
  doi={10.1103/PhysRev.88.101}
}

@article{Kapustin_2015,
   title={Fermionic symmetry protected topological phases and cobordisms},
   author={Kapustin, Anton and Thorngren, Ryan and Turzillo, Alex and Wang, Zitao},
   journal={Journal of High Energy Physics},
   volume={2015},
   number={12},
   pages={1--21},
   year={2015},
   doi={10.1007/JHEP12(2015)052}
}

@book{may1996equivariant,
  title={Equivariant Homotopy and Cohomology Theory: Dedicated to the Memory of Robert J. Piacenza},
  author={May, J. Peter and Cole, Michael and others},
  series={CBMS Regional Conference Series in Mathematics},
  volume={91},
  year={1996},
  publisher={American Mathematical Society},
  address={Providence, RI}
}

@article{friis2016reasonable,
  title={Reasonable fermionic quantum information theories require relativity},
  author={Friis, Nicolai},
  journal={New Journal of Physics},
  volume={18},
  number={3},
  pages={033014},
  year={2016},
  publisher={IOP Publishing},
  doi={10.1088/1367-2630/18/3/033014}
}

@article{atiyah1968bott,
  title={Bott periodicity and the index of elliptic operators},
  author={Atiyah, Michael F.},
  journal={The Quarterly Journal of Mathematics},
  volume={19},
  number={1},
  pages={113--140},
  year={1968},
  publisher={Oxford University Press},
  doi={10.1093/qmath/19.1.113}
}

@article{vasconcellos2022operator,
  title={Operator {$K$}-theory algebra spectra of {$C^*$}-algebras},
  author={Vasconcellos, Renato V. and M{\"u}ssnich, L. C. P. A. M. and Aza, N. J. B.},
  journal={arXiv preprint arXiv:2203.03050},
  year={2022}
}

@phdthesis{prasad2022cut,
  title={Cut Locus of Submanifolds: A Geometric and Topological Viewpoint},
  author={Prasad, Sachchidanand},
  year={2022},
  school={Indian Institute of Science Education and Research Kolkata}
}

@article{freed2013twisted,
  title={Twisted Equivariant Matter},
  author={Freed, Daniel S. and Moore, Gregory W.},
  journal={Annales Henri Poincar{\'e}},
  volume={14},
  number={8},
  pages={1927--2023},
  year={2013},
  publisher={Springer},
  doi={10.1007/s00023-013-0236-x}
}

@article{kennedy2016bott,
  title={Bott periodicity for {$\mathbb{Z}_2$}-symmetric ground states of gapped free-fermion systems},
  author={Kennedy, Ricardo and Zirnbauer, Martin R.},
  journal={Communications in Mathematical Physics},
  volume={342},
  number={3},
  pages={909--963},
  year={2016},
  publisher={Springer},
  doi={10.1007/s00220-015-2535-x}
}

@article{cartan1927classe,
  title={Sur une classe remarquable d'espaces de {R}iemann. {II}},
  author={Cartan, {\'E}lie},
  journal={Bulletin de la Soci{\'e}t{\'e} Math{\'e}matique de France},
  volume={55},
  pages={114--134},
  year={1927}
}

@article{heinzner2005symmetry,
  title={Symmetry classes of disordered fermions},
  author={Heinzner, Peter and Huckleberry, Alan and Zirnbauer, Martin R.},
  journal={Communications in Mathematical Physics},
  volume={257},
  number={3},
  pages={725--771},
  year={2005},
  publisher={Springer},
  doi={10.1007/s00220-005-1330-9}
}

@book{bratteli-dois,
  title={Operator Algebras and Quantum Statistical Mechanics 2: Equilibrium States. Models in Quantum Statistical Mechanics},
  author={Bratteli, Ola and Robinson, Derek W.},
  series={Theoretical and Mathematical Physics},
  year={1997},
  edition={2nd},
  publisher={Springer-Verlag},
  address={Berlin},
  isbn={978-3-540-61481-4}
}

@book{book-derezinski,
  title={Mathematics of Quantization and Quantum Fields},
  author={Derezi{\'n}ski, Jan and G{\'e}rard, Christian},
  series={Cambridge Monographs on Mathematical Physics},
  year={2013},
  publisher={Cambridge University Press},
  address={Cambridge},
  doi={10.1017/CBO9780511843518}
}

@article{araki-AMS,
  title={Bogoliubov automorphisms and {F}ock representations of canonical anticommutation relations},
  author={Araki, Huzihiro},
  journal={Contemporary Mathematics},
  volume={62},
  pages={23--141},
  year={1987},
  publisher={American Mathematical Society}
}

@book{book-bru-pedra,
  title={C*-Algebras and Mathematical Foundations of Quantum Statistical Mechanics},
  author={Bru, Jean-Bernard and de Siqueira Pedra, Walter Alberto},
  year={2023},
  series={Lecture Notes in Mathematics},
  volume={2314},
  address={Cham},
  publisher={Springer Nature},
  doi={978-3-031-28949-1}
}

%%%%%%%%%%%%%%%%%%%%%%%%%%%%%%%%%%%%%%%%%%%%%%%%%%%%%%%%%%%%%%%%%%%%%%%%%

%%%%%% BEGIN DOCUMENT

\end{document}